\documentclass[11pt]{article}

\usepackage[a4paper,margin=1in]{geometry}

\usepackage[T1]{fontenc}
\usepackage{lmodern}
\usepackage{microtype}

\usepackage[tbtags]{amsmath}
\usepackage{amssymb}
\usepackage{amsthm}
\usepackage{mathtools}
\usepackage{bm}

\allowdisplaybreaks

\usepackage[svgnames]{xcolor}
\usepackage{graphicx}
\usepackage{tikz}
\usetikzlibrary{arrows.meta,positioning,quotes}
\usepackage{subcaption}
\usepackage{booktabs}
\usepackage{pgfplots}
\pgfplotsset{compat=1.18}
\tikzset{myarrow/.style={-{Triangle[scale=0.8]}}}

\usepackage[hyphens]{url}
\usepackage[
backend=biber,
style=alphabetic,
maxbibnames=999,
maxalphanames=4,
backref=true
]{biblatex}
\usepackage{hyperref}
\hypersetup{
colorlinks=true,
linkcolor=DarkBlue,
citecolor=DarkGreen,
urlcolor=DarkBlue
}
\usepackage[capitalise,nameinlink]{cleveref}
\AddToHook{cmd/appendix/before}{\crefalias{section}{appendix}}

\newcommand{\R}{\mathbb{R}}
\newcommand{\N}{\mathbb{N}}
\newcommand{\sset}[1]{\left\{#1\right\}}
\newcommand{\ssets}[1]{\{#1\}}

\newcommand{\card}[1]{\left|#1\right|}

\DeclareMathOperator*{\argmax}{arg\,max}

\DeclarePairedDelimiter{\norm}{\lVert}{\rVert}
\newcommand{\maxnorm}[1]{\norm{#1}_{\infty}}
\newcommand{\union}{\cup}
\newcommand{\map}{\longrightarrow}

\newcommand{\then}{\Longrightarrow}
\newcommand{\ifif}{	\Longleftrightarrow}
\newcommand{\minshort}{\wedge}
\newcommand{\pospart}[1]{\left(#1\right)_{+}}

\newcommand{\vecc}{\bm}
\DeclareMathOperator*{\expectation}{\mathbb{E}}
\newcommand{\expect}[2][]{\expectation\nolimits_{#1}\left[#2\right]}

\DeclareMathOperator*{\probability}{\mathrm{Prob}}
\newcommand{\prob}[2][]{\probability\nolimits_{#1}\left[#2\right]}

\renewcommand{\vec}{\bm}

\newcommand{\rev}{\ensuremath{\mathsf{REV}}}
\newcommand{\srev}{\ensuremath{\mathsf{SREV}}}
\newcommand{\brev}{\ensuremath{\mathsf{BREV}}}
\newcommand{\mrev}{\ensuremath{\max\{\srev,\brev\}}}

\newcommand{\srevshort}{\ensuremath{\mathsf{S}}}
\newcommand{\brevshort}{\ensuremath{\mathsf{B}}}
\newcommand{\mrevshort}{\ensuremath{\mathsf{M}}}
\newcommand{\support}{\operatorname{supp}}
\newcommand{\Kzero}{\ensuremath{[K]_*}}
\newcommand{\er}{\ensuremath{\mathsf{ER}}}
\newcommand{\tail}[1]{\bar{#1}}

\DeclareMathOperator*{\indicator}{\mathbf{1}}

\AddToHook{env/lemma/begin}{\crefalias{theorem}{lemma}}
\AddToHook{env/corollary/begin}{\crefalias{theorem}{corollary}}
\AddToHook{env/proposition/begin}{\crefalias{theorem}{proposition}}
\AddToHook{env/definition/begin}{\crefalias{theorem}{definition}}
\AddToHook{env/remark/begin}{\crefalias{theorem}{remark}}
\AddToHook{env/claim/begin}{\crefalias{theorem}{claim}}
\newtheorem{theorem}{Theorem}[section]
\newtheorem{lemma}[theorem]{Lemma}
\newtheorem{proposition}[theorem]{Proposition}

\newtheorem{claim}[theorem]{Claim}
\crefname{claim}{Claim}{Claims}
\Crefname{claim}{Claim}{Claims}
\theoremstyle{definition}
\newtheorem{definition}[theorem]{Definition}
\theoremstyle{remark}

\title{On the Power of Determinism in Multi-Item Auctions\thanks{This work was
supported by the German Research Foundation (DFG) within subproject B07 of the
Sonderforschungsbereich/Transregio 154 ``Mathematical Modelling, Simulation and
Optimization using the Example of Gas Networks''.}}

\author{
Yiannis Giannakopoulos\thanks{University of Glasgow.
Email: \href{mailto:yiannis.giannakopoulos@glasgow.ac.uk}{\nolinkurl{yiannis.giannakopoulos@glasgow.ac.uk}}}
\and
Johannes Hahn\thanks{University of Technology, Nuremberg.
Email: \href{mailto:johannes.hahn@utn.de}{\nolinkurl{johannes.hahn@utn.de}}
}}

\hypersetup{
pdftitle={On the Power of Determinism in Multi-Item Auctions},
pdfauthor={Yiannis Giannakopoulos and Johannes Hahn}
}

\date{November 3, 2026}

\begin{document}
\pagenumbering{roman}
\maketitle

\begin{abstract}
    We study the classical multi-item monopoly setting with a single additive buyer
    and $m$ heterogeneous items whose values are independent but not necessarily
    identically distributed. We compare the revenue ($\rev$) of the optimal truthful
    auction, which in general may be randomized or very complicated, with the
    revenues of three \emph{simple deterministic} auctions: selling all items
    separately ($\srev$), selling them as a single grand bundle ($\brev$), and
    choosing the better of the two.
    
    Our technical cornerstone is a general upper bound on the optimal
    multidimensional revenue, parameterized by freely chosen item-wise caps, and
    expressed in terms of optimal single-item, Myersonian revenues and basic
    statistics of the value distributions. We derive this inequality through
    Lagrangian duality for a nonlinear programming formulation of the worst-case
    approximation ratio of separate-selling. Combining this inequality with
    carefully chosen caps and a probabilistic analysis of the resulting terms, we
    obtain improved approximation guarantees for all three deterministic auctions,
    both for general independent values and for iid values. Notably, we prove that
$\rev \leq 2.5 \cdot \max\{\srev,\brev\}$, improving the long-standing $5.2$
    factor of Ma and Simchi-Levi [arXiv 2015; AISTATS'21] and the $6$ factor of
    Babaioff, Immorlica, Lucier, and Weinberg [FOCS'14; JACM 2020]. We complement
    this upper bound with a construction establishing a lower bound of $2.19$,
    improving the bound of $2$ by Rubinstein [ITCS'16]. Together, these bounds
    further narrow the gap for the approximation ratio of the better of
    separate-selling and grand-bundling down to $[2.19,2.5]$.
    
    For two iid items, we determine the \emph{exact} approximation ratio of selling
    separately. Our proof constructs tight Lagrangian dual certificates for the
    aforementioned programming formulation, when item values lie in the grid
$\{0,1/K,2/K,\dots ,1\}$, for any discretization parameter $K$. Taking
$K\to\infty$ yields the tight ratio $1+W(1/e)\approx 1.2785$ in the
    continuous-valued setting, where $W$ is the Lambert-W function, closing the
$[1.278,1.368]$ gap left by Hart and Nisan [EC'12; JET 2017].
\end{abstract}

\newpage
\tableofcontents
\newpage

\pagenumbering{arabic}
\setcounter{page}{1}
\section{Introduction}

The design of optimal \emph{multi}dimensional auctions is a notoriously
challenging problem within economics and computer
science~\parencite{McAfeeMcMillan1988,RochetChone1998,ManelliVincent2007,Daskalakis:2017aa,Cai2019}.
Although the single-item case has been well understood since the seminal works
of~\textcite{Vickrey1961a,Myerson1981a}, the characterization of the maximum
revenue an auctioneer can achieve when selling multiple items, even to a single
bidder, remains generally elusive. Nevertheless, the problem has attracted a
significant amount of attention from the community in the last 15 years,
resulting in a better understanding of the underlying mathematical structures
and intricacies that make the generalization of the sharp and elegant solution
of \textcite{Myerson1981a} so hard to generalize in the multidimensional
setting, and the corresponding revenue optimization task so hard to penetrate
rigorously.

Some of the most notable and clear-cut obstacles include the non-monotonicity of
the revenue objective \parencite{Hart2015} and the fact that, in general, the
optimal auction might be randomized, even with an uncountable-sized description
(see, e.g., \parencite{Daskalakis:2017aa}) and computationally intractable to
find~\parencite{Daskalakis2013,Chen2015a,Chen:2018aa}. As a result, there has
been a significant interest, especially from the algorithmic game theory
community, on understanding the performance, and limitations, of simple and
natural selling mechanisms; these mechanisms will, in general, extract
suboptimal revenue, however they are significantly easier to describe and
implement. 

A pivotal moment in this direction was the work of~\textcite{Babaioff2020} who
studied the mechanism that chooses the better of selling the items separately or
bundling them all together, and (inspired by the core-tail analysis introduced
by~\textcite{Li2013a}) proved that this achieves a $6$-factor approximation to
the optimal revenue achievable by any, no matter how complicated, mechanism. It
is interesting to note that this was the first time that determinism was shown
to achieve a constant approximation ratio. Although various improved
approximation guarantees are known for special cases, including two-item
settings or identical item-value priors and, furthermore, bounds can be derived
for each of selling separately or full-bundling, as separate mechanisms, in
general the existing bounds in the literature are not tight, sometimes with
large gaps; a detailed presentation of related literature can be found
in~\cref{sec:related-work} below.

\subsection{Related Work}
\label{sec:related-work}

\paragraph{Multidimensional optimal auctions}
Myerson's~\parencite{Myerson1981a} elegant virtual-value characterization of the
revenue objective and his single-parameter truthfulness lemma provide a sharp
and satisfying description of optimal auctions for a single item, as a simple
deterministic mechanism which, for the special case of a single buyer further
boils down to take-it-or-leave-it pricing (for historical completeness, see also
the works of~\textcite{riley1981optimal,Riley1983a}).
However, this clear characterization notoriously fails to extend to the
multi-item case, where truthfulness induces a much more complex feasibility
space~\parencite{Rochet1987,ManelliVincent2007} to optimize over. Nevertheless,
the economics literature produced an influential line of works trying to
illuminate this landscape~(see, e.g.,
\parencite{McAfeeMcMillan1988,RochetChone1998,Armstrong1996}), including the
investigation of the role of randomized mechanisms (lotteries) and the
limitations of determinism~\parencite{Thanassoulis2004a,Pycia2006a,Hart2015},
and a particular emphasis on the understanding of the nature and possible
optimality of item-bundling
\parencite{McAfeeMcMillanWhinston1989,ManelliVincent2006,Menicucci2015,Daskalakis:2017aa,Haghpanah2020}.

Theoretical computer science continued and complemented this structural program,
further introducing new perspectives, including algorithmic reductions and
hardness results \parencite{Cai2012b,Cai2013a,Daskalakis2013}, as well as
duality frameworks for characterizing or upper-bounding optimal revenue
\parencite{Daskalakis:2017aa,Kleiner:2019aa,gk2014,Cai2019}. Our paper can be
seen as part of this line of work, since its technical backbone is a
mathematical programming formulation, paired with tight and non-tight Lagrangian
dual solutions. The notorious difficulty of the multi-item revenue maximization
problem, even for a single bidder, naturally motivates the study of simpler
deterministic benchmarks, corresponding to selling mechanisms that are easier to
describe and implement.

\paragraph{Simple mechanisms}
\textcite{Hart:2017aa} initiated the systematic comparison of the unrestricted,
possibly randomized, optimal revenue (\rev) with the revenues achievable by
selling each item separately (\srev) and selling them all together, in a single
grand bundle (\brev). For arbitrary independent item values, \textcite{Li2013a}
introduced the core--tail decomposition and obtained logarithmic guarantees for
the best deterministic mechanism; building on this approach,
\textcite{Babaioff2020} proved the first constant guarantee for determinism:
$\rev\leq 6 \cdot \max\{\srev,\brev\}$. This factor was subsequently improved to
$5.2$ by \textcite{MaSimchiLevi2021}. \textcite{CDGM2019} specifically
highlight, from a technical perspective, the idea of upper-bounding the optimal
revenue by a parameterized linear combination of $\srev$ and $\brev$; although,
at the time, their particular bound (see~\parencite[Theorem~7]{CDGM2019}) was
not able to beat the state-of-the-art $5.2$ bound~\parencite{MaSimchiLevi2021},
it can be seen as conceptually laying the ground for a more careful
consideration, and optimization, of this combination technique which, to some
extent, underlies our approach in this paper as well (see,
e.g.~\eqref{eq:M-upper-bound-independent}). 
On the other hand, \textcite{Rubinstein2016} constructed independent-item
instances on which a partition mechanism extracts almost twice the revenue that
$\max\{\srev,\brev\}$ does, providing thus a lower bound of $2$ on the
approximation ratio for this joint benchmark. In this paper we provide a
new\footnote{See also our discussion at
p.~\pageref{comparison-concurrent-paragraph} below, for a more detailed
exposition of the timeline of our bounds regarding the approximation ratio of
$\mrev$, especially relating to the concurrent work of
\textcite{cai2026improvedrevenueguaranteesselling} and
\textcite{huang2026improvedapproximation}, who also study the same question.}
upper bound of $2.5$ for this approximation ratio, improving the
long-standing\footnote{An arXiv version of the manuscript
of~\textcite{MaSimchiLevi2021} first appeared online in~2015.} $5.2$ upper bound
of~\textcite{MaSimchiLevi2021}, as well as a new lower bound of $2.19$, thus
narrowing significantly the remaining gap to $[2.19,2.5]$. Furthermore, for the
special case of iid items, we can provide an even better upper bound of $2.15$
on the approximation ratio of $\mrev$ (\cref{th:M-upper-bound-iid}); note that,
for this iid case, the currently best known lower-bound is $1.1877$, which can
already be realized with only two items, as shown
by~\textcite[Theorem~3.4]{MaSimchiLevi2021}.

\paragraph{Two-item auctions}
The single-buyer, two-item setting is the smallest genuinely multidimensional
revenue-maximization problem and has long served as a central test-bed for
understanding optimal auctions. Even in this setting, optimal mechanisms may
require randomization and, for some distributions, (uncountably) infinite menu
sizes~\parencite{ManelliVincent2007,Daskalakis:2017aa}. A substantial line of
work gives exact solutions or structural characterizations for important
distributional families or restricted mechanism classes
\parencite{Pavlov2011,gk2015,Daskalakis:2017aa,WangTang2017,Thirumulanathan2019,BabaioffNisanRubinstein2018}.
These results underscore that the two-item case is already a substantive
frontier, but they do not determine the worst-case loss of a simple benchmark.
In that direction, \textcite{Hart:2017aa} bounded $\rev/\srev$ between
$1+W(1/e)\approx1.278$ and $1+1/e\approx1.368$ for the case of two iid items,
leaving its exact value open; for two independent, not necessarily identically
distributed items, \textcite{HartReny2019} subsequently proved that
$\rev/\srev\leq 1+1/\sqrt{e}\approx1.607$, improving to $1+1/e\approx1.368$
under regularity. A main result of our paper is the exact characterization of
this ratio for any two-item iid auction with values over the unit $1/K$-grid.
This requires a novel tight Lagrangian dual construction, which is a technical
highlight of our paper. Taking the limit with respect to $K$ we resolve the
Hart--Nisan gap of $[1.278,1.368]$ down to the $1.278$ lower bound.

\paragraph{Selling separately}
For a general number of $m\geq 2$ independent items, \textcite{Hart:2017aa}
first showed an upper bound of $O(\log^2 m)$ on the approximation ratio of
$\srev$, which was later improved to $O(\log m)$ by~\textcite{Li2013a} and to
$\ln m+3$ by~\textcite{Babaioff2020}. The logarithmic dependence here is
necessary, due to a lower bound by~\parencite[Remark after Proposition~25, and
Corollary~26]{Hart:2017aa}. In our paper, we improve the additive terms in these
upper-bound guarantees, both for arbitrary independent and iid items, by
providing a new harmonic-number upper bound of $H_m$ for the approximation ratio
of selling separately (see~\cref{th:S-upper-bound-independent}).

\paragraph{Grand bundling}\label{page:grand-bundling-related-work} In general,
for arbitrary independent items, selling all the items in a single grand bundle
cannot alone achieve a constant-factor approximation to the optimal revenue:
\textcite[Example~27]{Hart:2017aa} provide an instance where $\rev/\brev$ gets
arbitrarily close to $m$. However, \textcite[Lemma~28]{Hart:2017aa} also provide
an upper bound of $O(m)$, an improvement on which, to the best of our knowledge,
has not been since published. Our new upper bound result of $m+1-1/m$
(\cref{prop:B-upper-bound-independent}) matches the existing lower bound up to
an additive constant.

The situation for iid items is radically different though: grand bundling is
guaranteed to achieve a constant approximation ratio, as shown
by~\textcite{Li2013a}. Furthermore, in that case, bundling is also guaranteed to
be within a (multiplicative) constant away from selling separately: in
particular, \textcite{Kupfer2016} establishes that $\srev \leq \kappa \cdot
\brev$ for iid items, where\footnote{See~\parencite[Notation~1]{Kupfer2016} for
the precise definition of parameter $\kappa$.} $\kappa\approx 1.78669$,
improving the previous $4$ factor given by~\textcite{Hart:2017aa}. Although not
explicitly published in the literature of the problem, a combination of
\citeauthor{Kupfer2016}'s $\kappa$-bound with the parameterized optimal-revenue
bound of \textcite{CDGM2019} can
produce\footnote{\label{foot:brev-prior-improved}Combining the bounds of $\srev
\leq \kappa\, \brev $ of \parencite[Theorem~1 and Notation~1]{Kupfer2016} and
$\rev \leq (2+2/a^2)\brev + (a+1)\srev$ by  \parencite[Theorem~7]{CDGM2019}, we
get that $\rev \leq \left[2(1+1/a^2)+\kappa(a+1)\right]\brev$. Optimizing this
expression with respect to $a>0$, for $a=(4/\kappa)^{1/3}$ we get the bound of
$2+3\cdot 2^{-1/3}\kappa^{2/3}+\kappa\approx 7.29268$ on the approximation ratio
of $\rev/\brev$.} a concrete bound of $\rev \leq 7.2927\cdot \brev$ on the
approximation ratio of grand-bundling, improving the general $O(1)$ bound from
\parencite{Li2013a} mentioned above. In this paper we improve this upper bound
to $2.5$ (\cref{th:B-upper-bound-iid}).

\subsection{Our Results and Techniques}
\label{sec:results-techniques}
In~\cref{sec:model-notation} we formally introduce our auction model, which is
general enough to incorporate arbitrary joint (multi-item) value distributions.
Although most of our results later on deal with product distributions, we
deliberately choose that level of generality here, so that we can still use it
when relevant.\footnote{For example, that is the case for our
discrete-continuous revenue-approximation bounds
in~\cref{prop:discretization-approximation-revenues} which do not rely on any
independence assumptions.} In~\cref{sec:simple-auctions} we define the revenue
benchmarks of our simple deterministic auctions, namely those of selling
separately (\srev) and selling in a full-bundle (\brev); the goal of our paper
is to study how well these quantities approximate the best possible revenue
(\rev) of a given auction setting (see~\eqref{eq:opt-rev-def}) which, as
discussed in our introduction, could in general be randomized and very complex
to characterize.

Next, in~\cref{sec:discrete-auction-prelims} we formally describe our
discrete-auction setting, where we assume that the buyer's values for the items
lie in a $1/K$-fine uniformly discretized grid of the unit interval $[0,1]$,
where $K$ is a positive integer parameter, controlling how dense this
discretization is. This model will be the basis of all our results in this
paper; generalizing our results to distributions over the continuous interval
$[0,1]$ is essentially a straightforward limiting process, as $K\to\infty$.
Although it is known in the community (see, e.g.,
\parencite{hart2026rootrevenuecontinuity} and \parencite[Theorem~1]{Cai2019})
that such a convergence is in principle well-defined for the optimal $\rev$
objective, we provide a unified, rigorous statement and proof of this limiting
process in~\cref{prop:discretization-approximation-revenues}, with respect to
all our revenue objectives of interest, namely $\rev$, $\srev$, $\brev$, and
$\max\{\srev,\brev\}$, including bounds on the rate of convergence, as this
might be a useful reference for follow-up work.

In \cref{sec:srev-programming-formulatio} we present the foundation of all our
main results of the paper, which is the formulation of the adversarial problem
of the worst-case instance for the mechanism of selling all items separately,
i.e.\ maximizing the approximation ratio $\frac{\rev}{\srev}$, as a (highly)
nonlinear mathematical optimization problem~\eqref{mathematical_program}. We
then take the Lagrangian of this program \eqref{eq:Lagrangian}, and provide two
different dual solutions: 

\paragraph{The master inequality}
First, in~\cref{sec:lagrangian-dual-many-items-weak} we start with a
structurally simpler, and therefore easier to analyse, Lagrangian solution for
any number of items and arbitrary independent value distributions
(\cref{sec:lagrangian-dual-many-items-weak}). This is, in general, a \emph{weak}
dual. However, after a careful analysis (see the proof
of~\cref{lemma:rev_upper_bound_master_caps}), it still gives rise to a powerful,
and rather elegant, upper-bound formula~\eqref{eq:rev_master_inequality} on the
optimal revenue $\rev$; we will often be referring to this bound as our
\emph{master inequality}, since it is the starting point, and technical
cornerstone, for most of the results in this paper, and in particular all our
improved upper-bound guarantees for the approximation-ratios of simple
deterministic mechanism for many (independent) items
(see~\cref{tab:results-comparison}).

\providecommand{\TableOneResultCell}[3]{%
\begin{tabular}[c]{@{}c@{}}
    $\boldsymbol{#1}$\\[-4pt]
    {\scriptsize(\cref{#2})}\\[3pt]
    {\footnotesize #3}%
\end{tabular}}
\begingroup
\renewcommand{\thefootnote}{\ensuremath{\dagger}}
\begin{table}[t]
    \centering
    \caption{Upper bounds on the approximation ratios of the three deterministic
    revenue benchmarks we study in our paper, for a single additive buyer and
    $m\geq2$ items. Our results are shown in \textbf{bold}, with references to
    their individual technical propositions from \cref{sec:independent-results}
    (general independent items) and \cref{sec:iid-results} (iid items). The
    entries below our results state the previously best published bounds.
    Comparisons with the concurrent works of
    \parencite{cai2026improvedrevenueguaranteesselling,huang2026improvedapproximation}
    are discussed in detail at the end of \cref{sec:results-techniques}
    (see p.~\pageref{comparison-concurrent-paragraph}). Our results are also
    summarized in the introductory
    \cref{th:upper-bounds-independent-unified,th:upper-bounds-iid-unified} of
    our paper, but in somewhat simplified form compared to our exact bounds
    below. Here $H_m=\sum_{j=1}^m\frac{1}{j}$ and $H_{1/m}\coloneqq \int_0^1
    \frac{1-t^{1/m}}{1-t}\,\mathrm{d}t \geq \frac{1}{m}$.}
    \label{tab:results-comparison}
    \begingroup
    \small
    \setlength{\tabcolsep}{9pt}
    \renewcommand{\arraystretch}{1.12}
    \begin{tabular}{@{}lccc@{}}
        \toprule
        & \multicolumn{3}{c}{Benchmark ($\rev/ \;\cdot\;$)} \\
        \cmidrule(l){2-4}
        Item values & $\max\{\srev,\brev\}$ & $\srev$ & $\brev$ \\
        \midrule
        Independent
        & \TableOneResultCell{2.5}{th:M-upper-bound-independent}%
        {$5.2$ \parencite{MaSimchiLevi2021}}
        & \TableOneResultCell{H_m}{th:S-upper-bound-independent}%
        {$\ln m+3$ \parencite{Babaioff2020}}
        & \TableOneResultCell{m+1-H_{1/m}}{prop:B-upper-bound-independent}%
        {$O(m)$\footnotemark[0] \parencite{Hart:2017aa}} \\
        \addlinespace[8pt]
        IID
        & \TableOneResultCell{2.15}{th:M-upper-bound-iid}%
        {$5.2$ \parencite{MaSimchiLevi2021}}
        & \TableOneResultCell{H_{m}}{th:S-upper-bound-independent}%
        {$\ln m+3$ \parencite{Babaioff2020}}
        & \TableOneResultCell{2.5}{th:B-upper-bound-iid}%
        {$O(1)$\footnotemark[0] \parencite{Li2013a}} \\
        \bottomrule
    \end{tabular}
    \endgroup
\end{table}
\footnotetext[0]{Although not published, stronger bundling guarantees than the
ones appearing in this table follow by combining other prior work; see our
discussion in~\cref{foot:brev-prior-improved}.}
\endgroup

In terms of its ``nature'', our master
inequality~\eqref{eq:rev_master_inequality} upper-bounds the optimal,
multidimensional revenue $\rev$, by a quantity involving the single-dimensional,
Myersonian revenues of each item, as well as statistics of \emph{capped}
versions of these separate-item values, for any choice of caps of our liking.
Then, the technical strength of this master inequality is demonstrated, by
instantiating it for different choices of caps, each tailored to each
application at hand, and then using purely probabilistic arguments to derive our
desired approximation ratio bounds. With all caps set to zero,
\eqref{eq:rev_master_inequality} reduces to the uncapped revenue bound that
follows from the single-buyer decomposition of~\textcite[Section~6]{Cai2019} and
is stated explicitly
in~\parencite[Lemma~4.1]{cai2026improvedrevenueguaranteesselling}. Our
cap-dependent refinement follows from a different dual construction and provides
the additional flexibility used in our many-item bounds. We believe it is
interesting to also note that the capped/excess decomposition of our
inequality~\eqref{eq:rev_master_inequality} can conceptually provide a
core--tail-flavoured interpretation, in the spirit of
\textcite{Li2013a,Babaioff2020}. The quantitative refinement lies in the joint
bound: each item's monopoly revenue and capped expectation enter through a
maximum, and the residual excludes the largest excess (see also our discussion
in p.~\pageref{paragraph:core-tail-interpretation}).

To derive the approximation guarantees from our master inequality, we express
its residual expectation as an integral of \emph{tail} probabilities;
see~\eqref{eq:expectation-of-function-of-tail}. A key step is to upper-bound
this integrand by monotone, convex functions of the aggregate tail (see, e.g.,
\cref{lemma:tail_inequalities}). For a common cap $c$, a change of variables
turns~\eqref{eq:expectation-of-function-of-tail} into an integral over the
original item-value tails starting at $c$~\eqref{eq:common-caps}. Single-item
and grand-bundling revenues impose pointwise constraints on these tails. For the
constant-factor guarantees, we also control exponentially weighted tail
integrals using the grand-bundling revenue
(\cref{lemma:bound_on_weighted_tails}). Combining these constraints with cap
choices tailored to each revenue benchmark yields our approximation guarantees.

\medskip
Our upper bounds for $m\geq 2$ items, both for the general independent and the
iid case, are summarized in the following two
\cref{th:upper-bounds-independent-unified,th:upper-bounds-iid-unified},
respectively. The bounds in their statements are deliberately simplified for
some of them, in the interest of readability and unification; the exact bounds
can be found in \cref{tab:results-comparison} and, of course, in the
corresponding \cref{sec:independent-results,sec:iid-results}, respectively,
where we break them down in separate technical propositions, and provide their
rigorous proofs.

\begin{theorem}[Approximation Ratios for Independent Items]
    \label{th:upper-bounds-independent-unified}
    The optimal revenue of any single-bidder auction with $m$ items
    independently drawn from distributions $F_1,F_2,\dots,F_m$ is at most:
    \begin{equation*}
        \label{eq:upper-bound-independent-unified}
        \rev(\vec{F}) \leq \min\sset{ 2.5 \cdot \max \sset{\srevshort,\brevshort},H_m \cdot \srevshort, \left(m+1\right) \cdot \brevshort},
    \end{equation*}
    where $\vec{F}=F_1\times F_2\times\dots\times F_m$ and
    $\srevshort=\srev(\vec{F})$, $\brevshort=\brev(\vec{F})$ denote the
    (optimal) revenue from selling the items separately or in a grand-bundle,
    respectively.
\end{theorem}
For the $\mrev$ objective, we also provide a new lower bound of $2.19$
in~\cref{th:M-lower-bound}, complementing the upper bound of $2.5$ provided by
the theorem above.

\begin{theorem}[Approximation Ratios for IID Items]
    \label{th:upper-bounds-iid-unified}
    The optimal revenue of any single-bidder auction with $m$ iid items from
    distribution $F$ is at most:
    \begin{equation*}
        \label{eq:upper-bound-iid-unified}
        \rev(F^m) \leq \min\sset{ 2.15 \cdot \max \sset{\srevshort,\brevshort},2.5 \cdot \brevshort, H_{m}\cdot \srevshort},
    \end{equation*}
    where $\srevshort=\srev(F^m)$, $\brevshort=\brev(F^m)$ denote the (optimal)
    revenue from selling the items separately or in a grand-bundle,
    respectively.
\end{theorem}

\paragraph{Two iid items}
Secondly, in~\cref{sec:selling_separately_2_iid}, we provide a more specialized,
but significantly more technically involved, dual-instance construction for our
Lagrangian~\eqref{eq:Lagrangian}, for the special case of two, identically
distributed items. We start (\cref{lemma:two_iid_anchored_upper_bound}) by
showing that this construction provides an upper bound for a specific class of
distributions, which we call \emph{anchored}
(see~\cref{def:distribution_achored_in_r}). Then, we prove that this dual is
\emph{tight} by providing a lower bound (\cref{lemma:equal_revenue_tightness}),
achieved by particular ``equal-revenue-styled'' anchored distribution. Finally,
we argue that the restriction to anchored distribution is without loss, and our
tight analysis applies to arbitrary distributions
(see~\cref{prop:two_items_iid_SREV}). As a result we derive an exact,
closed-form characterization (see~\eqref{eq:two_items_iid_SREV} and
\eqref{def:discrete_lambert_approximate}) of the approximation ratio of selling
separately two iid items, point-wise for \emph{any} given discretization
parameter $K$. As an immediate by-product of our analysis by taking the limit as
$K\to\infty$, and by deploying our discrete-continuous proximity bounds
from~\cref{prop:discretization-approximation-revenues}, we fully resolve the
existing open gap on the value of $\frac{\rev}{\srev}$ for two iid items in the
standard, continuous $[0,1]$ interval setting, showing that it is equal to
$1+W(e^{-1})\approx 1.278$ (\cref{lemma:omega_K_properties}), where $W$ is the
Lambert-W function (see, e.g., \parencite{Corless1996}). A summary of these
results is provided in \cref{th:srev-2-iid-tight-master} below; its full
generality, and its technical proof, are provided via various separate lemmas
and propositions in~\cref{sec:selling_separately_2_iid} of our main text.

\begin{theorem}[Tight Approximation Ratio for Two IID Items]
    \label{th:srev-2-iid-tight-master}
    Selling two iid items separately to a single bidder guarantees a
    $1+W(1/e)\approx 1.2785$-approximation to the optimal revenue, and this is
    tight. That is, for any distribution $F$ it holds that
    \begin{equation*}
        \label{eq:srev-2-iid-tight-master}
        \rev(F\times F) \leq \left(1+\omega\right)\srev(F\times F),
    \end{equation*}
    where $\omega\coloneqq W(1/e)$ and $W$ denotes the Lambert-W function; furthermore, for any $\varepsilon>0$ there exists a distribution $F$ such that 
    \[ \rev(F\times F) \geq \left(1+\omega-\varepsilon\right)\srev(F\times F). \]
\end{theorem}

\paragraph{Concurrent work and timeline of our
results}\label{comparison-concurrent-paragraph} Since the preparation of the
first version of our paper~\parencite{gh2026_arxiv_v1}, until the present
version of this manuscript, we have become aware of two concurrent works,
namely~\parencite{cai2026improvedrevenueguaranteesselling,huang2026improvedapproximation},
which include results overlapping with some of our own results here. To the best
of our knowledge, for all such overlapping results our present bounds provide a
strict improvement to those presented in those concurrent works; however, for
scientific completeness and to highlight some interesting connections, we
provide below an overview of the key points of comparison.

Chronologically, the first concurrent work was that from a group of Google
researchers \parencite{cai2026improvedrevenueguaranteesselling}, who used the
Cogentic~\parencite{cai2026cogenticmultiagentorchestrationautomated} agentic
discovery framework to autonomously derive a $3.52$-approximation for the better
of separate selling and grand bundling. Their proof combines an ``uncapped''
revenue bound~\eqref{eq:revenue-upper-bound-uncapped-CDW} (which follows from
the prior work of~\parencite{Cai2019}; see our discussion
in~\cref{foot:CDW-uncapped-revenue-discussion}) with a reduction to truncated
welfare and estimates of its second moment. This result
of~\parencite{cai2026improvedrevenueguaranteesselling} appeared in arXiv when we
were finalizing the first version of our paper~\parencite{gh2026_arxiv_v1}, in
which we already had the slightly better bound of $3.5$ for \mrev; however,
their approach of using the tail probabilities of the item values as the key,
unifying mathematical quantity for their technical analysis, had an important
influence, as a starting point, to many of our technical proofs as well. For
example, the capped version of~\eqref{eq:expectation-of-function-of-tail}
and~\cref{lemma:tail_inequalities} can be both seen as a generalization of the
starting technical step
in~\parencite[Lemma~4.2]{cai2026improvedrevenueguaranteesselling}.

The other concurrent work is \parencite{huang2026improvedapproximation}, of
which we became aware as we were finalizing the present version of our
manuscript. \textcite{huang2026improvedapproximation} proves that
$\rev\leq\srev+\frac{\pi^2}{6}\brev$, which yields an approximation ratio of
$(1+\pi^2/6)\approx 2.645$ for \mrev\ and, via a black-box use of
\citeauthor{Kupfer2016}'s bound~\parencite{Kupfer2016} for
$\frac{\srev}{\brev}$, a $3.432$-approximation for grand-bundling in the special
case of iid values. In this paper we provide improved bounds of $2.5$ for both
these results. For our grand-bundling upper-bound, in particular, we have to
forgo that ``high-level'' use of Kupfer's bound, and provide a more tailored
analysis (\cref{th:B-upper-bound-iid}). 
From a technical perspective, it is interesting to note that the analysis of
\parencite{huang2026improvedapproximation} does not invoke any duality-based
arguments, and relies on purely probabilistic techniques, including
tensorization inequalities for cumulative entropy. 

Furthermore, \textcite{huang2026improvedapproximation} presents a lower bound
exceeding $2.11$ for the approximation ratio of $\max\{\srev,\brev\}$, witnessed
by a  (deterministic) partition mechanism with three bundles, thus improving the
currently best lower bound of $2$ by~\textcite{Rubinstein2016}. In our paper we
further improve this lower bound to $2.19$, using a deterministic auction that
is not a partition mechanism (\cref{sec:M-lower-bound}). Finally, although not
directly comparable with our results in this paper, it is important to mention
that \textcite{huang2026improvedapproximation} also studies the competition
complexity of simple mechanisms, via which he also establishes a tight
$2$-approximation for grand-bundling, for the special case of
regularly-distributed items.

\section{Model and Notation}
\label{sec:model-notation}

Let $\R$ and $\R_{+}$ denote the sets of real and nonnegative real numbers,
respectively. We use $I\coloneq [0,1]$ to denote the real unit interval and
$I_K\coloneq \sset{0,1/K,2/K,\dots,(K-1)/K,1}$ to denote the unit $K$-grid,
where $K$ is a positive integer. For convenience, we will also use $[K]\coloneq
\sset{1,2,\dots,K}$ and $\Kzero \coloneqq [K]\union\ssets{0}$. Furthermore, for
any positive integer $n$, we let $H_n\coloneq\sum_{i=1}^n\frac{1}{i}$ denote the
$n$-th harmonic number. We will also use standard shorthand notation for the
minimum (see, e.g., \textcite[Appendix~B]{Knuth2011TAOCP4A}) and positive-part
operators, whenever this assists readability: $x\minshort y \coloneq
\min\ssets{x,y}$ and $\pospart{x}= \max\ssets{x,0}$, for any $x,y\in\R$.

When working with vectors we use boldface notation
$\vec{x}=(x_1,x_2,\dots,x_m)\in \R^m$ and their coordinate-wise (partial)
ordering: $\vec{x} \leq \vec{y} \ifif x_j\leq y_j\;\; \forall j\in[m]$. We also
use the standard notation $\maxnorm{\vec{x}}\coloneq\max_{j\in[m]}\card{x_j}$
for the maximum norm. Finally, a function $g:\R_{+}^m\map\R$ will be called
$L$-Lipschitz, for some parameter $L>0$, if $\card{g(\vec{x})-g(\vec{y})} \leq
L\cdot\maxnorm{\vec{x}-\vec{y}}$ for all $\vec{x},\vec{y}\in\R_{+}^m$.

We will identify probability distributions by their cumulative function (CDF),
usually denoted by $F$, and we will denote $X\sim F$ for a random variable $X$
distributed according to distribution $F$. Also, we use
$\support(F)=\support(X)$ for the support of a distribution or a random
variable. For distributions over $\R$, we use the standard notation
$F(x^{-})=\lim_{t\to x^{-}} F(x)= \prob[X\sim F]{X<x}$. In that way, we can
express the tail probabilities as $\prob{X\geq x}=1-F(x^{-})$  and $\tail{F}(x)
= \prob{X> x}=1-F(x)$; of course, for the case of continuous distributions, the
two quantities are equal. Finally, we will say that a random event holds
\emph{almost surely (a.s.)}, if it occurs with probability equal to $1$.

We consider settings with a single seller, a single \emph{buyer} (or
\emph{bidder}) and $m$ heterogeneous \emph{items} (or \emph{goods}); throughout
our paper we usually index items by $j\in[m]$. The seller has incomplete
knowledge of how much the buyer is willing to spend for the items, in the form
of a joint distribution (with CDF) $\vec{F}$ over a fixed value space
$\vecc{V}\subseteq \R_{+}^m$, from which the buyer values
$\vecc{X}=(X_1,X_2,\dots,X_m)\sim\vecc{F}$ for the items are drawn; the profile
of \emph{true values} $\vecc{x}=(x_1,x_2,\dots,x_m)\in \vecc{V}$ is
\emph{private} knowledge of the buyer only. 
We use $F_j$ to denote item $j$'s marginal distribution. For most of our results
in this paper we consider \emph{independent} items, that is, $\vecc{F}=F_1\times
F_2\times\dots\times F_m$ is a \emph{product} distribution. In that case,
whenever all marginals are identical, we will say that the items are \emph{iid}
and denote $\vec{F}=F^m$.

\subsection{Selling Mechanisms}

An \emph{auction} (or selling \emph{mechanism}) $\mathcal{M}=(\vec{a},p)$
consists of \emph{allocation} $\vec{a}=(a_1,\dots,a_m):\vec{V}\map [0,1]^m$ and
\emph{payment} $p:\vec{V}\map \R$ functions which, given as input a profile of
\emph{bids} $\vec{b}=(b_1,b_2,\dots,b_m)$ by the buyer, determine the
probability\footnote{Mechanisms such that $\vec{a}\in\ssets{0,1}^m$ are called
\emph{deterministic}, while the term \emph{lotteries} is sometimes also used for
general mechanisms with $\vec{a}\in[0,1]^m$, to emphasize the randomness of the
allocation rule.} $a_j(\vec{b})$ with which item $j\in[m]$ is sold to the buyer,
and the payment $p(\vec{b})$ that the buyer submits to the seller, in exchange.
We assume that the buyer is \emph{additive}, that is, her valuation from
receiving a subset of the items is simply the sum of the values of the items in
this subset. More specifically, the \emph{utility} of the buyer when submitting
bids $\vec{b}\in \vec{V}$ while having true values $\vec{x}\in \vec{V}$ is given
by
\begin{equation}
    \label{eq:utility-def}
    u_{\mathcal M}(\vec{b}|\vec{x}) \coloneq \vec{a}(\vec{b})\cdot\vec{x} - p(\vec{b})
    = \sum_{j=1}^m a_j(\vec{b})x_j - p(\vec{b}).
\end{equation}
In the following, we will feel free to drop the $\mathcal{M}$ subscript in the
above notation, whenever the auction deployed is clear from the context.

\paragraph{Truthful Mechanisms}
As it is standard in optimal auction design, in this paper we want to focus on
selling mechanisms that give no incentives to the buyer to misreport her true
values for the items, or abstain from the auction. Formally, an auction
$\mathcal{M}$ will be called \emph{truthful}, if for all true values and bids
$\vec{x},\vec{b}\in\vec{V}$ it holds that
\begin{align}
    u_{\mathcal{M}}(\vec{x}|\vec{x}) &\geq u_{\mathcal{M}}(\vec{b}|\vec{x})
    \label{eq:dsic-def} \tag{DSIC}\\
    \intertext{and}
    u_{\mathcal{M}}(\vec{x}|\vec{x}) &\geq 0. \label{eq:ir-def} \tag{IR}
\end{align}
Property~\eqref{eq:dsic-def} is also known in the literature as
\emph{dominant-strategy incentive compatibility} and \eqref{eq:ir-def} as
\emph{individual rationality}. For convenience, we denote the set of all
truthful mechanisms over a given value space $\vec{V}\subseteq\R^m$ by
$\mathcal{T}_{\vec{V}}$.

\paragraph{Revenue} Our main optimization objective throughout this paper is
that of maximizing the seller's \emph{revenue}; that is, given a distribution
$\vec{F}$ over a value space $\vec{V}$, design a truthful mechanism
$\mathcal{M}=(\vec{a},p)\in\mathcal{T}_{\vec{V}}$ with the highest expected
payment $\expect[\vec{X}\sim \vec{F}]{p(\vec{X})}$ by the buyer. We denote the
corresponding optimal benchmark by
\begin{equation}
    \label{eq:opt-rev-def}
    \rev_{\vec{V}}(\vec{X}) \coloneq \sup_{(\vec{a},p)\in \mathcal{T}_{\vec{V}}} \expect{p(\vec{X})}.
\end{equation}
Subscript $\vec{V}$ in the above notation of the optimal-revenue operator is
conceptually essential: different value spaces result in a different feasibility
space $\mathcal{T}_{\vec{V}}$ for truthfulness and thus, to a different
optimization problem. Nevertheless, it is known that requiring truthfulness
outside the ``realizable'' value points (and thus, shrinking the feasibility
space) does not actually reduce the revenue objective. Formally, this can be
expressed as 
\begin{align*}
    \rev_{\vec{V}}(\vec{X}) &=\rev_{\support(\vec{X})}(\vec{X}) &&\text{for all}\;\; \support(\vec{X}) \subseteq \vec{V} \subseteq \R^m_+,
\end{align*} 
which can be shown, e.g., immediately via the proof
of~\textcite[Lemma~13]{Hart2015}. As a result, we will feel free to drop the
$\vec{V}$ subscript in~\eqref{eq:opt-rev-def} and simply denote $\rev(\vec{X})$
for the optimal revenue. We will also sometimes interchangeably use
$\rev(\vec{F})$ to denote $\rev(\vec{X})$ if $\vec{X}\sim\vec{F}$.

\subsection{Simple Deterministic Auctions}
\label{sec:simple-auctions}
One of the simplest and more natural ways to sell a single item (i.e., for
$m=1$) is to use \emph{posted-price} mechanisms. These are deterministic
mechanisms that just offer the item to the buyer for a fixed price $r\in\R_+$;
the buyer accepts if and only if she can afford the item (that is, if her true
value is at least $r$) and pays the fixed price $r$. It is not hard to see that
posted-price mechanisms are truthful. Furthermore, from the seminal work
of~\textcite{Myerson1981a} (see also~\parencite{Monteiro2010} for general
distributions that may contain atoms) we know that this is the optimal way to
sell a single item, even among the much richer class of lotteries; that is, for
any real-valued random variable $X\in\R_+$ it is:
\begin{equation}
    \label{eq:monopoly-revenue}
    \rev(X) \coloneqq \sup_{r\geq 0} r \prob{X\geq r} = \sup_{r \geq 0} r \left[1-F(r^{-})\right],
\end{equation}
when $F$ is a distribution over $V$ such that $X\sim F$.

Moving to higher dimensions ($m\geq 2$), we consider two simple deterministic
mechanisms, that try to generalize the single-item posted-pricing idea to the
multi-item case. First, the seller can offer all items \emph{separately},
offering for each one of them a tailored take-it-or-leave-it price; the optimal
revenue under such mechanisms is 
\begin{equation}
    \label{eq:srev-def}
    \srev(\vec{X})\coloneq\sum_{j=1}^m \rev(X_j) \overset{\eqref{eq:monopoly-revenue}}{=} \sum_{j=1}^m \sup_{r \geq 0} r \left[1-F_j(r^{-})\right],
\end{equation}
where $F_j$ is the marginal distribution of $X_j$, for $j\in[m]$. At the other
end of the spectrum, we can alternatively try to sell all items in a single,
\emph{full bundle}, by offering a single take-it-or-leave-it price for all of
them; this results in a revenue of 
\begin{equation}
    \label{eq:brev-def}
    \brev(\vec{X})\coloneq \rev(X_1+X_2+\dots+X_m).
\end{equation}

\subsection{Discrete Auctions}
\label{sec:discrete-auction-prelims}

The technical backbone of our paper deals with the analysis of \emph{discrete}
auctions, where the value space $\vec{V}$ (and thus, also the value
distributions $\vec{F}$) is a discrete, finite subset. In particular, we study
auctions over the unit\footnote{The normalization within the unit interval $I$
is essentially without loss with respect to our revenue analysis, and simply a
matter of scaling. More precisely, it is easy to see that the optimal revenue
operator~\eqref{eq:opt-rev-def} is positive homogeneous, i.e.\
$\rev_{\vec{\lambda\cdot V}}(\lambda \cdot\vec{X})=\lambda\cdot
\rev_{\vec{V}}(\vec{X})$ for any scalar $\lambda>0$ (see, e.g.,
\parencite[p.~322]{Hart:2017aa}). In other words, all our results in this paper
transfer automatically to distributions over general bounded domains of the form
$[0,H]$, for any (fixed) $H$. Note that the restriction to such bounded
distributions is, again, essentially without loss: starting from a value
distribution $X_j$, $j\in [m]$, defined over $[0,\infty)$, one can truncate to
$\tilde{X}_j=\min\{X_j,H\}$ and perform our revenue analysis on the new,
truncated instance. Then, taking the limit as $H\to\infty$, one recovers the
appropriate approximation; see, e.g., \parencite[Theorem~12]{HartReny2019} for a
rigorous way of formulating this.} $K$-grid $I_K$, since this corresponds to the
natural interpretation of a ``bidding language'' consisting of multiples of some
``minimum currency'' $1/K$; the auction designer can choose how large integer
$K$ is, making the setting more expressive. In the limit as $K\to\infty$, we
expect this setting to ``approximate'' the classic, continuous-valued auction
theory~\parencite{Krishna2009a,Menezes2005}; we formalize this
in~\cref{prop:discretization-approximation-revenues} below. 

In the detailed analysis of the discrete setting $\vec{V}\coloneq I_K^m$, in
abuse of notation and for reasons of clear presentation, we make heavy use of
value indices. Representing a value $\vec{v}=(k_1/K,k_2/K,\dots,k_m/K)\in I_K^m$
by the index vector
\[\vec{k} = (k_1, \dots, k_m) \in [K]_*^m,\] we denote allocation variables by
$a_j(\vec{k})$ as well as payment variables by $p(\vec{k})$. For item $j$ with
value $X_j \sim F_j$ we denote the marginal probability mass by
$f_j({k_j})=\prob{X_j = \frac{k_j}{K}}$. The joint probability of value vector
$\vec{v}$ is denoted by $f(\vec{k})$ and, if the items' values are independently
distributed
\[f(\vec{k}) = \prod_{j=1}^m f_j({k_j}).\] Furthermore, for a single fixed item
value we use $\vec{k}_{-j}$ to denote the indices of all other values except
$k_j$ and $\vec{k} = (k_j,\vec{k}_{-j})$. By this we also write $f(\vec{k}) =
f_j(k_j)f_{-j}(\vec{k}_{-j})$.

To properly transfer our results from the discrete to the general setting
(including the standard continuous setting), we will agree on a canonical way to
discretize any general, multidimensional value distribution given to us. Of
course, there are many ways to discretize a distribution; in this paper, we
choose one of the simplest, and more natural ones: rounding down to the closest
point of our $K$-grid. We formalize this in~\cref{def:k-grid-discretization}
below. This type of discretization is standard in the algorithmic game theory
community; see, e.g., \parencite[Sec.~6.1]{Daskalakis2012b}, \parencite[Sec.~4.1
and p.~168--169]{Cai2019}, \parencite[Sec.~2]{Gonczarowski2021}.

\begin{definition}[Distribution Discretization]
    \label{def:k-grid-discretization}
    Let $\vec{F}$ be an arbitrary distribution over $\R_+^m$, and let random
    variable $\vec{X}=(X_1,X_2,\dots,X_m)\sim\vec{F}$. For a positive integer $K$,
    the \emph{(downward) $\frac{1}{K}$-lattice discretization} of $\vec{F}$, denoted
    by $\vec{F}^{(K)}$, is the distribution of the random variable $\vec{X}^{(K)}$
    defined by 
    \begin{equation*}
        \vec{X}^{(K)}\coloneq\frac{\lfloor K\cdot \vec{X} \rfloor}{K}= \left(\frac{\lfloor K X_1 \rfloor}{K},\frac{\lfloor K X_2 \rfloor}{K},\dots,\frac{\lfloor K X_m \rfloor}{K}\right).
    \end{equation*}
\end{definition}

Then, after analysing the revenue of such discrete distributions, we can
transfer them back to the original general setting, via the following formal
bounds:

\begin{proposition}[Discrete vs General Distributions]
    \label{prop:discretization-approximation-revenues}
    Let $\vec{X}$ be a random variable over the unit hypercube $I^m$. Then, for
    its downward discretization $\vec{X}^{(K)}$ in the unit $K$-grid $I^m_K$, for
    any positive integer $K$, the following inequalities hold:
    \begin{align}
        0 \leq R(\vec{X}) - R(\vec{X}^{(K)}) & \leq \frac{m}{K}, \label{eq:discrete-distribution-approximation-srev-brev} \\
        \intertext{for any operator $R\in\sset{\srev,\brev,\max\ssets{\srev,\brev}}$, and}
        \card{\rev(\vec{X}) - \rev(\vec{X}^{(K)})} &\leq \frac{2m}{\sqrt{K}}-\frac{m}{K}=O\left(\frac{m}{\sqrt{K}}\right). \label{eq:discrete-distribution-approximation-rev}
    \end{align}
\end{proposition}
\begin{proof}[Proof overview.]
    The full proof of~\cref{prop:discretization-approximation-revenues} can be
    found in~\cref{appendix:proof-discretization-approximation-revenues}. Its
    optimal-revenue part (see~\cref{eq:discrete-distribution-approximation-rev})
    is a direct application of a general revenue-continuity result
    by~\textcite{hart2026rootrevenuecontinuity}, while for the proximity bounds
    for the various deterministic-revenue operators
    (see~\cref{eq:discrete-distribution-approximation-srev-brev}) we make use of
    standard probabilistic arguments about the single-dimensional revenue
    operator~\eqref{eq:monopoly-revenue} and the linearity of selling-separately
    (see~\eqref{eq:srev-def}). Nevertheless, to the best of our knowledge, a
    unifying statement (and proof) of approximation bounds, as stated
    in~\cref{prop:discretization-approximation-revenues}, has not appeared before
    in the optimal-auctions literature, and therefore we present it here, as we
    believe it could be useful for follow-up work as well. Additionally, for
    proving~\eqref{eq:discrete-distribution-approximation-srev-brev}, we give a
    technical lemma that is more general than the scope for which we apply it in
    this paper (see~\cref{lemma:approx-revenue-lipschitz-general}); however we
    believe that this might also be of independent interest, and thus we choose to
    state and prove it in its full generality
    in~\cref{appendix:proof-discretization-approximation-revenues}.    
\end{proof}

\section{Selling Separately: A Mathematical Programming Approach}
\label{sec:srev-programming-formulatio}

In this section we analyse the optimization problem of the approximation ratio
of selling each item separately compared to the optimal expected revenue in the
discrete setting. Considering the probability masses also as variables, the
program becomes highly nonlinear and standard primal-dual linear programming
techniques may fail to provide upper bounds or optimality certificates. Using
Lagrangian duality and carefully choosing the Lagrangian multipliers, we obtain
an upper bound in the form of an elegant, and technically very useful inequality
(see \cref{lemma:rev_upper_bound_master_caps}), involving capped versions of the
random variables of each item's value; we will call this the \emph{master
revenue inequality}. We refine this technique of a flow interpretation of the
Lagrangian multipliers in \cref{sec:selling_separately_2_iid} for the special
case of two iid items and yield a tight dual bound.

\subsection{The Formulation}
\label{sec:srev-ratio-formulation}
Next we state the optimization problem which computes the approximation ratio of
selling $m$ independently distributed items optimally versus selling them
separately at their optimal individual price. Note that, without loss we will
``silently''\footnote{\label{foot:srev-strictly-positive}Formally, one should
explicitly add an extra constraint of $\sum_{j=1}^m\sum_{k_j=1}^{K} f_j({k_j}) >
0$ to our optimization program~\eqref{mathematical_program} below; this would
not have any effect on our subsequent analysis or results as we do not relax any
constraints which purely constraint probability masses, so for simplicity we
choose to drop it altogether. Note that this is an important point to clarify,
since formally one might think that \eqref{mathematical_program} is
\emph{unbounded}, resulting in an approximation ratio of $\infty$ for
$\srev(\vec{X})$.} restrict our analysis to instances that satisfy
$\srev(\vec{X})>0$; due to the discreteness of the random variables, condition
$\srev(\vec{X})=0$ would imply that all item values are deterministically equal
to $0$, resulting trivially to $\rev(\vec{X})=\srev(\vec{X})=\brev(\vec{X})=0$
in this case.

\begin{align}
    \sup_{\vec{a},\rho, p, f} \quad &\rho \tag{P}\label{mathematical_program}\\
    \text{s.t.}\quad
    &\rho \ge0, & \notag\\
    &f_j({k_j})\ge0,                  &k_j \in \Kzero, j \in [m] \notag\\
    &\sum_{k_j=0}^{K} f_j({k_j}) = 1,  & j \in [m] \notag\\
    &a_j(\vec{k})\leq 1, & \vec{k} \in \Kzero^m &\qquad& [\psi_j(\vec{k})]\notag\\
    &a_j(\vec{k})\geq 0, & \vec{k} \in \Kzero^m && [\alpha_j(\vec{k})]\notag\\
    &\sum_{j=1}^m \frac{k_j}{K} a_j(\vec{k}) - p(\vec{k}) \geq \sum_{j=1}^m \frac{k_j}{K} a_j(\vec{k}^\prime) - p(\vec{k}^\prime), & \vec{k}, \vec{k}^\prime \in \Kzero^m &&[\lambda_{\vec{k} \to \vec{k}^\prime}]\notag\\
    &\sum_{j=1}^m \frac{k_j}{K} a_j(\vec{k}) - p(\vec{k}) \geq 0, & \vec{k} \in \Kzero^m &&[\mu(\vec{k})]\notag\\
    &\sum_{\vec{k}^\prime \in \Kzero^m} f(\vec{k}^\prime) p(\vec{k}^\prime) \geq
    \rho\cdot \sum_{j=1}^m \frac{k_j}{K}\sum_{l=k_j}^K f_j(l), &\vec{k} \in \Kzero^m. && [\theta({\vec{k}})] \notag
\end{align}

We already indicate Lagrangian multipliers of constraints we aim to relax in the
following at the end of each corresponding line. Observe how $\rho$ exactly
captures the approximation ratio: to maximize the value of rho the program
maximizes the expected revenue $\sum_{\vec{k}^\prime \in \Kzero^m}
f(\vec{k}^\prime) p(\vec{k}^\prime)$ which is the left-hand side of every
$\theta({\vec{k}})$ constraint. On the right-hand side $\rho$ is multiplied with
the expected revenue offering each item for the posted-price vector indexed by
$\vec{k}$. The optimal solution therefore satisfies
\[\rho = \sup_{\vec{X}} \frac{\rev(\vec{X})}{\srev(\vec{X})},\]
where $\vec{X}$ ranges over all random variables on the $m$-dimensional $K$-grid $I_K^m$.

Maximizing the Lagrangian function for each given feasible configuration of
multipliers, thus, serves as an upper bound for the approximation ratio $\rho$.
We relax all constraints except for the first three, which ensure that the
approximation ratio is nonnegative as well as that we consider a valid
probability distribution. Therefore, we obtain

\begin{align}
    &L(\rho, \vec{f}, \vec{a}, \vec{p}; \vec{\psi}, \vec{\alpha}, \vec{\lambda}, \vec{\mu}, \vec{\theta}) = \tag{L}\label{eq:Lagrangian}\\
    &\hspace{1cm} \rho + \sum_{j=1}^m \sum_{\vec{k}} \psi_j(\vec{k})\left[1-a_j(\vec{k})\right] + \sum_{j=1}^m \sum_{\vec{k}} \alpha_j(\vec{k})a_j(\vec{k}) \notag\\
    &\hspace{1cm} +\sum_{\vec{k}, \vec{k}^\prime} \lambda_{\vec{k}\to\vec{k}^\prime} \left[\sum_{j=1}^m \frac{k_j}{K} a_j(\vec{k}) - p(\vec{k}) - \sum_{j=1}^m \frac{k_j}{K} a_j(\vec{k}^\prime) + p(\vec{k}^\prime)\right]
    \notag\\
    &\hspace{1cm} + \sum_{\vec{k}} \mu(\vec{k}) \left[\sum_{j=1}^m \frac{k_j}{K} a_j(\vec{k}) - p(\vec{k})\right] \notag\\
    &\hspace{1cm} + \sum_{\vec{k}} \theta(\vec{k}) \left[\sum_{\vec{k}^\prime} f(\vec{k}^\prime) p(\vec{k}^\prime) -
    \rho\cdot \sum_{j=1}^m \frac{k_j}{K}\sum_{l=k_j}^K f_j(l)\right],\notag
\end{align}
paired with the remaining constraints of the approximation ratio being
nonnegative and each item's probability mass functions being valid (i.e.,
pointwise nonnegative and summing up to one). Furthermore, we remind the reader
our ``silent'' assumption that $\srev(\vec{X})=0$ (see the discussion
in~\cref{foot:srev-strictly-positive}).

In the following section we choose multipliers for the general problem of $m$
independent items and prove an upper bound in form of a weak dual solution.
Later on in \cref{sec:selling_separately_2_iid} we characterize a tight dual
solution, i.e., an optimality certificate, for the case of two iid items.

\subsection{The Master Inequality: Bounding Optimal Revenue via Duality}
\label{sec:lagrangian-dual-many-items-weak}

In this section we prove the cornerstone for our results: a weak dual upper
bound derived as a solution of \eqref{eq:Lagrangian}. For this we fix an
arbitrary distribution and choose Lagrangian multipliers such that the
approximation-ratio variable, all payment variables, and all allocation
variables cancel from the Lagrangian as their coefficients become zero. What
remains is a constant upper bound on the approximation ratio, which can then be
re-expressed as an upper bound on the expected revenue. The remainder of this
section is devoted to proving this result.

\begin{proposition}[Master Revenue Inequality]
    \label{lemma:rev_upper_bound_master_caps}
    Let $X_1,X_2,\dots,X_m$ be independent random variables taking values in the unit $K$-grid $I_K$. Then, for any $c_1,c_2,\dots,c_m\in\R_{+}$: 
    \begin{equation}\label{eq:rev_master_inequality}
        \rev(\vec{X}) \leq \sum_{j=1}^m \max \sset{ r_j, \expect{Z_j} } + \expect{\sum_{j=1}^m Y_j - \max_{j\in[m]} Y_j },
    \end{equation}
    where
    \begin{equation*}
        r_j\coloneqq \rev(X_j), \qquad Z_j\coloneqq \min\sset{X_j,c_j} \qquad\text{and}\qquad Y_j\coloneqq \pospart{X_j-c_j},
    \end{equation*}
    for  $j=1,2,\dots,m$.
\end{proposition}

\paragraph{Core--tail interpretation}\label{paragraph:core-tail-interpretation}
The identity $X_j=Z_j+Y_j$ provides an interesting core--tail interpretation to
our master inequality~\eqref{eq:rev_master_inequality}, related to the
decompositions of \textcite{Li2013a,Babaioff2020}, and their duality
interpretation in~\parencite[Section~6]{Cai2019}. Here $Z_j$ can be viewed as
the capped value and $Y_j$ as its excess over the cap. The key quantitative
feature of our bound is the joint term $\max\{r_j,\expect{Z_j}\}$, together with
a residual expectation that excludes the largest excess. In the proof below, a
one-dimensional price above the cap leads to the bound $r_j$, whereas a price
below the cap leads to $\expect{Z_j}$. Routing the source mass along a
coordinate with largest excess then yields the subtraction of $\max_jY_j$.
Consequently, only realizations with at least two items exceeding their
respective caps contribute to the residual expectation.

\begin{proof}[Proof of~\cref{lemma:rev_upper_bound_master_caps}]
    Let each item's probability mass function $f_j$ be arbitrary but fixed.
    Recalling the definition of the optimal revenue selling each item
    separately~\eqref{eq:srev-def} we write $\srevshort = \srev(\vec{X})$ for
    convenience and assume $\srevshort>0$.
    
    Choose a vector of optimal price indices $\vec{r}$ such that selling each item
$j$ separately at value indexed by $r_j$ achieves exactly $\srevshort$, i.e.,
\[r_j \in \argmax_{k\in\Kzero} \frac{k}{K} \sum_{l=k}^K f_{j}(l),\]
    and set 
    \begin{equation}
        \theta(\vec{r}) \coloneq \frac{1}{\srevshort}, \qquad \theta(\vec{r}^\prime) \coloneq 0, \; \vec{r}^\prime \neq \vec{r}.
    \end{equation}
    By this, the corresponding term in the Lagrangian \eqref{eq:Lagrangian} is
\[
    \sum_{\vec{r}} \theta({\vec{r}}) \left[\sum_{\vec{k}} f(\vec{k}) p(\vec{k}) -
    \rho\cdot \sum_{j=1}^m \frac{r_j}{K}\sum_{l=r_j}^K f_j(l) \right]
    =\frac{1}{\srevshort} \left[\sum_{\vec{k}} f(\vec{k}) p(\vec{k}) -
    \rho\cdot \srevshort \right]
    = \sum_{\vec{k}} \frac{f(\vec{k})}{\srevshort} p(\vec{k}) -
    \rho.
\]
    This cancels the initial $\rho$ from the Lagrangian and leaves the nonnegative
    coefficient $\frac{f(\vec{k})}{\srevshort}$ for every payment variable
$p(\vec{k})$. We call this coefficient the initial source mass at $\vec{k}$.
    Furthermore, we interpret the grid as a complete directed graph where multiplier
$\lambda_{\vec{k}\to\vec{k}^\prime}$ sends mass from $\vec{k}$ to
$\vec{k}^\prime$, while $\mu(\vec{k})$ absorbs mass leaving the grid at point
$\vec{k}$. Hence, the payment coefficient at $\vec{k}$ cancels exactly if the
    net flow leaving $\vec{k}$, including absorption by $\mu(\vec{k})$, equals the
    initial source mass $\frac{f(\vec{k})}{\srevshort}$.
    
    In our construction we only allow flows to neighbouring grid points whose
    indices differ by one in exactly one coordinate. For this we define a flow
    routing $\vec{w}$ assigning mass in every grid point to the different coordinate
    directions. Formally, for every grid point $\vec{k}$ let the routing weights
    satisfy
    \begin{equation}\label{eq:flow_rule_definition}
        w_j(\vec{k}) \geq 0 ,\quad j\in [m], \qquad \sum_{j=1}^m w_j(\vec{k}) = 1.
    \end{equation}
    Thus, a flow routing $\vec{w}$ assigns the share $\frac{f(\vec{k})}{\srevshort}
    w_j(\vec{k})$ of the initial source mass to direction $j$. Our flow construction
    works for any flow routing, and we are free to choose it after bounding the
    objective.
    
    For each item $j$ and each coordinate line $\vec{k}_{-j}$ define the unscaled
    routed source mass and its cumulative mass at $k \in \Kzero$ by
    \[\begin{aligned}
        g_j(k, \vec{k}_{-j}) &\coloneq f(k,\vec{k}_{-j}) w_j(k,\vec{k}_{-j}),\\
        G_j(k, \vec{k}_{-j}) &\coloneq \sum_{l=k}^K g_j(l, \vec{k}_{-j}).
    \end{aligned}\] Note that the actual initial source mass assigned to
    direction $j$ is $\frac{1}{\srevshort}g_j(k, \vec{k}_{-j})$. For a fixed
    $\vec{k}$ we have $g_j(k_j, \vec{k}_{-j})=f(\vec{k})w_j(\vec{k})$, and
    therefore
    \begin{equation}\label{eq:feasible_flow}
        \sum_{j=1}^m g_j(k_j, \vec{k}_{-j}) = f(\vec{k}).
    \end{equation}
    Furthermore, we define the maximum line revenue by
    \begin{equation}\label{eq:def_line_revenue}
        R_j(\vec{k}_{-j}) \coloneq \max_{k \in \Kzero}  \sum_{l=k}^K \left[ \frac{k}{K} g_j(l,\vec{k}_{-j}) + \frac{l}{K} \left( f(l,\vec{k}_{-j})-g_j(l,\vec{k}_{-j}) \right) \right].
    \end{equation}
    
    With this notation we construct our flow as follows. For each item $j\in[m]$ and
    each coordinate line $\vec{k}_{-j}$, for every $k\in [K]$ we define the
    \emph{rectangular flow}
    \begin{equation}\label{eq:rect_flow}
        \lambda^\mathrm{rect}_{(k,\vec{k}_{-j})\to(k-1,\vec{k}_{-j})} = \frac{1}{\srevshort} G_j(k, \vec{k}_{-j}),
    \end{equation}
    as well as the \emph{circular flow} in both directions:
    \begin{equation}
        \begin{aligned}\label{eq:circ_flow}
            \lambda^\mathrm{circ}_{(k,\vec{k}_{-j})\to(k-1,\vec{k}_{-j})} &= \lambda^\mathrm{circ}_{(k-1,\vec{k}_{-j})\to(k,\vec{k}_{-j})} \\ &= \frac{K}{\srevshort} \left( R_j(\vec{k}_{-j}) - \sum_{l=k}^K \left[ \frac{k}{K} g_j(l,\vec{k}_{-j}) + \frac{l}{K} \left( f(l,\vec{k}_{-j})-g_j(l,\vec{k}_{-j}) \right) \right] \right).
        \end{aligned}    
    \end{equation}
    The multipliers of the incentive compatibility constraints are their sums:
    \begin{equation}\label{eq:IC-multipliers}
        \begin{aligned}
            \lambda_{(k,\vec{k}_{-j})\to(k-1,\vec{k}_{-j})} &= \lambda^\mathrm{rect}_{(k,\vec{k}_{-j})\to(k-1,\vec{k}_{-j})} + \lambda^\mathrm{circ}_{(k,\vec{k}_{-j})\to(k-1,\vec{k}_{-j})},\\
            \lambda_{(k-1,\vec{k}_{-j})\to(k,\vec{k}_{-j})} &= \lambda^\mathrm{circ}_{(k-1,\vec{k}_{-j})\to(k,\vec{k}_{-j})}.
        \end{aligned}
    \end{equation}
    We finalize the flow definition choosing $\mu$ such that for all $j$ and all
$\vec{k}_{-j}$ the mass reaching the lower end of every such coordinate line is
    absorbed. As also multiple coordinate lines can end at the same grid point, we
    set
    \begin{equation}\label{eq:IR-multipliers}
        \mu(\vec{k}) = \frac{1}{\srevshort} \sum_{\substack{j\in[m]\\k_j=0}} G_j(0,\vec{k}_{-j}).
    \end{equation}
    
    To cancel the remaining allocation coefficients at the boundary, we choose $\alpha_j$ and $\psi_j$ for every $j$ and every $\vec{k}_{-j}$ as follows
    \begin{align}\label{eq:def_alpha_multiplier}
        \alpha_j{(0,\vec{k}_{-j})} &= \frac{1}{\srevshort} \left[R_j(\vec{k}_{-j}) - \sum_{l=0}^K \frac{l}{K} \left( f(l,\vec{k}_{-j})-g_j(l,\vec{k}_{-j}) \right) \right],\\
        \label{eq:def_psi_j_multiplier}
        \psi_j{(K,\vec{k}_{-j})} &= \frac{1}{\srevshort} R_j(\vec{k}_{-j}).
    \end{align}
    All multipliers not specified above are set to zero. Observe that the circular
    flows are nonnegative, because \eqref{eq:def_line_revenue} takes the maximum
    over all $k$. The same definition evaluated at $k=0$ shows that
$\alpha_j{(0,\vec{k}_{-j})}\geq 0$. All other multipliers are nonnegative
    directly from their definitions and our choice of multipliers is feasible.
    
    Now we evaluate the Lagrangian objective~\eqref{eq:Lagrangian} under this choice
    of feasible multipliers. Further upper bounding this objective yields the
    proposed master inequality~\eqref{eq:rev_master_inequality}. We do that
    in~\cref{append:lemma:rev_upper_bound_master_caps}.
\end{proof}

\paragraph{Special choices of caps}
Setting $c_j=0$ for every item $j$ in our master
inequality~\eqref{eq:rev_master_inequality} gives
\begin{equation}
    \label{eq:revenue-upper-bound-uncapped-CDW}
    \rev(\vec{X})
    \leq \srev(\vec{X})
    + \expect{\sum_{j=1}^m X_j-\max_{j\in[m]}X_j}.
\end{equation}
This bound follows from the single-buyer decomposition
of~\textcite[Section~6]{Cai2019}.\footnote{\label{foot:CDW-uncapped-revenue-discussion}Their
Theorem~28 bounds revenue by the sum of the \textsc{Single} and
\textsc{Nonfavorite} terms. For a single additive buyer, Lemma~29 bounds the
former by $\srev$, while bounding allocation probabilities by one bounds the
latter by $\expect{\sum_j X_j-\max_j X_j}$.} It is also the starting point of
the analysis in~\parencite[Lemma~4.1]{cai2026improvedrevenueguaranteesselling}.
At the other extreme, since our values lie in $[0,1]$, choosing $c_j=1$ for
every item recovers the welfare bound $\rev(\vec{X})\leq\sum_j\expect{X_j}$. The
freedom to choose intermediate, possibly different caps allows us to strengthen
the uncapped bound; for our approximation guarantees, in
\cref{sec:approximation_ratios} we use positive caps chosen in terms of the
single-item revenues $r_j=\rev(X_j)$. The choice $c_j=r_j$ ensures that
$\expect{\min\{X_j,c_j\}}\leq r_j$, while several of our bounds use caps
satisfying the stronger condition
\[
\expect{\min\{X_j,c_j\}}=r_j
\]
Notice that such caps exist because $c\mapsto\expect{\min\{X_j,c\}}$ is
continuous on $[0,1]$ and ranges from $0$ to $\expect{X_j}\geq r_j$ (due
to~\eqref{eq:ir-def}). Both choices make the first sum
in~\eqref{eq:rev_master_inequality} equal to $\srev$, leaving us to bound the
residual expectation in~\cref{sec:approximation_ratios}.

\section{Approximation Ratios of Simple Deterministic Auctions}
\label{sec:approximation_ratios}

In this section we make use of the power of the novel master
inequality~\eqref{eq:rev_master_inequality}. In particular, a careful choice of
the caps allows individual treatment of items, averaging caps to better handle
them together, or even comparing different choices and choosing the most
promising in different cases. First, we show an equivalent representation of
$\expect{\sum_{j=1}^m \pospart{X_j - c_j} - \max_{j\in[m]} \pospart{X_j - c_j}}$
as an integral over capped tail functions. Subsequently, we split our focus into
the independent and the iid case as, by our techniques, the handling of the caps
naturally distinguishes the two settings. In \cref{sec:independent-results} we
provide upper bounds of $\rev$ by the simple auctions $\srev$, $\brev$ and the
better of the two. \cref{sec:iid-results} subsequently handles the same upper
bounds under the additional assumption that items are iid. 

Recall the notation $\tail{F}_j(t) = 1-F_j(t) = \prob{X_j>t}$, define $Y_j
\coloneq \pospart{X_j - c_j}$ and with this $\tail{F}_j^{c_j}(t) \coloneq
\prob{X_j>c_j+t} = \prob{Y_j>t}$. For individual caps $c_j$ we obtain
\begin{align}
    \expect{\sum_{j=1}^m Y_j - \max_{j\in[m]} Y_j}
    &= \int_0^\infty \left[ \sum_{j=1}^m \prob{ Y_j > t} - \prob{\max_{j\in[m]}Y_j > t} \right] \, \mathrm{d}t \notag\\
    &= \int_0^\infty \left[ \sum_{j=1}^m \tail{F}_j^{c_j}(t) - 1 + \prod_{j=1}^m \prob{ Y_j \leq t} \right] \, \mathrm{d}t \notag\\
    &= \int_0^\infty \left[ \sum_{j=1}^m \tail{F}_j^{c_j}(t) - 1 + \prod_{j=1}^m (1-\tail{F}_j^{c_j}(t)) \right] \, \mathrm{d}t. \label{eq:expectation-of-function-of-tail}
\end{align}
For the special case of a common cap, i.e.\ $c_j=c$ for all $j\in[m]$, we can substitute $u = c+t$ in~\eqref{eq:expectation-of-function-of-tail} and we obtain
\begin{equation}\label{eq:common-caps}
    \expect{\sum_{j=1}^m Y_j - \max_{j\in[m]} Y_j} = \int_c^\infty \left[ \sum_{j=1}^m \tail{F}_j(u) - 1 + \prod_{j=1}^m (1-\tail{F}_j(u)) \right] \, \mathrm{d}u.
\end{equation}
Not only this brings back the uncapped tails, but also the integration starts at
the common cap $c$. This makes the effect of caps explicit as they directly
reduce the integration domain.  

To prove bounds for the expected revenue with our master
inequality~\eqref{eq:rev_master_inequality} our approaches in this section rely
in particular on proving upper bounds for
\eqref{eq:expectation-of-function-of-tail} and~\eqref{eq:common-caps}. The
following lemma provides some elementary but throughout the section useful upper
bounds for tail probabilities and their sums. Its proof can be found
in~\cref{append:lemma:bounds_on_sum_of_tails}.

\begin{lemma}
    \label{lemma:bounds_on_sum_of_tails}
    Let $X_1, X_2, \dots, X_m$ be independent random variables. Denote
    $r_j=\rev(X_j)$, $\srevshort=\srev(\vec{X})$, and
    $\brevshort=\brev(\vec{X})$. Then for any $u>0$ it holds
    \begin{align}
        \prob{X_j > u} &\leq \frac{r_j}{u}, & \sum_{j=1}^m \tail{F}_j(u) &\leq \frac{\srevshort}{u} \label{eq:bound_for_tail_srev} \\
        \prob{\sum_{j=1}^m X_j > u} &\leq \frac{\brevshort}{u}, & \sum_{j=1}^m \tail{F}_j(u) &\leq - \ln (1-\brevshort / u), \label{eq:bound_for_tail_brev}
    \end{align}
    where the last bound is defined only for $u>\brevshort$.
\end{lemma}

\subsection{General Independent Items}
\label{sec:independent-results}

In this part we prove~\cref{th:upper-bounds-independent-unified}, upper bounds
for the approximation ratios of selling independent items. We start with proving
an upper bound of grand-bundling (\cref{prop:B-upper-bound-independent}) where
the same caps are chosen for every item. For the selling separately
(\cref{th:S-upper-bound-independent}) as well as for the better of the two
(\cref{th:M-upper-bound-independent}) we choose individual caps for each item to
control items separately.

\begin{proposition}
    \label{prop:B-upper-bound-independent}
    For any independent random variables $X_1,X_2,\dots,X_m$ it holds that 
    \begin{equation}
        \label{eq:brev_upper_bound_m_ind}
        \rev(\vec{X}) \leq \left(m +1-H_{1/m}\right) \, \brev(\vec{X}),
    \end{equation}
    where $H_{1/m}\coloneqq \int_0^1 \frac{1-t^{1/m}}{1-t}\,\mathrm{d}t \geq \frac{1}{m}$.
\end{proposition}
\begin{proof}
    We will prove our proposition for random variables $X_j$ defined over the
    unit-grid $I_K$, for arbitrary positive integer $K$. This is
    enough,\footnote{Without explicitly mentioning it again, we will do exactly
    that for all the proofs of our upper-bound propositions
    in~\cref{sec:independent-results,sec:iid-results}.} due
    to~\cref{prop:discretization-approximation-revenues} (by taking
    $K\to\infty$). 
    
    For convenience denote $r_j=\rev(X_j)$. We choose the common cap $c=\max_{j}
    r_j$ for all items $j\in [m]$. Applying these common caps to the master
    inequality~\eqref{eq:rev_master_inequality} and using~\eqref{eq:common-caps}
    establishes our starting point
    \[\rev (\vec X) \leq mc + \int_c^\infty \left[ \sum_{j=1}^m \tail{F}_j(u) - 1 + \prod_{j=1}^m (1-\tail{F}_j(u)) \right] \, \mathrm{d}u. \]
    
    We now upper bound the integral. Consider the function $\xi_m(x) = m \left(1
    - (1-x)^{1/m} \right)-x$ for $0\leq x \leq 1$, which is nonnegative and
    monotone as $\xi_m(0)=0$ and $\xi_m^\prime(x) = \left((1-x)^{1/m-1}
    \right)-1 \geq 0$. Let 
    \[x(u) = 1 - \prod_{j=1}^m (1-\tail{F}_j(u))= \prob{\max_{j\in [m]} X_j > u} \leq \prob{\sum_{j=1}^m X_j > u}.\]
    Applying~\eqref{eq:bound_for_tail_brev} of \cref{lemma:bounds_on_sum_of_tails} directly yields $x(u) \leq  \min \sset{1, \frac{\brevshort}{u}}$. By the arithmetic mean-geometric mean inequality it holds
    \[\sum_{j=1}^m \tail{F}_j(t)  \leq m \left[ 1 - \left( \prod_{j=1}^m (1-\tail{F}_j(t)) \right)^{1/m} \right],\]
    hence, 
    \[\int_c^\infty \left[ \sum_{j=1}^m \tail{F}_j(u) - 1 + \prod_{j=1}^m (1-\tail{F}_j(u)) \right] \, \mathrm{d}u \leq \int_c^\infty \xi_m \left( \min \sset{1, \frac{\brevshort}{u}} \right) \, \mathrm{d}u . \]

    The minimum changes exactly at $u=\brevshort$, i.e., for $u\leq \brevshort$
    the minimum is one, and $\brevshort/u$ otherwise. Therefore, we use the
    master inequality~\eqref{eq:rev_master_inequality}
    and~\eqref{eq:common-caps} and obtain
    \[\begin{aligned}
        \rev(\vec{X}) &\leq m c + \int_c^\infty \xi_m \left( \min \sset{1, \frac{\brevshort}{u}} \right) \, \mathrm{d}u \\
        &= m c + \int_c^{\brevshort} \xi_m \left( 1 \right) \, \mathrm{d}u + \int_{\brevshort}^\infty \xi_m \left( \frac{\brevshort}{u} \right) \, \mathrm{d}u \\
        &= m c + (m-1) (\brevshort-c) + \int_{\brevshort}^\infty \xi_m \left( \frac{\brevshort}{u} \right) \, \mathrm{d}u.
    \end{aligned}
    \] To evaluate the last integral substitute $x=\brevshort/u$.
    Differentiating gives $\mathrm{d}u = -\brevshort/(x^2)\,\mathrm{d}x$ and the
    limits become one and zero. Hence, we reverse the limits' order as well as
    the integral's sign. Using integration by parts we obtain
    \[
    \begin{aligned}
        \int_{\brevshort}^\infty \xi_m \left( \frac{\brevshort}{u} \right) \, \mathrm{d}u &= \brevshort \int_{0}^1 \xi_m \left( x \right) \frac{1}{x^2}\, \mathrm{d}x \\
        &= \brevshort \left( \left[-\xi_m(x) \frac{1}{x} \right]_0^1 + \int_{0}^1 \left( \left((1-x)^{1/m-1} \right)-1 \right) \frac{1}{x}\, \mathrm{d}x \right) \\
        &= \brevshort \left( -(m-1) + \int_{0}^1  (1-x)^{1/m-1}\, dx - \int_{0}^1 \frac{1-(1-x)^{1/m}}{x}\, \mathrm{d}x \right) \\
        &= \brevshort \left( -(m-1) + m - H_{1/m} \right) = \brevshort \left( 1 - H_{1/m} \right)
    \end{aligned}
    \]
    By this and $c\leq \brevshort$, we obtain the desired result
    \[\rev(\vec{X}) \leq c + (m-1)\brev(\vec X) + \left( 1 - H_{1/m} \right) \brev(\vec X) \leq  \left( m + 1 - H_{1/m} \right) \brev(\vec X).\]
    
\end{proof}

Next we establish our upper bound for the approximation ratio of the optimal
revenue by selling each item separately, the canonical Myersonian mechanism. To
handle the items separately caps are chosen individually per item. Elementary
bounds for the capped tail functions and a monotone function allow us to upper
bound the integrand by means of the average selling separate revenue.

\begin{proposition}
    \label{th:S-upper-bound-independent}
    For any independent random variables $X_1,X_2,\dots,X_m$ it holds that  
    \begin{equation}
        \label{eq:S-upper-bound-independent}
        \rev(\vec{X}) \leq H_m \cdot \srev(\vec{X}).
    \end{equation}
\end{proposition}

\begin{proof}
    For convenience denote $r_j=\rev(X_j)$ and $\bar{r}=\frac{1}{m}
    \srev(\vec{X})$. We choose the individual caps $c_j = r_j$ for $j\in [m]$.
    Applying~\eqref{eq:bound_for_tail_srev} of
    \cref{lemma:bounds_on_sum_of_tails} with $u=r_j+t$ for any $t > 0$ it holds
    $\prob{X_j \geq r_j + t} \leq \frac{r_j}{r_j + t}.$ The right-hand side is
    concave in $r_j$. For this observe $\left(\frac{x}{x+t}\right)^\prime =
    \frac{t}{(x+t)^2}$ and $\left(\frac{x}{x+t}\right)^{\prime\prime}=
    -\frac{2t}{(x+t)^3}<0$. Therefore, 
    \[\frac{1}{m} \sum_{j=1}^m \prob{X_j \geq r_j + t} \leq \frac{1}{m}
    \sum_{j=1}^m \frac{r_j}{r_j + t} \leq   \frac{\frac{1}{m} \sum_{j=1}^m
    r_j}{\frac{1}{m} \sum_{j=1}^m r_j + t} = \frac{\bar{r}}{\bar{r} + t},\]
    where we apply Jensen's inequality at the second inequality sign. Consider
    the function $\eta_m(x) = mx - 1 + (1-x)^m$ for $0 \leq x \leq 1$, which is
    nonnegative and monotone as $\eta_m(0)=0$ and $ \eta_m^\prime (x) = m \left(
    1-(1-x)^{m-1} \right) \geq 0$. By the arithmetic mean-geometric mean
    inequality
    \[\prod_{j=1}^m \left( 1-\tail{F}_j^{c_j}(t) \right) \leq \left( 1 - \frac{1}{m} \sum_{j=1}^m \tail{F}_j^{c_j} (t) \right)^m.\]
    
    Using~\eqref{eq:expectation-of-function-of-tail} and the monotonicity of
    $\eta_m$ we obtain
    \[\int_0^\infty \left[ \sum_{j=1}^m \tail{F}_j^{c_j}(t) - 1 + \prod_{j=1}^m
    (1-\tail{F}_j^{c_j}(t)) \right] \, \mathrm{d}t \leq \int_0^\infty \eta_m
    \left(\frac{\bar{r}}{\bar{r} + t} \right) \, \mathrm{d}t .\] We substitute
    $x = \frac{\bar{r}}{\bar{r} + t}$. Solving for $t$ gives $t =
    \frac{\bar{r}}{x}-\bar{r}$ and differentiating $\mathrm{d}t = -\bar{r}/x^2
    \mathrm{d}x$. The integral limits become one to zero, hence, we reverse the
    limits' order as well as the integral's sign. Using integration by parts we
    obtain
    \[\begin{aligned} \int_0^\infty \eta_m \left(\frac{\bar{r}}{\bar{r} + t}
        \right) \, \mathrm{d}t &= \int_0^1  \eta_m(x) \frac{\bar{r}}{x^2} \,
        \mathrm{d}x \\
        &= \bar{r} \int_0^1  (mx - 1 + (1-x)^m) \frac{1}{x^2} \, \mathrm{d}x \\
        &= \bar{r} \left[  -(mx - 1 + (1-x)^m) \frac{1}{x} \right]_0^1 + \bar{r}
        \int_0^1  (m -  m(1-x)^{m-1}) \frac{1}{x} \, \mathrm{d}x \\
        &= - \bar{r} (m - 1) + \bar{r} m\int_0^1  (1 -  (1-x)^{m-1}) \frac{1}{x}
        \, \mathrm{d}x \\
        &= - \bar{r} (m - 1) + \bar{r} m H_{m-1} = m\bar{r} (H_m-1)
    \end{aligned}\] where we apply the integral representation of the harmonic
    numbers in the last two lines. Finally, use $m\bar{r} = \srev(\vec X)$ and
    each item's contribution $\expect{X_j \minshort r_j} \leq r_j$ in the master
    inequality~\eqref{eq:rev_master_inequality} to conclude
    \[\rev(\vec{X}) \leq \srev(\vec X) + \srev(\vec X) (H_m-1) = H_m \, \srev(\vec X).\]
\end{proof}

Next, we establish upper bounds on the approximation ratio for the better of
selling separately and selling the grand bundle, for both independent and iid
items as well as a constant bound for selling in the full bundle in the iid
setting. The proofs share the same main idea: starting from the master
inequality~\eqref{eq:rev_master_inequality} we choose caps such that the first
summand equals $\srev$. For the expectation term we use its integral
representation~\eqref{eq:expectation-of-function-of-tail} and bound it by the
revenues from selling separately and selling the grand bundle. In fact, the two
bounds for iid items follow from a single proof. We first state the technical
lemmas used in all three bounds.

In the first lemma we define a nonnegative, monotone, and convex function which
yields a new master inequality which is the starting point for all our constant
bound proofs. A second technical lemma provides an upper bound of an integral
which exponentially weights each tail function. This bound brings in $\brev$ as
well as it subtracts a weighted average of each item's individual revenue. The
resulting bound then is independent of the tail. The proofs of both lemmas can
be found
in~\cref{append:lemma:tail_inequalities,append:lemma:bound_on_weighted_tails}.

\begin{lemma}
    \label{lemma:tail_inequalities}
    Let $X_1, X_2, \dots, X_m$ be independent random variables and caps $c_j$
    chosen such that $\expect{X_j \minshort c_j} = r_j$ for every $j\in [m]$.
    Furthermore, let
    \begin{equation}\label{def:phi_function}
        \phi(x) = x-1+e^{-x}, 
    \end{equation}
    for $x\geq 0$ which is nonnegative, monotone, and convex. Then
    \begin{equation}\label{eq:master_inequality_phi}
        \rev(\vec{X}) \leq \srev(\vec{X}) + \int_0^\infty \phi \left( \sum_{j=1}^m \tail{F}_j^{c_j}(t) \right) \, \mathrm{d}t
    \end{equation}
    and
    \begin{equation}\label{eq:phi_convex_bounded}
        \phi(x) \leq wx + \pospart{\phi(U) - wU},
    \end{equation}
    for arbitrary $w,U \geq 0$ and $x\in[0,U]$.
\end{lemma}
It is interesting to note here that, in the special uncapped case of $c_j=0$ for
all items $j$, \eqref{eq:master_inequality_phi} recovers the tail-integral bound
of \textcite[Lemma~4.2]{cai2026improvedrevenueguaranteesselling}.

\begin{lemma}
    \label{lemma:bound_on_weighted_tails}
    Let $X_1, X_2, \dots, X_m$ be independent random variables and caps $c_j$
    chosen such that $\expect{X_j \minshort c_j} = r_j$ for every $j\in [m]$.
    Then for every $\lambda > 0$
    \begin{equation}
        \sum_{j=1}^m \int_{c_j}^\infty e^{-\lambda u} \tail{F}_j(u) \, \mathrm{d}u \leq - \frac{1}{\lambda} \ln \left( \int_1^\infty e^{-\lambda \brevshort x} \frac{1}{x^2} \, \mathrm{d}x \right) - \sum_{j=1}^m r_j \frac{1-e^{-\lambda c_j}}{\lambda c_j},
    \end{equation}
    where $r_j = \rev(X_j)$ and $\brevshort = \brev(\vec{X})$. 
\end{lemma}

\begin{proposition}
    \label{th:M-upper-bound-independent}
    For any independent random variables $X_1,X_2,\dots,X_m$ it holds that 
    \begin{equation}
        \label{eq:M-upper-bound-independent}
        \rev(\vec{X}) \leq \frac{7}{6}\srev(\vec{X})+\frac{4}{3}\brev(\vec{X})\leq \frac{5}{2} \max \sset{\srev(\vec{X}),\brev(\vec{X})}
    \end{equation}
\end{proposition}
\begin{proof}
    For convenience denote $r_j=\rev(X_j)$. We choose caps $c_j$ individually
    for each item $j\in [m]$ such that $\expect{X_j \minshort c_j} = r_j$.
    
    Recall the nonnegative, monotone, and convex function $\phi$
    from~\eqref{def:phi_function}. Applying ~\cref{lemma:tail_inequalities}
    gives us~\eqref{eq:master_inequality_phi} as the starting point
    \[\rev(\vec{X}) \leq \srev(\vec{X}) + \int_0^\infty \phi \left( \sum_{j=1}^m \tail{F}_j^{c_j}(t) \right) \, \mathrm{d}t. \]
    
    As $\phi$ is convex we may use Jensen's inequality. By this we achieve an
    item separation which allows us to use the caps in the integrals bounds by
    substitution $u=c_j+t$:
    \begin{equation}\label{eq:upper_bound_jensen}
        \int_0^\infty \phi \left( \sum_{j=1}^m \tail{F}_j^{c_j}(t) \right) \, \mathrm{d}t \leq \int_0^\infty \sum_{j=1}^m \pi_j \phi \left(  \frac{\tail{F}_j^{c_j}(t)}{\pi_j} \right) \, \mathrm{d}t = \sum_{j=1}^m \pi_j \int_{c_j}^\infty  \phi \left(  \frac{\tail{F}_j(u)}{\pi_j} \right) \, \mathrm{d}u    
    \end{equation}
    for any weights $\pi_j > 0$ with $\sum_{j=1}^m \pi_j = 1$. The Jensen
    weights allow us to consider and control the integral for each item
    separately, depending on each item's individual cap. 
    
    \begin{claim}\label{claim:mrev_independent} For caps $c_j$, such that
        $\expect{X_j \minshort c_j} = r_j$, there exist Jensen weights $\pi_j$
        such that
        \[\sum_{j=1}^m \pi_j \int_{c_j}^\infty  \phi \left(  \frac{\tail{F}_j(u)}{\pi_j} \right) \, \mathrm{d}u \leq \frac{1}{6} \srev + \frac{4}{3} \brev. \]
    \end{claim}
    The proof of the claim is moved to~\cref{appendix:proof_of_claims}. This
    yields
    \[\rev (\vec X) \leq \srev + \frac{1}{6} \srev + \frac{4}{3} \brev = \frac{7}{6} \srev + \frac{4}{3} \brev .\]
    
\end{proof}

\subsection{IID Items}
\label{sec:iid-results}
Now we turn our attention to the special case of iid items, and we
provide~\cref{th:upper-bounds-iid-unified}. We split the proof in two technical
propositions: one corresponding to the upper bound on the approximation ratio of
\mrev\ (\cref{th:M-upper-bound-iid}) and one for the grand-bundling
(\cref{th:B-upper-bound-iid}). Again, we will make use of our master revenue
inequality~\eqref{eq:rev_master_inequality}; however, this time we will
naturally instantiate it via a common cap for all items. Since both analyses
share many elements in common, we first prove~\cref{th:M-upper-bound-iid} and
then, for~\cref{th:B-upper-bound-iid} we continue from the point where the two
proofs diverge.

\begin{proposition}
    \label{th:M-upper-bound-iid}
    For any iid random variables $X_1,X_2,\dots,X_m$ it holds that
    \begin{equation}
        \label{eq:M-upper-bound-iid}
        \rev(\vec{X}) \leq 2.15 \cdot \max \sset{\srev(\vec{X}),\brev(\vec{X})}
    \end{equation}
\end{proposition}

\begin{proof}
    For convenience we denote $r_j = \rev(X_j)$, $\srevshort = \srev(\vec X)$,
    $\brevshort = \brev(\vec X)$, as well as $\mrevshort = \max
    \sset{\srev,\brev}$ whenever it supports readability. We also feel free to
    drop the items' index $j$ for the tail functions as well as the expected
    single item revenues as they are identical. We choose the common cap $c$ for
    all items such that $\expect{X_j \minshort c} = r$. 
    
    Recall from~\eqref{def:phi_function} the nonnegative, monotone, and convex
    function $\phi$. Applying \cref{lemma:tail_inequalities}
    to~\eqref{eq:common-caps} in the iid setting, we
    use~\eqref{eq:master_inequality_phi} as the starting point of our analysis:
    \[\rev(\vec{X}) \leq \srev(\vec{X}) + \int_c^\infty \phi \left( m \tail{F}(u) \right) \, \mathrm{d}u.\]
    
    We expand the integrand by a cancelling term with parameters $\alpha \geq0$
    and $\lambda>0$:
    \[
    \phi \left( m \tail{F}(u) \right) = \left[ \phi \left( m \tail{F}(u) \right) - \alpha e^{-\lambda u} m \tail{F}(u) \right] + \alpha e^{-\lambda u} m \tail{F}(u).
    \]
    Applying \eqref{eq:phi_convex_bounded} with $w=\alpha e^{-\lambda u}$ and
    any $U(u) \geq m \tail{F}(u)$ yields
    \begin{equation}\label{eq:phi_integrand_split}
        \phi \left( m \tail{F}(u) \right) \leq \alpha e^{-\lambda u} m \tail{F}(u) +  \pospart{ \phi (U(u)) - \alpha e^{-\lambda u} U(u)}.
    \end{equation}
    
    We apply \cref{lemma:bound_on_weighted_tails} to the integral of the first
    summand of the right-hand side of \eqref{eq:phi_integrand_split} 
    \[
    \int_c^\infty e^{-\lambda u} m \tail{F}(u) \, \mathrm{d}u  \leq -\frac{1}{\lambda} \ln \left( \int_1^\infty e^{-\lambda \brevshort x} \frac{1}{x^2} \, \mathrm{d}x \right) - \srevshort  \frac{1-e^{-\lambda c}}{\lambda c}.
    \]
    Combining this with the master inequality~\eqref{eq:rev_master_inequality}
    and~\eqref{eq:common-caps}, yields 
    \begin{equation}\label{eq:mixed_upper_bound_rev}
        \rev(\vec{X}) \leq \left(1-\alpha \frac{1-e^{-\lambda c}}{\lambda c} \right) \srevshort -\frac{\alpha}{\lambda} \ln \left( \int_1^\infty e^{-\lambda \brevshort x} \frac{1}{x^2} \, \mathrm{d}x \right) + \int_c^\infty \pospart{ \phi (U(u)) - \alpha e^{-\lambda u} U(u)} \, \mathrm{d}u
    \end{equation}
    for any $0\leq \alpha \leq 1$ and $\lambda > 0$. $\alpha \leq 1$ ensures
    that the coefficient of $\srevshort$ is nonnegative. It remains to specify
    appropriate bounds $m \tail{F}(u) \leq U(u)$, and to choose $\alpha$ and
    $\lambda$ with respect to any cap value $c$. These differ for the two
    approximation ratios and we now separate the final analysis.
    
    By~\eqref{eq:bound_for_tail_srev} of~\cref{lemma:bounds_on_sum_of_tails}, $U(u)
    = \mrevshort/u$ is a valid choice to upper bound $m\tail{F}(u)$. Hence, we can
    further bound the right-hand side of~\eqref{eq:mixed_upper_bound_rev}, and it
    remains to show that there exist choices of $\alpha$ and $\lambda$ such that
    this is upper-bounded by a constant factor times $\mrevshort$.
    
    \begin{claim}\label{claim:mrev_iid} For every $c>0$ there exists a choice of
        $0\leq \alpha \leq 1$ and $\lambda > 0$ such that 
        \begin{equation}\label{eq:mrev_integral_upper_bound}
            \left(1-\alpha \frac{1-e^{-\lambda c}}{\lambda c}\right) \mrevshort - \frac{\alpha}{\lambda} \ln \left( \int_1^\infty e^{-\lambda \mrevshort x} \frac{1}{x^2} \, \mathrm{d}x \right) + \int_c^\infty \pospart{ \phi \left( \frac{\mrevshort}{u} \right) - \alpha e^{-\lambda u} \frac{\mrevshort}{u}} \, \mathrm{d}u  \leq \frac{43}{20} \mrevshort.
        \end{equation}
    \end{claim}
    The proof of this claim is moved to \cref{appendix:proof_of_claims}. This
    directly shows $ \rev \leq 2.15 \, \mrevshort$. 
\end{proof}

\begin{proposition}
    \label{th:B-upper-bound-iid}
    For any iid random variables $X_1,X_2,\dots,X_m$ it holds that
    \begin{equation}
        \label{eq:B-upper-bound-iid}
        \rev(\vec{X}) \leq \frac{5}{2} \brev(\vec{X})
    \end{equation}
\end{proposition}
\begin{proof}
    For our proof here, we continue from the proof of~\cref{th:M-upper-bound-iid},
    and in particular from \eqref{eq:mixed_upper_bound_rev}. For the analysis of the
    bundling in the iid setting, we will also make use of \citeauthor{Kupfer2016}'s
    \parencite{Kupfer2016} bound (see p.~\pageref{page:grand-bundling-related-work})
$\srevshort \leq \bar{\kappa} \brevshort$ for $\bar{\kappa} =
    \frac{1787}{1000}$. We further use from~\cref{lemma:bounds_on_sum_of_tails} the
    upper bounds of $m \tail{F}(u) \leq \srevshort/u \leq \bar{\kappa} \brevshort
    /u$, by~\eqref{eq:bound_for_tail_srev}, and $m \tail{F}(u) \leq - \ln
    (1-\brevshort / u)$, by~\eqref{eq:bound_for_tail_brev}. Hence, $U(u) = - \ln
    (1-\brevshort / u)$ for any $u>\brevshort$ and $U(u) = \bar{\kappa} \brevshort
    /u$ for all $u>0$ are valid choices.
    
    Applying these bounds to~\eqref{eq:mixed_upper_bound_rev}, again with
    $\alpha \leq 1$ and $\lambda>0$, it remains to show the following claim. 
    
    \begin{claim}\label{claim:brev_iid}
        For every cap $c>0$ there exists a choice of parameters $0\leq \alpha \leq 1$ and $\lambda > 0$ such that 
        \begin{equation}\label{eq:brev_integral_upper_bound}
            \left(1-\alpha \frac{1-e^{-\lambda c}}{\lambda c} \right) \bar{\kappa} \brevshort - \frac{\alpha}{\lambda} \ln \left( \int_1^\infty e^{-\lambda \brevshort x} \frac{1}{x^2} \, \mathrm{d}x \right) + \int_c^\infty \pospart{ \phi (U(u)) - \alpha e^{-\lambda u} U(u)} \mathrm{d}u \leq \frac{5}{2} \, \brevshort.
        \end{equation}
    \end{claim}
    The proof of this claim is moved to \cref{appendix:proof_of_claims} and $\rev \leq 2.5 \, \brevshort$ is immediate.
\end{proof}

\section{Tight Analysis for Selling Two IID Items Separately}
\label{sec:selling_separately_2_iid}

In this section we return to the discrete auction setting and prove a tight
bound for the approximation ratio of selling two identically and independently
distributed items compared to the optimal selling mechanism. We do this by
providing a carefully chosen dual solution of Lagrangian multipliers. Here, the
discretization according to integer $K$ is of particular interest, as we show
the tight upper bound $(1+\omega_K)$ depending on a fixed choice of $K$. For
\emph{every} choice of discretization the approximation's upper bound is
strictly smaller than the known lower bound result from Hart and Nisan
\cite{Hart:2017aa}, yet, these bounds hold only for the restricted discretized
setting. In the limit of the discretization $(1+\omega_K)$ converges to
$(1+W(1/e))$, where $W$ denotes the Lambert-W function (see,
e.g.,~\cite{Corless1996}) closing the open gap of $[1.278, 1.368]$
\cite{Hart:2017aa} at the lower interval bound. The exact definition of
$\omega_K$ can be found in~\eqref{def:discrete_lambert_approximate} and a plot
of how it approaches $W(1/e)$ is provided in~\cref{fig:omega_convergence_plot}.

\subsection{Preliminaries and Anchored Distribution}
\label{sec:2-iid-prelims}

We define our flow for a certain class of distributions where the smallest
support point with positive probability mass is the optimal single-item selling
price. 

\begin{definition}[Anchored Distributions]
    \label{def:distribution_achored_in_r}
    Let $K\geq 1$ and $r \in [K]$. A distribution supported on $\ssets{\frac{r}{K},\frac{r+1}{K},\dots,1} \subseteq I_K$ is called \emph{anchored} in $r$, if
    \begin{equation}\label{def:anchored_distribution_equation}
        k\sum_{l=k}^K f(l) \leq r \qquad \text{for all}\;\; k= r,r+1,\dots,K.
    \end{equation}
\end{definition}

Note that condition~\eqref{def:anchored_distribution_equation} implies that for
any distribution anchored in $r$, $r/{K}$ is an optimal single-item posted
price. Therefore, the expected revenue of selling two iid $r$-anchored items
separately is ${2r}/{K}$.

While the lower bound of \cite{Hart:2017aa} is attained by the (continuous)
equal-revenue distribution\footnote{The distribution supported on $[1,\infty)$
with CDF $F(x) = 1-1/x$ and density function $f(x) = 1/x^2$. It has the property
that selling at any posted price $x\geq 1$ it gives the same revenue, that is,
$\rev(F)=1$.} we obtain the tight bound by a truncated discrete version of this
distribution. For each $K$ we define distributions with similar properties and
will choose the candidate which maximizes the approximation ratio.

\begin{definition}[Discrete Equal-Revenue Distribution]
    \label{def:distribution_equal_revenue_r}
    Let $K\geq 1$. For $r\in[K]$, let $\er_{K}(r)$ denote the \emph{discrete
    equal-revenue} distribution with probability mass function
    \begin{equation}
        f(k) \coloneq 
        \begin{cases}
            0, & k< r,\\
            \dfrac{r}{k(k+1)}, & r\leq k< K,\\
            \dfrac{r}{K}, & k=K.
        \end{cases}
    \end{equation}
\end{definition}
Observe that every discrete equal-revenue distribution $\er_K (r)$ is anchored
in $r$ which implies $\srev(\er_K (r) \times \er_K (r))=2r/{K}$. Furthermore,
the tail probability can be computed exactly as 
\[\sum_{k=s}^K f(k) = \frac{r}{s} \qquad \text{for all}\;\; s= r,\dots,K.\]
Hence, any posted price ${s}/{K}$ with $s\geq r$ yields the single item revenue ${r}/{K}$.

Next we provide a mathematical relation of the discrete equal-revenue to general
anchored distributions which highlights its role within this class; its proof
can be found
in~\cref{append:lemma:first_order_stochastic_dominance_er_anchored}. For ease of
notation for the remainder of this section we will drop the word
\textit{discrete} for equal-revenue distributions whenever it is obvious from
the context that we refer to discrete distributions.

\begin{lemma}\label{lemma:first_order_stochastic_dominance_er_anchored} Let
    $K\geq 1$ and $r\in[K]$. Then \emph{any} distribution anchored in $r$ is
    (first-order) stochastically dominated\footnote{We say that a distribution
    $F$ \emph{stochastically dominates} (in the usual first order) a
    distribution $\hat{F}$, if $F(x)\leq \hat{F}(x)$ for any $x\in \R$; see,
    e.g., \textcite[Sec.~6.D]{Mas-Colell1995a}.} by the discrete equal-revenue
    distribution $\er_K(r)$. 
\end{lemma}

Now we define the quantity $\omega_K$ and its properties which connect our
results derived in the discrete setting, in the limit, to the continuous. Their
proof can be found in~\cref{append:lemma:omega_K_properties}.
\Cref{fig:omega_convergence_plot} illustrates the limiting behaviour of this
quantity for different values of $K$.
\begin{figure}[t]
    \centering
    \caption{Values of $\omega_K$ for $K=1,\dots,100$. The dashed horizontal
    line indicates the limiting value $W(1/e)\approx 0.27846$, approached from
    below as $K$ increases (see \cref{lemma:omega_K_properties}). For
    $K=1,2,3,4,\dots$ the values of $\omega_K$ are
    $0,\frac{1}{8},\frac{1}{6},\frac{13}{72},\dots$.}
    \begin{tikzpicture}
\begin{axis}[
    width=13cm,
    height=8cm,
    xmin=1, xmax=100,
    ymin=0, ymax=0.285,
    xlabel={$K$},
    ylabel={$\omega_K$},
    ylabel style={
        rotate=-90,
        yshift=-8pt
    },
    scaled y ticks=false,
    yticklabel style={
        /pgf/number format/fixed,
        /pgf/number format/precision=3
    },
    xtick={1,10,20,30,40,50,60,70,80,90,100},
    ymajorgrids=true,
    xmajorgrids=true,
    grid style={gray!20},
    legend style={
        at={(0.98,0.03)},
        anchor=south east,
        draw=none,
        fill=white
    },
    legend cell align={left},
]

\addplot[
    red,
    thick,
    mark=*,
    mark size=1.2pt,
] coordinates {
(1,0.000000000000)
(2,0.125000000000)
(3,0.166666666667)
(4,0.180555555556)
(5,0.195833333333)
(6,0.211111111111)
(7,0.218571428571)
(8,0.221428571429)
(9,0.229761904762)
(10,0.234778911565)
(11,0.237308802309)
(12,0.239971139971)
(13,0.243733427067)
(14,0.245942945943)
(15,0.246979131979)
(16,0.249408104522)
(17,0.251351997676)
(18,0.252453071846)
(19,0.253324536264)
(20,0.255053257143)
(21,0.256148293878)
(22,0.256719982804)
(23,0.257743244439)
(24,0.258803351384)
(25,0.259445674069)
(26,0.259785731053)
(27,0.260799913421)
(28,0.261473274834)
(29,0.261851927312)
(30,0.262354247314)
(31,0.263036401786)
(32,0.263470058886)
(33,0.263685342929)
(34,0.264275941439)
(35,0.264741992402)
(36,0.265019349538)
(37,0.265281585069)
(38,0.265765427052)
(39,0.266084757542)
(40,0.266256204512)
(41,0.266605023251)
(42,0.266952464386)
(43,0.267168844798)
(44,0.267305134423)
(45,0.267670960770)
(46,0.267919923018)
(47,0.268062149603)
(48,0.268274410778)
(49,0.268546908824)
(50,0.268723110227)
(51,0.268810787404)
(52,0.269076828853)
(53,0.269278887391)
(54,0.269400300110)
(55,0.269529604155)
(56,0.269751312121)
(57,0.269899291071)
(58,0.269978794776)
(59,0.270156763662)
(60,0.270325694042)
(61,0.270431552657)
(62,0.270507735475)
(63,0.270693181209)
(64,0.270820371797)
(65,0.270893023049)
(66,0.271012327498)
(67,0.271156807522)
(68,0.271250613854)
(69,0.271296812786)
(70,0.271449912197)
(71,0.271561219030)
(72,0.271628063747)
(73,0.271706859556)
(74,0.271832660850)
(75,0.271916860223)
(76,0.271961747256)
(77,0.272071203373)
(78,0.272170016033)
(79,0.272231888549)
(80,0.272281910200)
(81,0.272393040743)
(82,0.272469406443)
(83,0.272512761838)
(84,0.272590427192)
(85,0.272679177319)
(86,0.272736750203)
(87,0.272765867824)
(88,0.272865210291)
(89,0.272935068257)
(90,0.272976815753)
(91,0.273030824900)
(92,0.273111311209)
(93,0.273165132719)
(94,0.273193484884)
(95,0.273268476798)
(96,0.273332843921)
(97,0.273372985293)
(98,0.273409079685)
(99,0.273482667210)
(100,0.273533189169)
};
\addlegendentry{$\omega_K$}

\addplot[
    blue,
    dashed,
    thick,
] coordinates {
    (1,0.278464542761)
    (100,0.278464542761)
};
\addlegendentry{$W(1/e)\approx 0.27846$}

\end{axis}
\end{tikzpicture}
    \label{fig:omega_convergence_plot}
\end{figure}
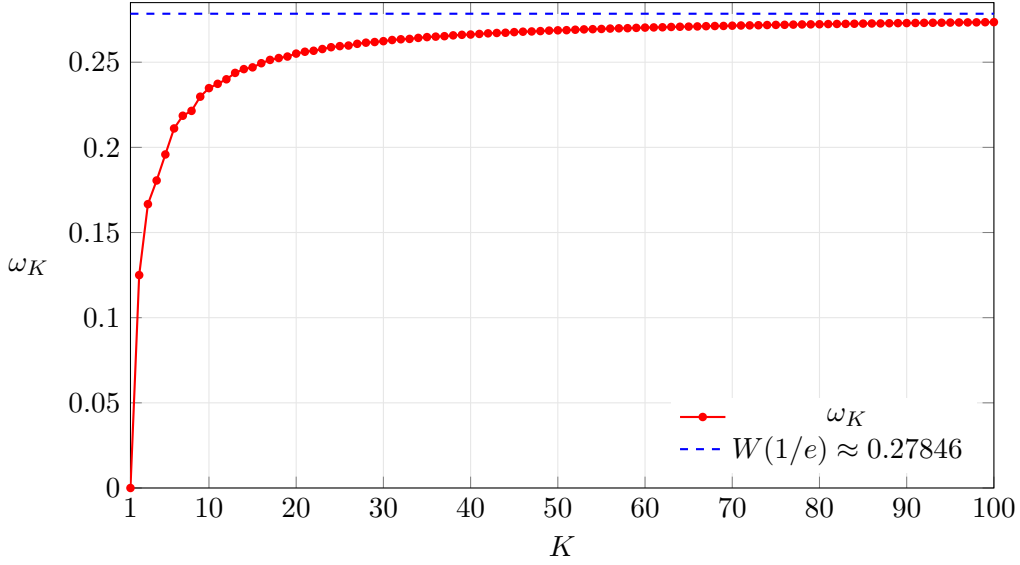

\begin{lemma}\label{lemma:omega_K_properties}
    For any positive integer $K$ we define the parameter:
    \begin{equation}\label{def:discrete_lambert_approximate}
        \omega_K \coloneq \max_{r\in[K]} \frac{r(H_K - H_r)}{K+r+1}.
    \end{equation}
    Then, the following properties hold:
    \begin{align}
        &\frac{s(H_k-H_s)}{k+s+1} \leq\omega_K < W(1/e), &&\text{for all}\;\;1\le s\le k\le K, \label{eq:omega_K_bounds-new}\\
        &\lim_{K\to\infty}\omega_K = W(1/e) \approx 0.27846, \label{eq:omega_K_convergence-new}
    \end{align}
    where $W$ denotes the Lambert-W function.\footnote{For $z\geq 0$, the value
    $W(z)=w$ of the Lambert-W function is defined as the (unique real) solution
    to equation $we^w=z$; the interested reader is referred
    to~\textcite{Corless1996} for more details and analytic properties of the
    function.} A plot of $\omega_K$ is provided
    in~\cref{fig:omega_convergence_plot}.
\end{lemma}

In the following section we define Lagrangian multipliers for
\eqref{eq:Lagrangian} in the setting of two iid items. We introduce a flow for
this more restricted setting of two iid items with an anchor-distribution
marginal. Subsequently, we prove in \cref{lemma:two_iid_anchored_upper_bound}
that this restricted setting of distributions has an approximation rate of at
most $(1+\omega_K)$ and that this bound is tight
(see~\cref{lemma:equal_revenue_tightness}). Furthermore, in
\cref{prop:two_items_iid_SREV} we show, that the restriction is without loss of
generality, i.e., the approximation rate $(1+\omega_K)$ holds for the two iid
item setting for any distribution supported in $I_K$.

\subsection{A Tight Dual for Anchored Distributions}
\label{sec:tight_flow_anchored_distribution}

For the construction of Lagrangian multipliers we define four layers of flows.
Whenever multiple layers affect the same arc  their values are simply added. To
this end, assume that all multipliers are initially zero and all contributions
are added to the initial values. From now on, we consider an arbitrary but fixed
distribution which is anchored in $r$. Therefore, $\srev(F_r \times F_r) =
\frac{2r}{K}$.

First, define the $\theta$ variables as $\theta(r,r)=\frac{K}{2r}$ and
$\theta(k_1,k_2)=0$ otherwise. Observe that with this choice the coefficient of
$\rho$ directly cancels from the Lagrangian. From the relaxed $\theta(r,r)$
constraint every $p(k_1,k_2)$ variable receives an initial flow of $\frac{K}{2r}
f(k_1) f(k_2)$. For this initial flow mass we now exceed the flow-routing rule
\eqref{eq:flow_rule_definition} by choosing $w_1(k_1,k_2)=1$ for $k_1>r$ and
$k_2\geq r$ as well as $w_2(k_1,k_2)=1$ for $k_1\geq r$ and $k_2>r$, both scaled
by $\theta(r,r)$. This defines the first flow layer. Throughout the proof, we
will call all grid points $(r,\cdot)$ or $(\cdot,r)$ boundary points, and all
other grid points, interior points.

\paragraph{Rectangular downwards flow}
For $k_1=r+1, \dots, K$ and $k_2=r,\dots,K$,
\begin{equation}
    \lambda_{(k_1,k_2)\rightarrow(k_1-1,k_2)}^\text{rect} = \frac{K}{2r} f(k_2) \sum_{l=k_1}^K f(l),
\end{equation}
and symmetrically for $k_1=r, \dots, K$ and $k_2=r+1,\dots,K$, 
\begin{equation}
    \lambda_{(k_1,k_2)\rightarrow(k_1,k_2-1)}^\text{rect} = \frac{K}{2r} f(k_1) \sum_{l=k_2}^K f(l).
\end{equation}
This outgoing net flow at all points with $k_1>r$ and $k_2>r$ is exactly
twice the initial mass, hence, the respective coefficients of the payment
variables change from $\frac{K}{2r} f(k_1) f(k_2)$ to $-\frac{K}{2r} f(k_1)
f(k_2)$. A visualization of this layer is provided in
\cref{fig:rectangular_flows}. The next flow layer defines circular flows
which do not affect the payments' but only the allocations' coefficients.
Therefore, we also define $\psi$ variables to control this.
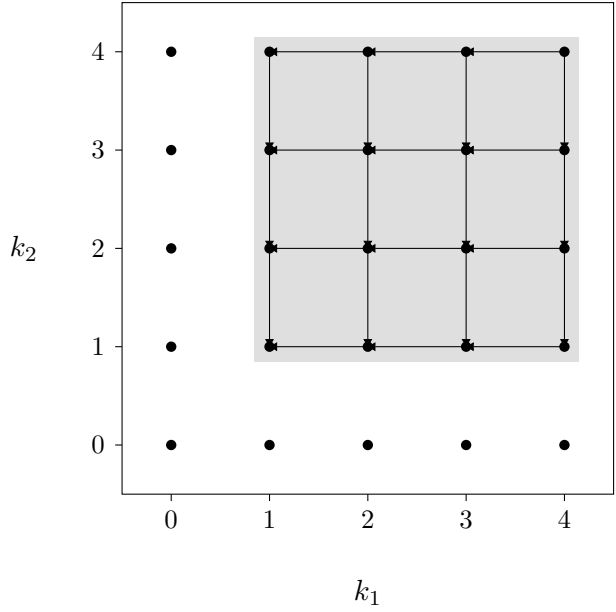
\begin{figure}[t]
    \centering
    \caption{The first flow layer of rectangular flows for $K=4$ and $r=1$. These flows cause an imbalance in every interior point as $w_1(k_1,k_2)=w_2(k_1,k_2)=1$.}
    \label{fig:rectangular_flows}
    \begin{tikzpicture}[scale=1.3]

    \def\K{4}
    \def\r{1}

    \fill[gray!25]
        ({\r-.15},{\r-.15})
        rectangle
        ({\K+.15},{\K+.15});

    \foreach \i in {2,...,4}
        \foreach \j in {1,...,4}
            \draw[myarrow] (\i,\j) -- ({\i-1},\j);

    \foreach \i in {1,...,4}
        \foreach \j in {2,...,4}
            \draw[myarrow] (\i,\j) -- (\i,{\j-1});

    \foreach \i in {0,...,4}
        \foreach \j in {0,...,4}
            \fill (\i,\j) circle (1.5pt);

    \draw (-.5,-.5) rectangle (4.5,4.5);

    \foreach \i in {0,...,4} {
        \draw (\i,-.5) -- (\i,-.57);
        \node[below] at (\i,-.57) {\small \i};
    }

    \foreach \j in {0,...,4} {
        \draw (-.5,\j) -- (-.57,\j);
        \node[left] at (-.57,\j) {\small \j};
    }

    \node at (2,-1.5) {$k_1$};
    \node at (-1.5,2) {$k_2$};

\end{tikzpicture}
\end{figure}

\paragraph{Circular flow}
For $k_1=r+1, \dots, K$ and $k_2=r,\dots,K$, add the circular flows
\begin{align*}
    \lambda_{(k_1,k_2)\rightarrow(k_1-1,k_2)}^\text{circ} = \frac{K}{2r} f({k_2}) (r - k_1\sum_{l=k_1}^K f(l)), \\
    \lambda_{(k_1-1,k_2)\rightarrow(k_1,k_2)}^\text{circ} = \frac{K}{2r} f({k_2}) (r - k_1\sum_{l=k_1}^K f(l)),
\end{align*}
and symmetrically for $k_1=r, \dots, K$ and $k_2=r+1,\dots,K$, add the flows
\begin{align*}
    \lambda_{(k_1,k_2)\rightarrow(k_1,k_2-1)}^\text{circ} = \frac{K}{2r} f({k_1}) (r - k_2\sum_{l=k_2}^K f(l)), \\
    \lambda_{(k_1,k_2-1)\rightarrow(k_1,k_2)}^\text{circ} = \frac{K}{2r} f({k_1}) (r - k_2\sum_{l=k_2}^K f(l)).
\end{align*}
We also set for $k_2=r,\dots,K$
\[ \psi_1^\text{circ}(K,k_2) = \frac{1}{2} f(k_2),\]
as well as for $k_1=r,\dots,K$
\[ \psi_2^\text{circ}(k_1,K) = \frac{1}{2} f(k_1).\] Note, that because of
the choice of the $\psi$ variables for both items, their total mass at the
border sums up to one. \Cref{fig:circular_flows} visualizes the circular
flows as well as the mass induced by $\psi$ which accumulates at the top
borders. With the next flow layer, we manage to cancel all the interior
payment variables' coefficients. To do so, we send a total flow of the
missing mass from two border points, i.e., a grid point where at least one
index equals $r$, to an interior one. For each such point in the lower
triangular part we choose the two border grid points which form the $45$
degree diagonal the interior point lies on. In the upper triangular part no
such border points exist. There, the flow comes from $(K,r)$ and $(r,K)$.
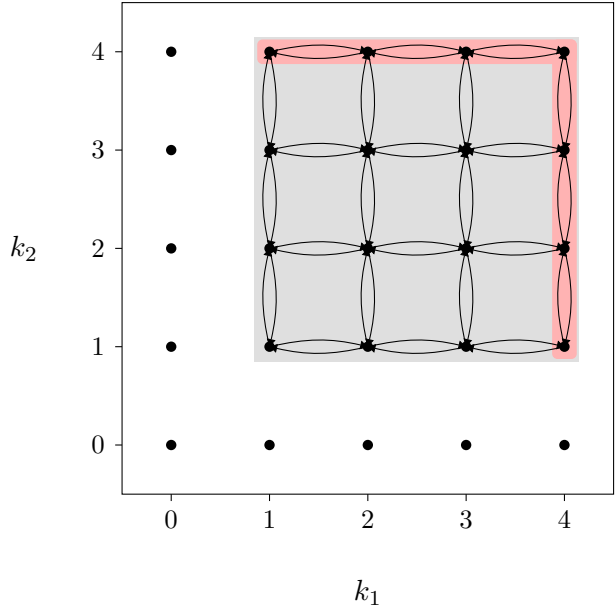
\begin{figure}[t]
    \centering
    \caption{The second flow layer of circular flows for $K=4$ and $r=1$. These flows cancel with respect to payment variables but affect the allocation variables by shifting a total constant mass of one to the upper boundaries.}
    \begin{tikzpicture}[scale=1.3]

    \def\K{4}
    \def\r{1}

    \fill[gray!25]
        ({\r-.15},{\r-.15})
        rectangle
        ({\K+.15},{\K+.15});

    \fill[red!30,rounded corners=2pt]
        (3.875,.875)
        rectangle
        (4.125,4.125);

    \fill[red!30,rounded corners=2pt]
        (.875,3.875)
        rectangle
        (4.125,4.125);

    \foreach \i in {2,...,4}
        \foreach \j in {1,...,4} {
            \draw[myarrow] (\i,\j)
                to[bend left=13] ({\i-1},\j);
            \draw[myarrow] ({\i-1},\j)
                to[bend left=13] (\i,\j);
        }

    \foreach \i in {1,...,4}
        \foreach \j in {2,...,4} {
            \draw[myarrow] (\i,\j)
                to[bend left=13] (\i,{\j-1});
            \draw[myarrow] (\i,{\j-1})
                to[bend left=13] (\i,\j);
        }

    \foreach \i in {0,...,4}
        \foreach \j in {0,...,4}
            \fill (\i,\j) circle (1.5pt);

    \draw (-.5,-.5) rectangle (4.5,4.5);

    \foreach \i in {0,...,4} {
        \draw (\i,-.5) -- (\i,-.57);
        \node[below] at (\i,-.57) {\small \i};
    }

    \foreach \j in {0,...,4} {
        \draw (-.5,\j) -- (-.57,\j);
        \node[left] at (-.57,\j) {\small \j};
    }

    \node at (2,-1.5) {$k_1$};
    \node at (-1.5,2) {$k_2$};

\end{tikzpicture}
    \label{fig:circular_flows}
\end{figure}

\paragraph{Diagonal flow}
For any inner grid point of the lower triangular part, i.e., for any $(k_1,k_2)$ with $k_1+k_2 \le K+r$ and $k_1,k_2 > r$, set
\begin{align*}
    \lambda_{(k_1+k_2-r,r)\rightarrow(k_1,k_2)}^\text{diag} = \frac{K}{2r} f(k_1) f(k_2) \frac{k_1-r}{k_1+k_2-2r},\\
    \lambda_{(r,k_1+k_2-r)\rightarrow(k_1,k_2)}^\text{diag} = \frac{K}{2r} f(k_1) f(k_2) \frac{k_2-r}{k_1+k_2-2r}.
\end{align*}
For any grid point in the upper triangular part, i.e., for any $(k_1,k_2)$ with $k_1+k_2 > K+r$, set
\begin{align*}
    \lambda_{(K,r)\rightarrow(k_1,k_2)}^\text{diag} = \frac{K}{2r} f(k_1) f(k_2) \frac{k_1-r}{k_1+k_2-2r}, \\
    \lambda_{(r,K)\rightarrow(k_1,k_2)}^\text{diag} = \frac{K}{2r} f(k_1) f(k_2) \frac{k_2-r}{k_1+k_2-2r}.
\end{align*}

We also add for these grid points, i.e., for any $(k_1,k_2)$ with $k_1+k_2 > K+r$,
\[ \psi_1^\text{diag}(k_1,k_2) = \frac{K}{2r} f(k_1) f(k_2) \frac{k_1-r}{k_1+k_2-2r} \frac{k_1+k_2-r-K}{K},\]
as well as
\[ \psi_2^\text{diag}(k_1,k_2) = \frac{K}{2r} f(k_1) f(k_2)
\frac{k_2-r}{k_1+k_2-2r} \frac{k_1+k_2-r-K}{K}.\] We visualize the diagonal
flows in \cref{fig:diagonal_flows}. Note that now every interior point has
exactly a net outflow of the initial flow mass, hence, all payment
variables' coefficients cancel. In addition, the definition of the $\psi$
variables ensures that also all allocation variables' coefficients become
zero in the upper triangle. There the flow comes from boundary points that
are not on the same $45$ degree diagonal, but from the corner points which
are \textit{too close}. This leads to an imbalance on these boundary points
which is resolved by the final flow layer.
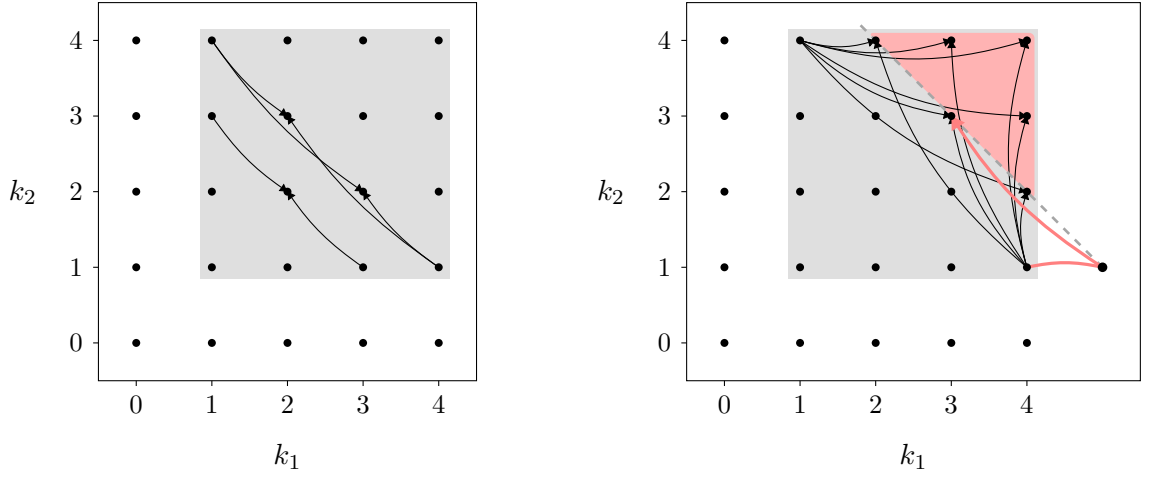
\begin{figure}[t]
    \caption{The third flow layer of diagonal flows for $K=4$ and $r=1$. These flows repair the imbalance from the first layer in every interior point.}
    \centering
    \begin{subfigure}[t]{0.48\textwidth}
        \centering
        \begin{tikzpicture}[scale=1.0]

    \def\K{4}
    \def\r{1}

    \fill[gray!25]
        ({\r-.15},{\r-.15})
        rectangle
        ({\K+.15},{\K+.15});

    \foreach \i in {2,...,4}
        \foreach \j in {2,...,4} {
            \pgfmathtruncatemacro{\sumij}{\i+\j}
            \ifnum\sumij<6
                \pgfmathtruncatemacro{\s}{\i+\j-1}

                \draw[myarrow]
                    (\s,1)
                    to[bend left=10]
                    (\i,\j);

                \draw[myarrow]
                    (1,\s)
                    to[bend right=10]
                    (\i,\j);
            \fi
        }

    \foreach \i in {0,...,4}
        \foreach \j in {0,...,4}
            \fill (\i,\j) circle (1.5pt);

    \draw (-.5,-.5) rectangle (4.5,4.5);

    \foreach \i in {0,...,4} {
        \draw (\i,-.5) -- (\i,-.57);
        \node[below] at (\i,-.57) {\small \i};
    }

    \foreach \j in {0,...,4} {
        \draw (-.5,\j) -- (-.57,\j);
        \node[left] at (-.57,\j) {\small \j};
    }

    \node at (2,-1.5) {$k_1$};
    \node at (-1.5,2) {$k_2$};

\end{tikzpicture}
        \caption{In the lower triangular part every interior point can be expressed as a convex combination of lower boundary points.}
        \label{fig:lower_diagonal_flows}
    \end{subfigure}
    \hfill
    \begin{subfigure}[t]{0.48\textwidth}
        \centering
        \begin{tikzpicture}[scale=1.0]

    \def\K{4}
    \def\r{1}
    \def\t{1}

    \fill[gray!25]
        ({\r-.15},{\r-.15})
        rectangle
        ({\K+.15},{\K+.15});

    \fill[red!30,rounded corners=2pt]
        (4.1,1.9)
        -- (1.9,4.1)
        -- (4.1,4.1)
        -- cycle;

    \foreach \i in {1,...,4}
        \foreach \j in {1,...,4} {
            \pgfmathtruncatemacro{\sumij}{\i+\j}
            \ifnum\sumij>5
                \draw[myarrow]
                    (4,1)
                    to[bend left=16]
                    (\i,\j);

                \draw[myarrow]
                    (1,4)
                    to[bend right=16]
                    (\i,\j);
            \fi
        }

    \coordinate (A) at (2,4);
    \coordinate (B) at (3,3);
    \coordinate (C) at (5,1);

    \draw[dashed,gray!70,line width=1pt]
        (1.8,4.2) -- (C);

    \draw[red!50,line width=1.2pt]
        (4,1)
        to[bend left=12]
        (C);

    \draw[red!50,line width=1.2pt,myarrow]
        (C)
        to[bend left=12]
        (B);

    \foreach \i in {0,...,4}
        \foreach \j in {0,...,4}
            \fill (\i,\j) circle (1.5pt);

    \fill (C) circle (1.8pt);

    \draw (-.5,-.5) rectangle (5.5,4.5);

    \foreach \i in {0,...,4} {
        \draw (\i,-.5) -- (\i,-.57);
        \node[below] at (\i,-.57) {\small \i};
    }

    \foreach \j in {0,...,4} {
        \draw (-.5,\j) -- (-.57,\j);
        \node[left] at (-.57,\j) {\small \j};
    }

    \node at (2.5,-1.5) {$k_1$};
    \node at (-1.5,2) {$k_2$};

\end{tikzpicture}
        \caption{In the upper triangular part the expression as a convex combination is not possible. The detour from the corner points is captured by $\psi$ multipliers (red).}
        \label{fig:upper_diagonal_flows}
    \end{subfigure}
    \label{fig:diagonal_flows}
\end{figure}

\paragraph{Boundary flow}
The boundary flow requires more involved expressions. For this we define the
boundary imbalance as
\begin{equation}\label{def:boundary_correction_term}
    B(f;s) \coloneq \frac{1}{K} \left[ \sum_{k_1+k_2\geq r+s+1} (k_1+k_2 - r-s) f(k_1)f(k_2) - 2 \sum_{k=s+1}^K (k-s) f(k)  \right],
\end{equation}
where the second sum conventionally equals zero for $s=K$. Recall the
quantity $\omega_K$ from \eqref{def:discrete_lambert_approximate} and set
for every interior boundary point $k_1=r+1,\ldots,K-1$,
\[
\lambda_{(k_1+1,r)\rightarrow(k_1,r)}^{\mathrm{bdry}}
=\lambda_{(k_1-1,r)\rightarrow(k_1,r)}^{\mathrm{bdry}}
=\frac{K}{2r}\frac{K}{2} \left( \frac{2r}{K}\omega_K-B(f;k_1) \right).
\]
At the lower endpoint we add
\[\lambda_{(r+1,r)\rightarrow(r,r)}^{\mathrm{bdry}}
= \frac{K}{2r}\frac{K}{2}\left( \frac{2r}{K}\omega_K-B(f;r) \right),\]
as well as
\[\alpha_1^{\mathrm{bdry}}(r,r)= \frac{K}{2r}\frac{1}{2}
\left( \frac{2r}{K}\omega_K-B(f;r)\right),\]
and at the upper boundary endpoint define
\[\lambda_{(K-1,r)\rightarrow(K,r)}^{\mathrm{bdry}}=
\frac{K}{2r} \frac{K}{2}\left(\frac{2r}{K}\omega_K-B(f;K)\right),\]
as well as
\[\psi_1^{\mathrm{bdry}}(K,r)=\frac{K}{2r}\frac{1}{2}\left(\frac{2r}{K}\omega_K-B(f;K)\right).\]
We use the symmetric analogue construction on the boundary $(r,k_2)$ for
$k_2=r,\dots, K$, with $\alpha_2^{\mathrm{bdry}}(r,r)$ at the lower endpoint and
$\psi_2^{\mathrm{bdry}}(r,K)$ at the upper endpoint. Finally, we set
\[\mu(r,r)=\frac{K}{2r},\]
and leave $\mu(k_1,k_2)=0$ for all other $(k_1,k_2)\neq(r,r)$. This multiplier
cancels the remaining individual-rationality term at $(r,r)$.

This boundary flow which is visualized in \cref{fig:boundary_flows} now adds
exactly the missing amount required to replace the distribution-dependent
coefficient $B(f;k)$ by the common coefficient $\frac{2r}{K}\omega_K$ at every
boundary position $(r,k)$ and $(k,r)$.
\begin{figure}[t]
    \centering
    \caption{The fourth flow layer of boundary flows for $K=4$ and $r=1$. These flows handle the remaining imbalance on the lower boundaries. While the $\psi$ multipliers repair the imbalance at the top boundary corners (red), $\alpha$ nicely captures all outgoing mass from the grid (blue).}
    \begin{tikzpicture}[scale=1.3]

    \def\K{4}
    \def\r{1}

    \fill[gray!20]
        ({\r-.15},{\r-.15})
        rectangle
        ({\K+.15},{\K+.15});

    \fill[red!30] (4,1) circle (.15);
    \fill[red!30] (1,4) circle (.15);
    \fill[blue!25] (1,1) circle (.15);

    \foreach \i in {1,...,3} {
        \pgfmathtruncatemacro{\ip}{\i+1}

        \draw[myarrow]
            (\i,1)
            to[bend left=16]
            (\ip,1);

        \draw[myarrow]
            (\ip,1)
            to[bend left=16]
            (\i,1);
    }

    \foreach \j in {1,...,3} {
        \pgfmathtruncatemacro{\jp}{\j+1}

        \draw[myarrow]
            (1,\j)
            to[bend left=16]
            (1,\jp);

        \draw[myarrow]
            (1,\jp)
            to[bend left=16]
            (1,\j);
    }

    \foreach \i in {0,...,4}
        \foreach \j in {0,...,4}
            \fill (\i,\j) circle (1.5pt);

    \draw (-.5,-.5) rectangle (4.5,4.5);

    \foreach \i in {0,...,4} {
        \draw (\i,-.5) -- (\i,-.57);
        \node[below] at (\i,-.57) {\small \i};
    }

    \foreach \j in {0,...,4} {
        \draw (-.5,\j) -- (-.57,\j);
        \node[left] at (-.57,\j) {\small \j};
    }

    \node at (2,-1.5) {$k_1$};
    \node at (-1.5,2) {$k_2$};

\end{tikzpicture}
    \label{fig:boundary_flows}
\end{figure}
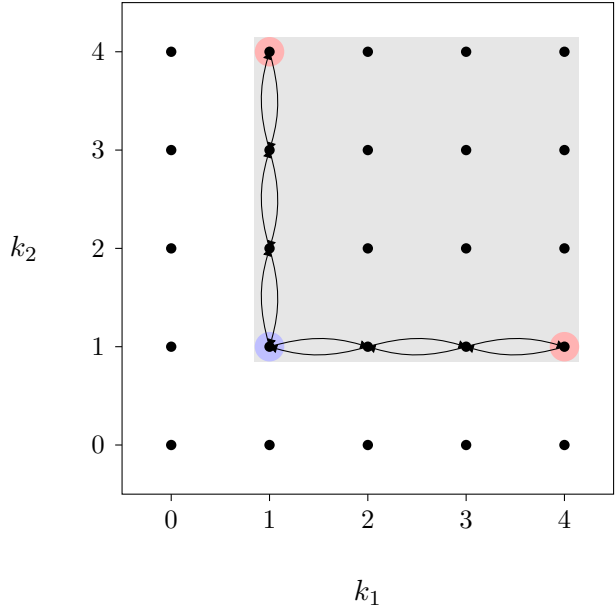

In the next lemma we prove an upper bound of the approximation ratio of selling
two items, both independently distributed according to an anchored distribution
supported on $I_K$. We first verify that indeed, all multipliers are
nonnegative, i.e., the proposed solution is feasible, and then quantify the
objective as $(1+\omega_K)$. Its proof is the most technically demanding part of
our paper, and we defer it to~\cref{append:lemma:two_iid_anchored_upper_bound}.

\begin{lemma}\label{lemma:two_iid_anchored_upper_bound} Let $K\geq1$, $r\in[K]$,
    and the values of two iid items be distributed according to an anchored
    distribution $F_r$. Then
    \begin{equation}\label{eq:two_iid_anchored_upper_bound}
        \rev (F_r \times F_r) \leq \left( 1 + \omega_K \right) \srev(F_r \times F_r).
    \end{equation}
\end{lemma}

Before extending the upper bound from anchored distributions to arbitrary
distributions, we show that the factor $(1+\omega_K)$ is already tight within
the class of anchored distributions, which of course provides a lower bound for
arbitrary distributions as well. 

\begin{proposition}\label{lemma:equal_revenue_tightness} Fix an arbitrary $K\geq
    1$ and let $r^\star \in\argmax_{r\in[K]} \frac{r(H_K-H_{r})}{K+r+1}$. Then
    the equal-revenue distribution $\er_K(r^\star)$ attains the upper bound
    in~\cref{lemma:two_iid_anchored_upper_bound}. That is, for $F^\star\coloneqq
    \er_K(r^\star)$ it holds that
    \[ \rev\bigl(F^\star \times F^\star\bigr) \geq  (1+\omega_K) \, \srev\bigl(F^\star \times F^\star\bigr),\]
    where $\omega_K$ is defined in~\eqref{def:discrete_lambert_approximate}.
\end{proposition}

\begin{proof}
    See~\cref{append:lemma:equal_revenue_tightness}.
\end{proof}

\subsection{Tightness for General Distributions}
In this section we prove that the upper bound
of~\cref{lemma:two_iid_anchored_upper_bound} generalizes to arbitrary (two-item
iid) distributions. Making use of the discrete setting we are in, allows us to
use induction over the number of support points with positive probability mass.
The proof of the proposition below can be found
in~\cref{append:prop:two_items_iid_SREV}.

\begin{proposition}\label{prop:two_items_iid_SREV} If $K\geq1$ and the values of
    two iid items are distributed according to an arbitrary $F$ supported on
    $I_K$, then
    \begin{equation}
        \label{eq:two_items_iid_SREV}
        \rev (F \times F) \leq \left( 1 + \omega_K \right) \srev(F \times F)
    \end{equation}
    where $\omega_K$ is defined in~\eqref{def:discrete_lambert_approximate}.
\end{proposition}

\section{Conclusion}
\label{sec:conclusion}

In this paper we studied the worst-case revenue guarantees of separate selling,
grand bundling, and the better of the two. Using duality-theory techniques, we
were able to derive a novel, cap-parameterized optimal-revenue
inequality~\eqref{eq:rev_master_inequality}. Using various probabilistic
arguments involving the tails of these capped item values, our master inequality
yields improved approximation bounds, including factors of $2.5$ for general
independent values and $2.15$ for iid items for the combined benchmark. For the
special case of two iid items, tight dual certificates determine the exact
separate-selling ratio for every discretization, with a worst-case limit of
$1+W(1/e)\approx 1.2785$.

Determining the \emph{exact} approximation ratio of $\max\{\srev,\brev\}$ for
general independent items is, perhaps, the most prominent open question
remaining: our new upper and lower bounds place this ratio within the interval
$[2.19,2.5]$. Additionally, we believe that determining this approximation ratio
for the special case of iid items is also a very interesting question; this
would entail improving the upper bound of $2.15$ we provide in this paper and/or
constructing new, more sophisticated lower-bound instances, that go beyond the
two-item construction of~\textcite[Theorem~3.4]{MaSimchiLevi2021} which provides
the current state-of-the-art lower bound of $1.1877$.

A mathematically deeper question concerns the limitations of our
dual-certificate construction, via which we derive our master inequality
(see~\cref{lemma:rev_upper_bound_master_caps} and its proof). How much ``slack''
is inherent in the master inequality itself, and how much comes from the
probabilistic bounds and derivations we use to estimate its second-part,
residual expectation? Determining the strongest guarantees obtainable by
optimizing the caps in~\eqref{eq:rev_master_inequality} would help identify when
alternative, stronger dual constructions are fundamentally needed. 

Finally, our exact two-item analysis of~\cref{sec:selling_separately_2_iid}
suggests pursuing tight dual-certificates for more items---beginning with the
case of $m=3$ iid items, and investigating whether extremal distributions can
still be chosen from an equal-revenue-style family---or remaining in the
two-item case but considering general independent distributions.


\appendix
\section*{Appendix}

\section{Proof of
\texorpdfstring{\cref{prop:discretization-approximation-revenues}}{Proposition
2.2}: From Discrete to Continuous Revenue}
\label{appendix:proof-discretization-approximation-revenues}
For proving the first part of~\cref{prop:discretization-approximation-revenues},
i.e.~\cref{eq:discrete-distribution-approximation-srev-brev}, we will make use
of the following technical lemma. We note here that this lemma is stated and
proved in more generality than what we need for within the strict scope of the
proof of~\cref{prop:discretization-approximation-revenues}. Nevertheless, we
believe that this abstraction is beneficial for our presentation, since not only
it can be of independent interest for follow-up work, but it also highlights the
exact technical core of the result and allows for a cleaner, higher-level
argument. 

\begin{lemma}
    \label{lemma:approx-revenue-lipschitz-general}
    Let $\vec{Y},\vec{Y}'$ be random variables in $\R_+^m$ such that 
    \[ \vec{Y}'\leq \vec{Y} \qquad\text{and}\qquad \lVert
    \vec{Y}-\vec{Y}'\rVert_{\infty}\leq \delta\] almost surely. Let $A$ be a
    finite index set, and for any $a\in A$ let 
    \[g_{a,1},g_{a,2},\dots,g_{a,n_a}:\R_+^m\map\R_+\]
    be (coordinate-wise) nondecreasing and $L_{a,i}$-Lipschitz functions (with respect to the maximum norm). Then, for the quantity 
    \[R(\vec{Z}) \coloneq \max_{a\in A}\sum_{i=1}^{n_a}\srev(g_{a,i}(\vec{Z})) \]
    it holds that 
    \[0\leq R(\vec{Y}) - R(\vec{Y}') \leq \delta L,  \qquad\text{where}\;\; L\coloneq \max_{a\in A}\sum_{i=1}^{n_a}L_{a,i}. \]
\end{lemma}
\begin{proof}
    First we show that, for any \emph{single-dimensional} random variables $Z,Z'\in \R_{+}$ and $\varepsilon>0$ it holds that 
    \begin{equation}
        \label{eq:approx-revenue-lipschitz-general-helper-single-dimensional}
        0\leq Z- Z' \leq \varepsilon \quad \text{a.s.} 
        \qquad\then\qquad 
        0 \leq \srev(Z) -\srev(Z') \leq \varepsilon.
    \end{equation}
    Indeed, if $Z' \leq Z \leq Z'+\varepsilon$ a.s., then for any potential posted-price $p\in\R_{+}$ we have that 
    \begin{align*}
        p\prob{Z'\geq p} \leq p\prob{Z\geq p} &\leq p\prob{Z'+\varepsilon\geq p}\\
        &=(p-\varepsilon)\prob{Z'\geq p-\varepsilon} + \varepsilon \prob{Z'\geq p-\varepsilon}\\
        &\leq (p-\varepsilon)\prob{Z'\geq p-\varepsilon} + \varepsilon.
    \end{align*} 
    Now note that
\[ \sup_{p\geq 0}(p-\varepsilon)\prob{Z'\geq p-\varepsilon} = \sup_{p\geq
    0}p\prob{Z'\geq p}, \] since $p-\varepsilon\leq 0$ for $p\in[0,\varepsilon]$ and
$Z'\geq 0$. Therefore, by taking suprema (with respect to $p\geq 0$) in the
    above inequalities, and recalling definition~\eqref{eq:srev-def}, we finally get
    the desired
$$
    \srev(Z') \leq \srev(Z) \leq \srev(Z') + \varepsilon,
$$
    which establishes~\eqref{eq:approx-revenue-lipschitz-general-helper-single-dimensional}.
    
    Next, to prove our \cref{lemma:approx-revenue-lipschitz-general}, consider an
    arbitrary $a\in A$ and $i\in[n_a]$. Due to the monotonicity and Lipschitzness of
    function $g_{a,i}$ we have that the following hold a.s.: 
\[
    0\leq g_{a,i}(\vec{Y}) - g_{a,i}(\vec{Y}') 
    \leq \card{g_{a,i}(\vec{Y}) - g_{a,i}(\vec{Y}')} 
    \leq L_{a,i} \maxnorm{\vec{Y}-\vec{Y}'}
    \leq L_{a,i} \cdot \delta.
\]
    Thus,
    applying~\eqref{eq:approx-revenue-lipschitz-general-helper-single-dimensional}
    with $Z\gets g_{a,i}(\vec{Y})$, $Z'\gets g_{a,i}(\vec{Y}')$ and $\varepsilon
    \gets \delta L_{a,i}$ we get that 
\[ 0 \leq \srev(g_{a,i}(\vec{Y}))- \srev(g_{a,i}(\vec{Y}')) \leq \delta L_{a,i}.\] 
    Summing the above inequality for all $i\in [n_a]$ we get that 
\[ \sum_{i=1}^{n_a}\srev(g_{a,i}(\vec{Y}')) \leq
    \sum_{i=1}^{n_a}\srev(g_{a,i}(\vec{Y})) \leq \sum_{i=1}^{n_a}
    \srev(g_{a,i}(\vec{Y}')) + \delta \sum_{i=1}^{n_a} L_{a,i}. \] Taking maxima,
    with respect to $a\in A$, in the above inequality, we finally derive the
    desired:
\[
    R(\vec{Y}') \leq R(\vec{Y}) \leq R(\vec{Y}') + \delta\cdot\max_{a\in A}\sum_{i=1}^{n_a} L_{a,i}.
\]
\end{proof}

\bigskip
Now we are ready to prove~\cref{prop:discretization-approximation-revenues}.
First observe that, by the discretization~\cref{def:k-grid-discretization} we
clearly have that $\vec{X}^{(K)} \leq \vec{X}$; furthermore, since $Kx-\lfloor
Kx \rfloor \leq 1$ for any $x\geq 0$, we also get that 
\[
\maxnorm{\vec{X}-\vec{X}^{(K)}}
=\max_{j\in[m]} \card{X_j-X_j^{(K)}} 
=\max_{j\in[m]}\left(X_j-\frac{\lfloor K X_j \rfloor}{K}\right)
\leq \max_{j\in[m]} \frac{1}{K} = \frac{1}{K}.
\]

Next, we start with establishing the family of
inequalities~\eqref{eq:discrete-distribution-approximation-srev-brev}. For each
of the three revenue operators, we will deploy a different instantiation
of~\cref{lemma:approx-revenue-lipschitz-general}, setting $\vec{Y}\gets\vec{X}$,
$\vec{Y}'\gets \vec{X}^{(K)}$ and $\delta\gets\frac{1}{K}$ for all of them: 
\begin{itemize}
    \item For $R=\srev$, we set the index $A$ to be a singleton $A\gets
    \ssets{\hat{a}}$ with $n_{\hat{a}}\gets m$ functions
    $g_{\hat{a},j}(\vec{X})=X_j$ for each $j\in[m]$. Note that $g_{\hat{a},j}$ are
    increasing functions with Lipschitz constant $L_{\hat{a},j}=1$. Therefore, in
    this case $L=\sum_{j=1}^m L_{\hat{a},j}=m$.
    \item For $R=\brev$, we again choose a singleton index set $A\gets\ssets{\tilde{a}}$, and choose a single function $n_{\tilde{a}}\gets 1$ with $g_{\tilde{a},1}(\vec{X})=X_1+X_2+\dots+X_m$, which is not hard to see that it is increasing and $m$-Lipschitz. Therefore, here $L=m$.
    \item  For $R=\max\sset{\srev,\brev}$ we use a combination of the previous two
    cases, setting the index set to be $A=\sset{\hat{a},\tilde{a}}$, with the
    $g_{a,i}$ functions defined as above. Then, we have that
    $L=\max\ssets{m,m}=m$.
\end{itemize}  

Now we move to proving~\eqref{eq:discrete-distribution-approximation-rev}.
First, observe that, due to~\eqref{eq:ir-def}, we have the following trivial
upper bounds on the optimal revenues, since both random variables $\vec{X}$,
$\vec{X}^{(K)}$ are bounded within the unit hypercube $[0,1]^m$, and thus all
coordinate-marginals are at most $1$: 
\[\rev(\vec{X}),\rev(\vec{X}^{(K)}) \leq m\cdot 1 = m.\] Furthermore, the
Wasserstein distance of order $1$  of $\vec{X}$ and $\vec{X}^{(K)}$ is at most
(see, e.g., \parencite{Villani2009}
or~\parencite[p.~10]{hart2026rootrevenuecontinuity}): 
\[ W_1(\vec{X},\vec{X}^{(K)}) \leq \expect{\sum_{j=1}^m(X_j-X_j^{(K)})} \leq \expect{\sum_{j=1}^m\frac{1}{K}}=\frac{m}{K}\leq m.\]

Therefore, by making use of the recent revenue-continuity results
of~\textcite[Eq.~(10), Corollary~4]{hart2026rootrevenuecontinuity}, we get the
bound: 
\[ \card{\rev(\vec{X})-\rev(\vec{X}^{(K)})} \leq 2\sqrt{m \cdot
W_1(\vec{X},\vec{X}^{(K)})} - W_1(\vec{X},\vec{X}^{(K)}) \leq 2 \sqrt{m\cdot
\frac{m}{K}}-\frac{m}{K}=\frac{2m}{\sqrt{K}}-\frac{m}{K},\] the second
inequality holding due to the fact that function $w\mapsto 2\sqrt{mw}-w$ is
monotonically increasing for $w\in[0,m]$.

\section{Remainder of the Proof of
\texorpdfstring{\cref{lemma:rev_upper_bound_master_caps}: Bounding the
Lagrangian}{Proposition 3.1: Bounding the Lagrangian}}
\label{append:lemma:rev_upper_bound_master_caps}

We verify that by the choice of Lagrangian multipliers all coefficients of the
payment and the allocation variables equal zero, thus, cancel from the
Lagrangian. Using~\eqref{eq:IC-multipliers} we separate the rectangular and the
circular contributions. For any payment variable $p(k_j,\vec{k}_{-j})$ the
opposed circular flows~\eqref{eq:circ_flow} have zero net contribution and
cancel each other. For $k_j=1, \dots, K-1$ the rectangular flows in direction
$j$ leave net flow
\begin{align*}
    \lambda^\mathrm{rect}_{(k_j+1,\vec{k}_{-j})\to(k_j,\vec{k}_{-j})}-\lambda^\mathrm{rect}_{(k_j,\vec{k}_{-j})\to(k_j-1,\vec{k}_{-j})} &=  \frac{1}{S} G_j(k_j+1, \vec{k}_{-j}) - \frac{1}{S} G_j(k_j, \vec{k}_{-j})\\ 
    &= -\frac{1}{S} g_j(k_j, \vec{k}_{-j}).
\end{align*}
For $k_j=0$ by~\eqref{eq:IR-multipliers} we subtract the absorption term
$\frac{1}{S} G_j(0,\vec k_{-j})$ 
\begin{align*}
    \lambda^\mathrm{rect}_{(1,\vec{k}_{-j})\to(0,\vec{k}_{-j})}-\frac{1}{S} G_j(0,\vec k_{-j}) =  \frac{1}{S} G_j(1, \vec{k}_{-j}) - \frac{1}{S} G_j(0, \vec{k}_{-j})
    = -\frac{1}{S} g_j(0, \vec{k}_{-j}),
\end{align*}
and for $k_j=K$
\begin{align*}
    -\lambda^\mathrm{rect}_{(K,\vec{k}_{-j})\to(K-1,\vec{k}_{-j})} =  - \frac{1}{S} G_j(K, \vec{k}_{-j})
    = -\frac{1}{S} g_j(K, \vec{k}_{-j}).
\end{align*}
In every case, summing over all items' directions $j\in[m]$ we obtain
\[\sum_{j=1}^m -\frac{1}{S} g_j(k_j, \vec{k}_{-j}) = -\frac{1}{S} f(\vec{k})\]
which cancels exactly with the initial source mass.

The coefficient of the allocation variable $a_j(\vec k)$ in
\eqref{eq:Lagrangian} is
\begin{equation}\label{eq:allocation-variable-coefficient}
    \alpha_j(\vec{k}) - \psi_j(\vec{k}) + \frac{k_j}{K} \mu(\vec{k}) +  \sum_{\vec{k}^\prime} \frac{k_j}{K}\lambda_{\vec{k}\to\vec{k}^\prime} - \sum_{\vec{k}^\prime} \frac{k_j^\prime}{K} \lambda_{\vec{k}^\prime\to\vec{k}} 
\end{equation}
First, observe that a circular flow in any direction $i\neq j$ contributes zero,
as the two endpoints have the same value index $k_j$ and the two opposed flow
multipliers have the same mass. The rectangular flow in direction $i\neq j$
contributes for $1\leq k_i \leq K-1$
\[ 
\begin{aligned}
    \frac{k_j}{K} \left[\lambda^\mathrm{rect}_{(k_i,\vec{k}_{-i})\to(k_i-1,\vec{k}_{-i})} - \lambda^\mathrm{rect}_{(k_i+1,\vec{k}_{-i})\to(k_i,\vec{k}_{-i})} \right] &= \frac{1}{S}\frac{k_j}{K}\left[G_i(k_i,\vec k_{-i}) - G_i(k_i+1,\vec k_{-i})\right] \\&= \frac{1}{S}\frac{k_j}{K} g_i(k_i,\vec k_{-i}).
\end{aligned}
\]
For $k_i=K$ there is only the outgoing rectangular flow 
\[\frac{k_j}{K} \lambda^\mathrm{rect}_{(K,\vec{k}_{-i})\to(K-1,\vec{k}_{-i})} = \frac{1}{S}\frac{k_j}{K} G_i(K,\vec k_{-i})  = \frac{1}{S}\frac{k_j}{K} g_i(K,\vec k_{-i}), \]
and for $k_i=0$, only the incoming rectangular flow as well as the absorption term $\frac{1}{S} G_i(0,\vec k_{-i})$ from~\eqref{eq:IR-multipliers}
\[- \frac{k_j}{K} \lambda^\mathrm{rect}_{(1,\vec{k}_{-i})\to(0,\vec{k}_{-i})} + \frac{1}{S}\frac{k_j}{K} G_i(0,\vec k_{-i}) = \frac{1}{S}\frac{k_j}{K}\left[ - G_i(1,\vec k_{-i}) +  G_i(0,\vec k_{-i}) \right]= \frac{1}{S}\frac{k_j}{K} g_i(k_i,\vec k_{-i}).\]
Summing over the contributions from all $i\neq j$ and using~\eqref{eq:feasible_flow} they contribute
\begin{equation}\label{eq:rectangular-flow-ineqj}
    \frac{1}{S}\frac{k_j}{K} \sum_{i\neq j}g_i(k_i,\vec k_{-i}) = \frac{1}{S}\frac{k_j}{K} \left(f(\vec k) - g_j(k_j,\vec k_{-j}) \right).
\end{equation}
This contribution applies to every value of $k_j$ and all $\vec{k}_{-j}$.

Now fix an item $j$, $\vec{k}_{-j}$, and consider $1\leq k_j \leq K-1$. By~\eqref{eq:rect_flow} the rectangular flow contribution in direction $j$ is
\[ 
\frac{k_j}{K} \lambda^\mathrm{rect}_{(k_j,\vec{k}_{-j})\to(k_j-1,\vec{k}_{-j})} - \frac{k_j+1}{K}\lambda^\mathrm{rect}_{(k_j+1,\vec{k}_{-j})\to(k_j,\vec{k}_{-j})}  = \frac{1}{S}\left[ \frac{k_j}{K} G_j(k_j,\vec k_{-j}) - \frac{k_j+1}{K}G_j(k_j+1,\vec k_{-j})\right].
\]
The circular flows do not cancel in direction $j$. However, we make use of the
fact that the multipliers of two opposed circular flows are equal and their
coefficients in~\eqref{eq:allocation-variable-coefficient} differ by $1/K$.
Therefore, their contribution is
\[ 
\begin{aligned}
    \frac{1}{K} \lambda^\mathrm{circ}_{(k_j,\vec{k}_{-j})\to(k_j-1,\vec{k}_{-j})} &- \frac{1}{K} \lambda^\mathrm{circ}_{(k_j+1,\vec{k}_{-j})\to(k_j,\vec{k}_{-j})} \\
    & \hspace{-35pt}= \frac{1}{S} \left[ -\frac{k_j}{K}G_j(k_j,\vec k_{-j}) +\frac{k_j+1}{K}G_j(k_j+1,\vec k_{-j}) -\frac{k_j}{K}\left(f(k_j,\vec k_{-j})-g_j(k_j,\vec k_{-j})\right) \right].
\end{aligned}
\]
To obtain this, we apply the definition in~\eqref{eq:circ_flow} such that the
$R_j(\vec k_{-j})$ and all summands of the difference $(f-g_j)$ with $l > k_j$
cancel. Since $\alpha_j$ and $\psi_j$ are zero for all considered grid points,
in total the coefficient of $a_j(k_j,\vec{k}_{-j})$ cancels.

For $k_j=0$ the rectangular contributions~\eqref{eq:rectangular-flow-ineqj} from
items $i\neq j$ are zero. In direction of $j$ the rectangular contribution is 
\[ 
- \frac{1}{K}\lambda^\mathrm{rect}_{(1,\vec{k}_{-j})\to(0,\vec{k}_{-j})} = -\frac{1}{S} \frac{1}{K}G_j(1,\vec k_{-j}),
\]
and the circular contribution is
\[ 
- \frac{1}{K} \lambda^\mathrm{circ}_{(1,\vec{k}_{-j})\to(0,\vec{k}_{-j})} = \frac{1}{S} \left[ -R_j(\vec k_{-j}) + \frac{1}{K}G_j(1,\vec k_{-j}) + \sum_{l=1}^K \frac{l}{K} \left(f(l,\vec k_{-j})-g_j(l,\vec k_{-j})\right)
\right].
\]
Their sum is
\[
-\frac{1}{S} \left[ R_j(\vec k_{-j}) - \sum_{l=0}^K \frac{l}{K} \left(f(l,\vec k_{-j})-g_j(l,\vec k_{-j})\right)
\right] = -\alpha_j(0,\vec k_{-j}),
\]
hence, by~\eqref{eq:def_alpha_multiplier} the coefficient cancels.

For $k_j=K$ the rectangular contribution in direction $j$ is 
\[\lambda^\mathrm{rect}_{(K,\vec{k}_{-j})\to(K-1,\vec{k}_{-j})} = \frac{1}{S}G_j(K,\vec k_{-j}) = \frac{1}{S}g_j(K,\vec k_{-j}).\]
By \eqref{eq:rectangular-flow-ineqj} items' $i\neq j$ contributions are $\frac{1}{S}\left(f(K,\vec k_{-j})-g_j(K,\vec k_{-j})\right)$. Furthermore, the circular and by~\eqref{eq:def_psi_j_multiplier} the upper boundary contributions are 
\[\begin{aligned}
\frac{1}{K} \lambda^\mathrm{circ}_{(K,\vec{k}_{-j})\to(K-1,\vec{k}_{-j})} &= \frac{1}{S} \left[ R_j(\vec k_{-j}) -f(K,\vec k_{-j})\right],\\
-\psi_j(K,\vec k_{-j}) &= -\frac{1}{S} R_j(\vec k_{-j}).
\end{aligned}
\] In total, these four contributions add up to zero. Hence, all allocation
variables' coefficients cancel from~\eqref{eq:Lagrangian}.

Finally, we will quantify the Lagrangian objective. By our choice of the
multipliers the Lagrangian~\eqref{eq:Lagrangian} reduces to 
\[L(\rho, \vec{f}, \vec{a}, \vec{p}; \vec{\psi}, \vec{\alpha}, \vec{\lambda},
\vec{\mu}, \vec{\theta}) = \sum_{j=1}^m \sum_{\vec{k}} \psi_j(\vec{k})=
\sum_{j=1}^m \sum_{\vec{k}_{-j}} \frac{1}{S} R_j(\vec{k}_{-j}). \] By weak
duality and multiplication by $S$ this provides us an upper bound of the
optimal revenue
\[\rev(\vec{X}) \leq \sum_{j=1}^m \sum_{\vec{k}_{-j}} R_j(\vec{k}_{-j}).\] It
remains to bound $R_j(\vec{k}_{-j})$. For this fix a \emph{cap} $c_j \geq 0$ for
each item $j$ which need not be support points of $X_j$. Now fix an item $j$, a
coordinate line $\vec k_{-j}$, and some $k_j\in \Kzero$. If $\frac{k_j}{K} \geq
c_j$ observe for all $l\geq k_j$
\[
\begin{aligned}\frac{k_j}{K}g_j(l,\vec k_{-j})+\frac{l}{K}\left(f(l,\vec k_{-j})-g_j(l,\vec k_{-j})\right)& = \frac{k_j}{K} f(l,\vec k_{-j})+ \left( \frac{l}{K}-\frac{k_j}{K} \right) \left(f(l,\vec k_{-j})-g_j(l,\vec k_{-j})\right)\\
    &\leq \frac{k_j}{K} f(l,\vec k_{-j})+ \pospart{\frac{l}{K} - c_j} \left(f(l,\vec k_{-j})-g_j(l,\vec k_{-j})\right)
\end{aligned}
\]
If on the other hand $\frac{k_j}{K} \leq c_j$ then for all $l\geq k_j$
\[
\begin{aligned}\frac{k_j}{K}g_j(l,\vec k_{-j})&+\frac{l}{K}\left(f(l,\vec k_{-j})-g_j(l,\vec k_{-j})\right) \leq  \left(\frac{l}{K}\minshort c_j \right) g_j(l,\vec k_{-j})+\frac{l}{K}\left(f(l,\vec k_{-j})-g_j(l,\vec k_{-j})\right)\\
    &=\left(\frac{l}{K}\minshort c_j \right) g_j(l,\vec k_{-j})+  \left[\left(\frac{l}{K}\minshort c_j \right) + \pospart{\frac{l}{K} - c_j} \right]  \left(f(l,\vec k_{-j})-g_j(l,\vec k_{-j})\right)\\
    &=\left(\frac{l}{K}\minshort c_j \right) f(l,\vec k_{-j}) + \pospart{\frac{l}{K} - c_j} \left(f(l,\vec k_{-j})-g_j(l,\vec k_{-j})\right).
\end{aligned}
\]
Recall the definition of $R_j(\vec k_{-j})$ in~\eqref{eq:def_line_revenue} and
note, that in both cases, the left-hand side is one summand of $R_j(\vec
k_{-j})$. The two inequalities give a bound for every $l \geq k_j$, including
the maximizing index in~\eqref{eq:def_line_revenue}. For $\frac{k_j}{K} \geq
c_j$ the first term on the right-hand side sums to
\[\sum_{l=k_j}^K \frac{k_j}{K}  f(l,\vec k_{-j}) = f_{-j}(\vec k_{-j}) \frac{k_j}{K} \sum_{l=k_j}^K f_j(l) \leq f_{-j}(\vec k_{-j}) \rev(X_j)\]
and for $\frac{k_j}{K} \leq c_j$
\[\sum_{l=k_j}^K \left(\frac{l}{K}\minshort c_j \right) f(l,\vec k_{-j}) \leq
\sum_{l=0}^K \left(\frac{l}{K}\minshort c_j \right) f(l,\vec k_{-j}) =
f_{-j}(\vec k_{-j}) \expect{ \left(X_j \minshort c_j \right) }.\] The second
term is identical in both cases and as $f(\vec k)\geq g_j(\vec k)$ we
upper-bound
\[\sum_{l=k_j}^K \pospart {\frac{l}{K} - c_j} \left(f(l,\vec k_{-j})-g_j(l,\vec k_{-j})\right) \leq \sum_{l=0}^K \pospart {\frac{l}{K} - c_j} \left(f(l,\vec k_{-j})-g_j(l,\vec k_{-j})\right). \]
The resulting bound for $R_j(\vec k_{-j})$ is independent of $k_j$ and we obtain
\begin{equation}\label{eq:upper-bound-line-revenue}
    R_j(\vec k_{-j}) \leq f_{-j}(\vec k_{-j}) \max \sset{\rev(X_j), \expect{ \left( X_j \minshort c_j \right) }} + \sum_{l=0}^K \pospart {\frac{l}{K} - c_j} \left(f(l,\vec k_{-j})-g_j(l,\vec k_{-j})\right)
\end{equation}
for any choice of caps $c_j$. 

We finally choose the routing rule at each $\vec k$ such that $w_j(\vec k) = 1$
for the smallest index $j$ with $j \in \argmax_j \pospart {\frac{k_j}{K} - c_j}$
and is zero otherwise. Clearly such $j$ exists and satisfies
conditions~\eqref{eq:flow_rule_definition}. Now summing
\eqref{eq:upper-bound-line-revenue} over all items and lines, and using
$\sum_{\vec k_{-j}} f_{-j}(\vec k_{-j}) = 1$ we bound the objective and obtain
our final result
\[\begin{aligned}
\rev(\vec{X}) &\leq \sum_{j=1}^m \sum_{\vec{k}_{-j}} R_j(\vec{k}_{-j}) \\
&\leq \sum_{j=1}^m  \max \sset{\rev(X_j), \expect{\left(X_j\minshort c_j
\right) }} + \sum_{\vec k} f(\vec k) \sum_{j=1}^m  \pospart {\frac{k_j}{K} -
c_j} \left(1-w_j(\vec k)\right)\\
&= \sum_{j=1}^m \max \sset{ \rev(X_j), \expect{ \left(X_j\minshort c_j
\right) }} + \expect{ \sum_{j=1}^m  \pospart {X_j - c_j} - \max_j \pospart {X_j
- c_j}} \end{aligned}\] which proves our master
inequality~\eqref{eq:rev_master_inequality}.

\section{Omitted Proofs from
\texorpdfstring{\cref{sec:approximation_ratios}}{Section 4}}
\label{append:sec:approximation_ratios}

\subsection{Proof of \texorpdfstring{\cref{lemma:bounds_on_sum_of_tails}}{Lemma
4.1}}
\label{append:lemma:bounds_on_sum_of_tails}

By definition of the optimal single item revenue~\eqref{eq:monopoly-revenue} as
the supremum over all possible posted prices, for any posted price $u>0$ the
quantity $u\tail{F}_j(u) \leq r_j$. Choosing the same posted price for all items
is clearly upper bounded by individually choosing the optimal price which
achieves $r_j$, i.e., $\sum_{j=1}^m u \tail{F}_j(u) \leq \srevshort$.

Analogously, applying~\eqref{eq:monopoly-revenue} to the definition of the
optimal bundling mechanism~\eqref{eq:brev-def}, selling the grand bundle at
any posted price $u>0$ cannot exceed the optimal bundling mechanism, i.e.,
$u \prob{\sum_{j=1}^m X_j > u} \leq \brevshort$. For the last inequality we
use $1-x \leq e^{-x}$ to obtain
\[1-e^{-\sum_{j=1}^m \tail{F}_j(u)} \leq 1- \prod_{j=1}^m \left( 1-\tail{F}_j(u) \right) = \prob{\max_{j\in[m]} X_j > u} \leq  \prob{\sum_{j=1}^m X_j > u} \leq \frac{\brevshort}{u}. \]
Rearranging and taking the logarithm, if $u> \brevshort$, yields the inequality.

\subsection{Proof of \texorpdfstring{\cref{lemma:tail_inequalities}}{Lemma 4.4}}
\label{append:lemma:tail_inequalities}

Observe that for $x_1,\dots, x_m \in [0,1]$ by $1-x \leq e^{-x}$ pointwise holds
\[ \sum_{j=1}^m x_j - 1 + \prod_{j=1}^m (1-x_j) \leq \phi \left(
\sum_{j=1}^m x_j \right). \] Thus, applying this together with caps $c_j$
such that $\expect{X_j \minshort c_j} = r_j$
to~\eqref{eq:rev_master_inequality} we
obtain~\eqref{eq:master_inequality_phi}.

The properties of $\phi$ are straightforward: $\phi(0)=0, \phi^\prime(x) = 1-e^{-x}\geq0$, and $\phi^{\prime\prime}(x) = e^{-x}>0$. Finally, for $x\in[0,U]$ $\phi(x) - wx$ is convex and bounded from above by the larger of its two boundary points with $\phi(0) - 0 = 0$ and $\phi(U) - wU$.

\subsection{Proof of \texorpdfstring{\cref{lemma:bound_on_weighted_tails}}{Lemma 4.5}}
\label{append:lemma:bound_on_weighted_tails}
For convenience denote $\bar{X} = \sum_{j=1}^m X_j$. Since $\brevshort$ is the optimal revenue from selling the grand bundle, for any $u>0$ it must hold that $u\prob{\bar{X} \geq u} \leq \brevshort$. Thus for $u\geq \brevshort$, $\prob{\bar{X} \leq u} \geq 1 - \brevshort/u$. In the following we use the substitution $u=\brevshort x$ and integration by parts and obtain for every $\lambda>0$ the lower bound
\[\begin{aligned}
    \expect{e^{-\lambda \bar{X}}} &= \int_0^\infty \lambda e^{- \lambda u} \prob{\bar{X} \leq u} \, du \\
    &\geq \lambda \int_\brevshort^\infty e^{- \lambda u} \left(1- \frac{\brevshort}{u} \right) \, du = \lambda \brevshort \int_1^\infty e^{- \lambda \brevshort x} \left(1- \frac{1}{x} \right) \, dx \\
    &= \left[-e^{- \lambda \brevshort x} \left(1- \frac{1}{x} \right) \right]_1^\infty +\int_1^\infty e^{- \lambda \brevshort x} \frac{1}{x^2} \, dx = \int_1^\infty e^{- \lambda \brevshort x} \frac{1}{x^2} \, dx,
\end{aligned} \]
as well as the upper bound
\[ \expect{e^{-\lambda \bar{X}}} = \prod_{j=1}^m \expect{e^{-\lambda X_j}} = \prod_{j=1}^m \left( 1 - \lambda \int_0^\infty e^{- \lambda u} \tail{F}_j(u) \, du\right) \leq \exp \left( - \lambda \sum_{j=1}^m \int_0^\infty e^{- \lambda u} \tail{F}_j(u) \, du\right), \]
where the inequality follows from $1-x\leq e^{-x}$. Combining the two inequalities, applying the logarithm, and dividing by $-\lambda$ yields
\begin{equation}\label{eq:tail_inequality_without_caps}
    \sum_{j=1}^m \int_0^\infty e^{- \lambda u} \tail{F}_j(u) \, du \leq - \frac{1}{\lambda} \ln \left( \int_1^\infty e^{- \lambda \brevshort x} \frac{1}{x^2} \, dx \right).
\end{equation}

As both $e^{-\lambda u}$ and $\tail{F}_j(u)$ decrease in $u$ we use
Chebyshev's integral inequality (see, e.g., \cite[Eq.
12.314]{gradshteyn2007}) to obtain
\[ \int_0^{c_j} e^{- \lambda u} \tail{F}_j(u) \, du \geq \frac{1}{c_j} \left( \int_0^{c_j} e^{- \lambda u} \, du \right)  \left( \int_0^{c_j} \tail{F}_j(u) \, du \right) .\]
Choosing $c_j$ such that
\[ \int_0^{c_j} \tail{F}_j(u) \, du = \expect{X_j \minshort c_j} = r_j\]
we obtain
\begin{equation}\label{eq:cap_residual}
    \int_0^{c_j} e^{- \lambda u} \tail{F}_j(u) \, du \geq \frac{r_j}{c_j} \frac{1-e^{- \lambda c_j}}{\lambda}.
\end{equation}
Combining \eqref{eq:tail_inequality_without_caps} and summing \eqref{eq:cap_residual} over all items yields the desired inequality.

\subsection{Proofs of Technical Claims}
\label{appendix:proof_of_claims}

GPT AI models assisted with developing candidate tail estimates, weighting
functions, and arguments used for the proofs of our technical
\cref{claim:mrev_independent,claim:mrev_iid,claim:brev_iid} in this section;
further details appear in our \hyperref[sec:ai-methodology]{AI methodology
statement}. Our intention was to fine-tune the numerical parameters of our
analysis, and to strike a satisfactory balance between the strength of the
bounds and the transparency of the presentation. Throughout this exploration,
the AI agents were even able to suggest alternative analyses, achieving
(slightly) better bounds than the one we present below; however, we decided not
to include them here, due to their more obscure nature, and the lack of
added-value in terms of insights or techniques. Nevertheless, for follow-up
work, it is important to keep in mind that this suggests that the upper bounds
of \cref{th:M-upper-bound-independent,th:M-upper-bound-iid,th:B-upper-bound-iid}
(which rely on the technical inequalities provided by the claims proven below)
should be considered improvable.

\subsubsection{Proof of \texorpdfstring{\cref{claim:mrev_independent}}{Claim 4.7}}
\label{appendix:claim_mrev_independent}

We start by adding and subtracting $w(u) \tail{F}_j(u)$ to the integrand.
Applying~\eqref{eq:phi_convex_bounded} of~\cref{lemma:tail_inequalities}
with $x=\tail{F}_j(u)/\pi_j$ and $U=r_j/(\pi_ju)$ we obtain
\begin{align}
    \pi_j \phi \left(  \frac{\tail{F}_j(u)}{\pi_j} \right) &= \left[\pi_j \phi \left(  \frac{\tail{F}_j(u)}{\pi_j} \right) - w(u) \tail{F}_j(u) \right] + w(u) \tail{F}_j(u) \notag \\
    &\leq w(u) \tail{F}_j(u) + r_j \pospart{\frac{\pi_j}{r_j} \, \phi \left(  \frac{r_j}{\pi_j u} \right) - \frac{w(u)}{u}}. \label{eq:expanded_integrand_independent}
\end{align}
We now choose the weights $\pi_j$ and the function $w$ such that both
summands are individually bounded using $\srev$ and $\brev$ and the
additional terms in each bound cancel. Let $l>0$ be such that
\[\sum_{j=1}^m \frac{r_j}{\sqrt{6l^2 + 4l c_j}} = 1.\] Clearly, the sum is
monotone and continuous in $l$. For $l\to 0$ it diverges and for $l \to
\infty$ it goes to zero. Thus, existence and uniqueness follow. We set
\begin{align}
    \pi_j &= \frac{r_j}{\sqrt{6l^2 + 4l c_j}} \label{eq:jensen_weights}\\
    w(u) &= \frac{3}{4} e^{-u/l} + \frac{1}{4} e^{-u/(22l)}. \label{eq:weight_function}
\end{align}
For readability we denote $L(y) = -\ln \left( \int_1^\infty e^{-yx}/x^2 \,
\mathrm{d}x \right)$. We use~\eqref{eq:weight_function} and linearity for
the first equality, then we apply~\cref{lemma:bound_on_weighted_tails} on
each summand, the first with $\lambda=1/l$ and the second with
$\lambda=1/(22l)$. We obtain
\[\begin{aligned} \sum_{j=1}^m \int_{c_j}^\infty w(u) \tail{F}_j(u) \, du &=
    \frac{3}{4} \sum_{j=1}^m \int_{c_j}^\infty  e^{-u/l} \tail{F}_j(u) \, du +
    \frac{1}{4} \sum_{j=1}^m \int_{c_j}^\infty e^{-u/(22l)} \tail{F}_j(u) \, du
    \\    &\leq \frac{3l}{4} L\left(\frac{\brevshort}{l} \right) + \frac{22l}{4}
    L\left(\frac{\brevshort}{22l} \right) - \sum_{j=1}^m \frac{r_j}{c_j}
    \int_0^{c_j} w(u) \, du.
\end{aligned}
\]

Observe that $L(y)$ is concave on $y\in(0,\infty)$. This follows by
\[L^\prime(y) = \frac{\int_1^\infty e^{-yx}/x \, dx}{\int_1^\infty
e^{-yx}/x^2 \, dx}, \qquad L^{\prime \prime} (y) = \frac{\left(
\int_1^\infty e^{-yx}/x \, dx \right)^2 - \left( \int_1^\infty e^{-yx} \, dx
\right)\left( \int_1^\infty e^{-yx}/x^2 \, dx \right)}{\left( \int_1^\infty
e^{-yx}/x^2 \, dx \right)^2}\] and that the numerator of $L^{\prime \prime}$
is nonpositive by the Cauchy-Schwarz inequality. 

Now we show that
\begin{equation}\label{eq:brev_constant_independent}
    \frac{3}{4} L \left(y \right) + \frac{11}{2} L \left(\frac{y}{22} \right) - \frac{4}{3} y \leq 1.92   
\end{equation}
We evaluate the derivatives
\[
\begin{aligned}
    \frac{3}{4} L^\prime \left(6 \right) + \frac{1}{4} L^\prime \left(\frac{6}{22} \right) &\approx 1.342363 > \frac{4}{3} \\
    \frac{3}{4} L^\prime \left(6.5 \right) + \frac{1}{4} L^\prime \left(\frac{6.5}{22} \right) &\approx 1.326682 < \frac{4}{3}.
\end{aligned}
\]
Thus, the left-hand side of~\eqref{eq:brev_constant_independent} has its
global maximum in $[6,6.5]$. By concavity, its tangent at $y = 6$ provides
an upper bound. As we showed that this tangent has positive slope
\[\frac{3}{4} L \left(6 \right) + \frac{11}{2} L \left(\frac{6}{22} \right) - \frac{4}{3} 6 + \frac{1}{2} \left(\frac{3}{4} L^\prime \left(6 \right) + \frac{1}{4} L^\prime \left(\frac{6}{22} \right) - \frac{4}{3}  \right) \approx 1.914124 < 1.92\]
which proves~\eqref{eq:brev_constant_independent}. 

Substituting $y=\brevshort/l$ and multiplying with $l$ we establish that:
\begin{align}
    \sum_{j=1}^m \int_{c_j}^\infty w(u) \tail{F}_j(u) \, du &\leq \frac{3l}{4} L\left(\frac{\brevshort}{l} \right) + \frac{22l}{4} L\left(\frac{\brevshort}{22l} \right) - \sum_{j=1}^m \frac{r_j}{c_j} \int_0^{c_j} w(u) \, du \notag \\ 
    &\leq \frac{4}{3} \brevshort + 1.92l - \sum_{j=1}^m \frac{r_j}{c_j} \int_0^{c_j} w(u) \, du. \label{eq:m_independent_claim_brev}
\end{align} 

The remainder of the proof establishes that the choice of Jensen
weights~\eqref{eq:jensen_weights} ensures for every item
\begin{equation}\label{eq:independent_claim_pospart}
    \int_{c_j}^\infty \pospart{ \frac{1}{\sqrt{6l^2 + 4l c_j}} \ \phi\left( \frac{\sqrt{6l^2 + 4l c_j}}{u} \right) - \frac{w(u)}{u} } \, du \leq \frac{1}{6} - \frac{1.92l}{\sqrt{6l^2 + 4l c_j}} + \frac{1}{c_j}\int_0^{c_j} w(u) \, du.
\end{equation}
Multiplying with $r_j$ and summing over all items then yields the desired
inequality. 

Now fix an item $j$ and substitute $t=c_j/l$. Our remaining inequality
becomes
\begin{equation}\label{eq:independent_claim_indequality_t}
    \int_{t}^\infty \pospart{ \frac{1}{\sqrt{6 + 4t}} \ \phi\left( \frac{\sqrt{6+ 4 t}}{u} \right) - \frac{w(lu)}{u} } \, du + \frac{1.92}{\sqrt{6 + 4t}} - \frac{1}{t}\int_0^{t} w(lu) \, du \leq \frac{1}{6}
\end{equation}
and we need to prove this inequality for every $t>0$. To handle the first
integral we use $\pospart{x} = x + \pospart{-x}$ to separate it into two
parts. To further upper bound it we enlarge the integration domain of the
positive part. The left-hand side
of~\eqref{eq:independent_claim_indequality_t} is then at most $R(t) + K(t)$
with
\begin{align}
    R(t) &= \int_{t}^\infty \left[ \frac{1}{\sqrt{6 + 4t}} \ \phi\left( \frac{\sqrt{6+ 4 t}}{u} \right) - \frac{w(lu)}{u} \right] \, du + \frac{1.92}{\sqrt{6 + 4t}} - \frac{1}{t}\int_0^{t} w(lu) \, du \label{eq:claim_independent_R_integral}\\
    K(t) &= \int_{0}^\infty \pospart{ \frac{w(lu)}{u}-\frac{1}{\sqrt{6 + 4t}} \ \phi\left( \frac{\sqrt{6+ 4 t}}{u} \right) } \, du\label{eq:claim_independent_K_integral}
\end{align}
and it remains to show that $R(t) + K(t) \leq 1/6$ for all $t>0$. Both
functions assume finite values at $t=0$. Observe that for fixed $u>0$ and
$x\in[0,1]$ the function $e^{-x\sqrt{6+4t}/u}$ is convex and nonincreasing
in $t$ as $\sqrt{6+4t}$ is concave and increasing. Integration over $x$,
adding constants, taking the positive part, and finally integrating over $u$
preserve both properties. Hence, $K(t)$ is convex and nonincreasing for
$t>0$. Now we differentiate $R(t)$ and do this in two steps: The two terms
involving $w$
\[ \frac{d}{dt} \left[ - \int_{t}^\infty \frac{w(lu)}{u}  \, du -
\frac{1}{t}\int_0^{t} w(lu) \, du  \right] =  \frac{w(lt)}{t} -
\frac{w(lt)}{t} + \frac{1}{t^2}\int_0^{t} w(lu) \, du ,\] secondly the term
involving $\phi$. There we apply Leibniz's integral rule, since both the
integrand and its partial derivative with respect to $t$ decay in order
$O(u^{-2})$ as $u\to \infty$. Differentiating the integrand itself we obtain
\[ \frac{\partial}{\partial t} \left[ \frac{1}{\sqrt{6 + 4t}} \ \phi\left(
\frac{\sqrt{6+ 4 t}}{u} \right)  \right] = -\frac{2}{(6+4t)^{3/2}}
\frac{\partial}{\partial u}\left[ u \phi\left( \frac{\sqrt{6+ 4 t}}{u}
\right) \right].\] Combining both contributions therefore gives
\[\frac{d}{dt} \left[ \int_t^\infty \frac{1}{\sqrt{6 + 4t}} \ \phi\left(
\frac{\sqrt{6+ 4 t}}{u} \right)  \, du  \right] = -\frac{2t +
6}{(6+4t)^{3/2}} \phi\left( \frac{\sqrt{6+ 4 t}}{t} \right) . \] The only
remaining term is $1.92/\sqrt{6+4t}$ with derivative $-3.84/(6+4t)^{3/2}$.
Putting everything together we obtain
\[R^\prime (t) = \frac{1}{t^2} \int_0^t w(lu) \, du - \frac{(2t+6) \phi (\sqrt{6+4t}/t) + 3.84}{(6+4t)^{3/2}}. \]

We start with $0\leq t \leq 1.5$. $K$ is convex, $K(0) \approx 0.070269 <
0.071$, and $K(1.5) \approx 0.018047 < 0.02$. Hence, in this interval, $K(t)
\leq 0.071 - 0.034 t$.  By the fundamental theorem of calculus,  $R(t) =
R(0) + \int_0^t R^\prime (x) dx$, it holds, and therefore
\[R(t) + K(t) \leq R(0) + 0.071 + \int_0^t (R^\prime(x) - 0.034) \, dx  \leq
R(0) + 0.071 + \int_0^{1.5} \pospart{R^\prime(x) - 0.034 } \, dx .\] We use
$R(0) \approx 0.061387 < 0.062$ and $\int_0^{1.5} \pospart{R^\prime(x) -
0.034 } \, dx \approx 0.031380 < 0.032$ and obtain
\[R(t) + K(t) \leq 0.062 + 0.071 + 0.032 = 0.165 < \frac{1}{6}.\]

Next, consider $1.5\leq t \leq 3$. Here, we evaluate additionally $K(3)
\approx 0.000483 < 0.002$ and get $K(t) \leq 0.02 - 0.012 (t-1.5)$. We
further use $R(1.5) \approx 0.141594 < 0.142$ as well as $\int_{1.5}^3
\pospart{R^\prime(x) - 0.012 } \, dx \approx 0.002632 < 0.003$ and obtain
\[R(t) + K(t) \leq 0.142 + 0.020 + 0.003 = 0.165 < \frac{1}{6}.\]

Finally, consider $t\geq 3$. As $K$ is nonincreasing, it holds $K(t)\leq
K(3)$. We further use $R(3) \approx 0.157471 < 0.158$ and $\int_3^{\infty}
\pospart{R^\prime(x)} \, dx \approx 0.001851 < 0.002$ and obtain
\[R(t) + K(t) \leq  0.158 + 0.002 + 0.002 = 0.162 < \frac{1}{6}.\]

Now we showed that for every
$t>0$~\eqref{eq:independent_claim_indequality_t} holds, by this
also~\eqref{eq:independent_claim_pospart}.
Multiplying~\eqref{eq:independent_claim_pospart} for each item with $r_j$
and summing over all items shows the upper bound for the second summand
in~\eqref{eq:expanded_integrand_independent}
\[\sum_{j=1}^m r_j \int_{c_j}^\infty \pospart{\frac{\pi_j}{r_j} \, \phi
\left(  \frac{r_j}{\pi_j u} \right) - \frac{w(u)}{u}} \, \mathrm{d}u \leq
\frac{1}{6} \srevshort - 1.92l + \sum_{j=1}^m \frac{r_j}{c_j} \int_0^{c_j}
w(u) \, du. \] Together with~\eqref{eq:m_independent_claim_brev} we obtain
\[\begin{aligned} \sum_{j=1}^m \pi_j \int_{c_j}^\infty  \phi \left(
    \frac{\tail{F}_j(u)}{\pi_j} \right) \, du  &\leq \frac{4}{3} \brevshort
    + 1.92l - \sum_{j=1}^m \frac{r_j}{c_j} \int_0^{c_j} w(u) \, du +
    \frac{1}{6} \srevshort - 1.92l + \sum_{j=1}^m \frac{r_j}{c_j}
    \int_0^{c_j} w(u) \, du \\
    &= \frac{4}{3} \brevshort + \frac{1}{6} \srevshort,
\end{aligned}
\]
which proves the claim.

\subsubsection{Proof of \texorpdfstring{\cref{claim:mrev_iid}}{Claim 4.9}}
\label{appendix:claim_mrev_iid}

We start by choosing $\lambda = 7/(8\mrevshort)$ and $\alpha$ in dependence of
small or large caps: For $0 < c \leq \mrevshort 11/40$ we choose $\alpha =
1-c/(2\mrevshort)$ and $\alpha = 0$ otherwise. 

First we consider the case $0 < c \leq \mrevshort 11/40$. Using $\pospart{x} = x
+ \pospart{-x}$ and defining
\begin{align}
    R(c) &=  \left(1-\alpha \frac{1-e^{-\lambda c}}{\lambda c}\right) \mrevshort - \frac{\alpha}{\lambda} \ln \left( \int_1^\infty e^{-\lambda \mrevshort x} \frac{1}{x^2} \, dx \right) + \int_c^\infty \left[\phi \left( \frac{\mrevshort}{u} \right) - \alpha e^{-\lambda u} \frac{\mrevshort}{u}\right] \, du \\
    K(c) &= \int_0^\infty \pospart{ \alpha e^{-\lambda u} \frac{\mrevshort}{u} - \phi \left( \frac{\mrevshort}{u} \right) } \, du 
\end{align}
it suffices to show $R(c) + K(c) \leq {43}/{20} \mrevshort$. Observe that for
every fixed $u>0$ the term inside the $\pospart{\,\cdot \,}$ operator is
linearly decreasing in $c$. Taking the positive part and zero otherwise yields
convexity of the integrand. Integrating over $u$ preserves the convexity and
monotonicity of $K(c)$.

Next, we want to compute $R^\prime (c)$. We do this in two parts: Recall that
$\alpha^\prime(c) = -1/(2\mrevshort)$, then 
\begin{multline*}
    \frac{d}{dc} \left[ \left(1-\alpha \frac{1-e^{-\lambda c}}{\lambda c}\right) \mrevshort - \frac{\alpha}{\lambda} \ln \left( \int_1^\infty e^{-\lambda \mrevshort x} \frac{1}{x^2} \, dx \right) \right] \\
    = \frac{\mrevshort (1-e^{-\lambda c})}{\lambda c ^2} - \frac{\mrevshort  \alpha e^{-\lambda c}}{c} + \frac{1}{2 \mrevshort \lambda} \ln \left( \int_1^\infty e^{-\lambda \mrevshort x} \frac{1}{x^2} \, dx \right),
\end{multline*}
For the remaining part we apply Leibniz's integral rule, since both the
integrand and its partial derivative with respect to $c$ decay in order
$O(u^{-2})$ as $u\to \infty$, the second part decays even exponentially in $u$.
By that we obtain
\[\frac{d}{dc} \left[ \int_c^\infty  \phi \left( \frac{\mrevshort}{u} \right) - \alpha e^{-\lambda u} \frac{\mrevshort}{u} \, du \right] = -\phi \left( \frac{\mrevshort}{c} \right) + \alpha e^{-\lambda c} \frac{\mrevshort}{c} + \frac{1}{2} \int_c^\infty \frac{e^{-\lambda u}}{u} \, du .\]
Summing both parts together, the derivative now is
\begin{equation}
    R^\prime(c) = \frac{\mrevshort (1-e^{-\lambda c})}{\lambda c ^2}  + \frac{1}{2 \mrevshort \lambda} \ln \left( \int_1^\infty e^{-\lambda \mrevshort x} \frac{1}{x^2} \, dx \right)-\phi \left( \frac{\mrevshort}{c} \right) + \frac{1}{2} \int_c^\infty \frac{e^{-\lambda u}}{u} \, du .
\end{equation}

We now provide the upper bounds for $R(c) + K(c)$ for all $0 < c \leq \mrevshort
11/40$, starting with the case $0 < c \leq \mrevshort/20$. As $K$ is convex,
$K(0) \approx 0.146278 \mrevshort < 0.1463 \mrevshort$ and $K(\mrevshort/20)
\approx 0.087753 \mrevshort < 0.0878 \mrevshort$. Clearly, the straight line
through these two upper bound points lies above $K(c)$, therefore, $K(c) \leq
0.1463 \mrevshort - 1.17c$. By the fundamental theorem of calculus,  $R(c) =
R(0) + \int_0^c R^\prime (u) du$, it holds
\[
\begin{aligned}
    R(c) + K(c) &\leq R(0) + 0.1463 \mrevshort + \int_0^c (R^\prime (u) - 1.17) \, du \\ &\leq R(0) + 0.1463 \mrevshort + \int_0^{\mrevshort/20} \pospart{R^\prime (u) - 1.17} \, du.
\end{aligned}
\]
We use $R(0) \approx 1.986702 \mrevshort < 1.9868 \mrevshort$ as well as
$\int_0^{\mrevshort/20} \pospart{R^\prime (u) - 1.17} \, du \approx 0.013538
\mrevshort < 0.0136 \mrevshort$ and obtain
\[R(c) + K(c) \leq (1.9868 + 0.1463 + 0.0136) \mrevshort = 2.1467 \mrevshort <
\frac{43}{20} \mrevshort.\] The next case is $\mrevshort/20 < c \leq
7\mrevshort/50$. We use $K(7\mrevshort/50) \approx 0.035428 \mrevshort < 0.0355
\mrevshort$. The line between the boundary points has slope
$\frac{0.0878-0.0355}{0.14-0.05} \approx 0.581111$ and upper bounds every $K(c)$
in this interval. Therefore, $K(c) \leq 0.0878 \mrevshort - 0.581
(c-\mrevshort/20)$. Using $R(\mrevshort/20) \approx 2.055184 \mrevshort < 2.0552
\mrevshort$ as well as $\int_{\mrevshort/20}^{7\mrevshort/50} \pospart{R^\prime
(u) - 0.581} \, du \approx 0.006188 \mrevshort < 0.0062 \mrevshort$ we obtain
\[R(c) + K(c) \leq (2.0552+0.0878+0.0062) \mrevshort = 2.1492 \mrevshort <
\frac{43}{20} \mrevshort.\] The third case is $7\mrevshort/50 < c \leq
11\mrevshort/40$. We use $K(11\mrevshort/40) = 0 < 0.0001 \mrevshort$, because
for $c = 11\mrevshort/40$ the expression inside the positive part operator is
negative for every $u>0$, hence, the integrand is equal to zero and so is
$K(c)$. The line between the boundary points has slope
$\frac{0.0355-0.0001}{0.275-0.14} \approx 0.262222$ and upper-bounds every
$K(c)$ in this interval. Therefore, $K(c) \leq 0.0355 \mrevshort - 0.262
(c-7\mrevshort/50)$. Using $R(7\mrevshort/50) \approx 2.109831 \mrevshort <
2.1099 \mrevshort$ as well as $\int_{7\mrevshort/50}^{11\mrevshort/40}
\pospart{R^\prime (u) - 0.262} \, du \approx 0.004430 \mrevshort < 0.0045
\mrevshort$ we obtain
\[R(c) + K(c) \leq (2.1099+0.0355+0.0045) \mrevshort = 2.1499 \mrevshort <
\frac{43}{20} \mrevshort.\]

For the remaining case of large caps, i.e., $c > 11\mrevshort/40$ we choose
$\alpha = 0$ and the left-hand side of~\eqref{eq:mrev_integral_upper_bound}
simplifies to 
\[  \mrevshort + \int_c^\infty \pospart{ \phi \left( \frac{\mrevshort}{u}
\right)} \, du = \mrevshort + \int_c^\infty  \phi \left( \frac{\mrevshort}{u}
\right) \, du, \] where the equation holds as $\phi$ is nonnegative. Using
$\int_{11\mrevshort/40}^{\infty} \phi(\mrevshort / u) \, du \approx 1.141845
\mrevshort < 1.1419 \mrevshort$ we obtain
\[\mrevshort + \int_c^\infty  \phi \left( \frac{\mrevshort}{u} \right) \, du <
\mrevshort + 1.1419 \mrevshort = 2.1419 \mrevshort < \frac{43}{20} \mrevshort.\]

By this, all possible $c>0$ are covered and the claim follows.

\subsubsection{Proof of \texorpdfstring{\cref{claim:brev_iid}}{Claim 4.11}}
\label{appendix:claim_brev_iid}

We start by choosing the upper bound $U(u) = - \ln (1-\brevshort / u)$ if $u\geq
\brevshort 7/5$ and $U(u) = \bar{\kappa} \brevshort /u$ otherwise. Furthermore,
we set $\lambda = 5/(13\brevshort)$ and $\alpha$ in dependence of small or large
caps: For $0 < c \leq \brevshort 109/100$ we choose $\alpha =
1-3c/(8\brevshort)$ and $\alpha = 0$ otherwise. 

First we consider the case $0 < c \leq \brevshort 109/100$. Using $\pospart{x} = x + \pospart{-x}$ and defining
\begin{align}
    R(c) &=  \left(1-\alpha \frac{1-e^{-\lambda c}}{\lambda c}\right) \bar{\kappa} \brevshort - \frac{\alpha}{\lambda} \ln \left( \int_1^\infty e^{-\lambda \brevshort x} \frac{1}{x^2} \, dx \right) + \int_c^\infty \left[\phi \left( U(u) \right) - \alpha e^{-\lambda u} U(u)\right] \, du \\
    K(c) &= \int_0^\infty \pospart{ \alpha e^{-\lambda u} U(u) - \phi \left( U(u) \right) } \, du 
\end{align}
it suffices to show $R(c) + K(c) \leq {5}/{2} \brevshort$. Observe that, for
every fixed $u>0$, the term inside the $\pospart{\cdot}$ operator is linearly
decreasing in $c$. Taking the positive part and zero otherwise yields convexity
of the integrand. Integrating over $u$ preserves the convexity and monotonicity
of $K(c)$.

Next, we want to compute $R^\prime (c)$. We do this in two parts: Recall that $\alpha^\prime(c) = -3/(8\brevshort)$, then 
\begin{multline*}
    \frac{d}{dc} \left[ \left(1-\alpha \frac{1-e^{-\lambda c}}{\lambda c}\right) \bar{\kappa} \brevshort - \frac{\alpha}{\lambda} \ln \left( \int_1^\infty e^{-\lambda \brevshort x} \frac{1}{x^2} \, dx \right) \right] \\
    = \frac{\bar{\kappa} \brevshort (1-e^{-\lambda c})}{\lambda c ^2} - \frac{\bar{\kappa} \brevshort  \alpha e^{-\lambda c}}{c} + \frac{3}{8 \brevshort \lambda} \ln \left( \int_1^\infty e^{-\lambda \brevshort x} \frac{1}{x^2} \, dx \right),
\end{multline*}
For the remaining part we apply Leibniz's integral rule, since both the
integrand and its partial derivative with respect to $c$ decay in order
$O(u^{-2})$ as $u\to \infty$, the second part decays even exponentially in $u$.
By that we obtain
\[\frac{d}{dc} \left[ \int_c^\infty  \phi \left( U(u) \right) - \alpha e^{-\lambda u} U(u)\, du \right] = -\phi \left( U(c) \right) + \alpha e^{-\lambda c} U(c) + \frac{3}{8\brevshort} \int_c^\infty e^{-\lambda u} U(u) \, du .\]
As we consider $c \leq \brevshort 109/100 < \brevshort 7/5$, we have $U(c)=\bar{\kappa} \brevshort/c$. Therefore, summing both parts together, the derivative is
\begin{equation}
    R^\prime(c) = \frac{\bar{\kappa} \brevshort (1-e^{-\lambda c})}{\lambda c ^2}  + \frac{3}{8 \brevshort \lambda} \ln \left( \int_1^\infty e^{-\lambda \brevshort x} \frac{1}{x^2} \, dx \right)-\phi \left( \frac{\bar{\kappa} \brevshort}{c} \right) + \frac{3}{8 \brevshort} \int_c^\infty e^{-\lambda u} U(u) \, du .
\end{equation}

We now provide the upper bounds for $R(c) + K(c)$ for all $0 < c \leq \brevshort
109/100$, starting with the case $0 < c \leq \brevshort/5$. As $K$ is convex,
$K(0) \approx 0.657821 \brevshort < 0.6579 \brevshort$ and $K(\brevshort/5)
\approx 0.351571 \brevshort < 0.3517 \brevshort$. Clearly, the straight line
through these two upper bound points lies above $K(c)$ and has slope
$\frac{0.6579-0.3517}{0.2} = 1.531$. Therefore, $K(c) \leq 0.6579 \brevshort -
1.531 c$. By the fundamental theorem of calculus,  $R(c) = R(0) + \int_0^c
R^\prime (u) du$, it holds
\[
\begin{aligned}
    R(c) + K(c) &\leq R(0) + 0.6579 \brevshort + \int_0^c (R^\prime (u) - 1.531) \, du \\ &\leq R(0) + 0.6579 \brevshort + \int_0^{\brevshort/5} \pospart{R^\prime (u) - 1.531} \, du.
\end{aligned}
\]
We use $R(0) \approx 1.776756 \brevshort < 1.7768 \brevshort$ as well as
$\int_0^{\brevshort/5} \pospart{R^\prime (u) - 1.531} \, du \approx 0.062771
\brevshort < 0.0628 \brevshort$ and obtain
\[R(c) + K(c) \leq (1.7768 + 0.6579 + 0.0628) \brevshort = 2.4975 \brevshort <
\frac{5}{2} \brevshort.\] The next case is $\brevshort/5 < c \leq \brevshort
51/100$. We use $K(\brevshort 51/100) \approx 0.138020 \brevshort < 0.1381
\brevshort$. The line between the boundary points has slope
$\frac{0.3516-0.1381}{0.51-0.2} \approx 0.688710$ and upper bounds every $K(c)$
in this interval. Therefore, $K(c) \leq 0.3516 \brevshort - 0.6887
(c-\brevshort/5)$. Using $R(\brevshort/5) \approx 2.118436 \brevshort < 2.1185
\brevshort$ as well as $\int_{\brevshort/5}^{\brevshort 51/100}
\pospart{R^\prime (u) - 0.6887} \, du \approx 0.029543 \brevshort < 0.0296
\brevshort$ we obtain
\[R(c) + K(c) \leq (2.1185+0.3516+0.0296) \brevshort = 2.4997 \brevshort <
\frac{5}{2} \brevshort.\] The third case is $51 \brevshort/100 < c \leq
3\brevshort/4$. We use $K(\brevshort 3/4) \approx 0.061143 \brevshort < 0.0612
\brevshort$. The line between the boundary points has slope
$\frac{0.1381-0.0612}{0.75-0.51} \approx 0.320417$ and upper-bounds every $K(c)$
in this interval. Therefore, $K(c) \leq 0.1381 \brevshort - 0.3204 (c-\brevshort
51/100)$. Using $R(\brevshort 51/100) \approx 2.347989 \brevshort < 2.3480
\brevshort$ as well as $\int_{\brevshort 51/100}^{\brevshort 3/4}
\pospart{R^\prime (u) - 0.3204} \, du \approx 0.013564 \brevshort < 0.0136
\brevshort$ we obtain
\[R(c) + K(c) \leq (2.3480+0.1381+0.0136) \brevshort = 2.4997 \brevshort <
\frac{5}{2} \brevshort.\] The last small cap case is $\brevshort 3/4 < c \leq
\brevshort 109/100$. We use $K(\brevshort 109/100) \approx 0.004650 \brevshort <
0.0047 \brevshort$. The line between the boundary points has slope
$\frac{0.0612-0.0047}{1.09-0.75} \approx 0.166176$ and upper-bounds every $K(c)$
in this interval. Therefore, $K(c) \leq 0.0612 \brevshort - 0.166 (c-\brevshort
3/4)$. Using $R(\brevshort 3/4) \approx 2.435178 \brevshort < 2.4352 \brevshort$
as well as $\int_{\brevshort 3/4}^{\brevshort 109/100} \pospart{R^\prime (u) -
0.166} \, du \approx 0.003171 \brevshort < 0.0032 \brevshort$ we obtain
\[R(c) + K(c) \leq (2.4352+0.0612+0.0032) \brevshort = 2.4996 \brevshort < \frac{5}{2} \brevshort.\]

For the remaining case of large caps, i.e., $c > \brevshort 109/100$ we choose
$\alpha = 0$ and the left-hand side of~\eqref{eq:brev_integral_upper_bound}
simplifies to 
\[  \bar{\kappa} \brevshort + \int_c^\infty \pospart{ \phi \left( U(u) \right)}
\, du =  \bar{\kappa} \brevshort + \int_c^\infty \phi \left( U(u) \right) \, du
\leq \bar{\kappa} \brevshort + \int_{\brevshort 109/100}^\infty \phi \left( U(u)
\right) \, du , \] where the equation holds as $\phi$ is nonnegative. By our
choice for $U(u)$, we split the integral at $\brevshort 7/5$, where the
definition of $U(u)$ changes:
\[\bar{\kappa} \brevshort + \int_{\brevshort 109/100}^\infty \phi \left( U(u) \right) \, du = \bar{\kappa} \brevshort + \int_{\brevshort 109/100}^{\brevshort 7/5} \phi \left( \frac{\bar{\kappa} \brevshort}{u} \right) \, du+ \int_{\brevshort 7/5}^\infty \phi \left( -\ln \left( 1 - \frac{\brevshort}{u} \right) \right) \, du.\]
Using $\int_{\brevshort109/100}^{\brevshort 7/5} \phi( \bar{\kappa} \brevshort / u) \, du \approx 0.210914 \brevshort < 0.2110 \brevshort$ as well as $\int_{\brevshort 7/5}^{\infty} \phi(-\ln(1-\brevshort/u)) \, du \approx 0.498895 \brevshort \allowbreak < 0.4990 \brevshort$ we obtain
\[\begin{aligned}
\bar{\kappa} \brevshort + \int_{\brevshort 109/100}^{\brevshort 7/5} \phi \left( \frac{\bar{\kappa} \brevshort}{u} \right) \, du+ \int_{\brevshort 7/5}^\infty \phi \left( -\ln \left( 1 - \frac{\brevshort}{u} \right) \right) \, du &< 1.787 \brevshort + 0.2110 \brevshort + 0.4990 \brevshort \\ &= 2.497 \brevshort < \frac{5}{2} \brevshort.
\end{aligned}\]

By this, all possible $c>0$ are covered and the claim follows.

\section{A Lower Bound of \texorpdfstring{$2.19$}{2.19} for \texorpdfstring{$\mrev$}{max(SREV,BREV)}}
\label{sec:M-lower-bound}

As we discuss in p.~\pageref{comparison-concurrent-paragraph}, concurrent work
by~\textcite{huang2026improvedapproximation} already includes a lower-bound of
$2.11$, improving upon the long-standing $2$-factor bound
of~\textcite{Rubinstein2016}. Here, we present a further improved lower bound of
$2.19$. A key technical difference with our work is that Huang's lower  (similar
to Rubinstein's) is witnessed by a \emph{partition} mechanism pricing three
disjoint bundles. Our construction uses a more involved deterministic mechanism,
as a witness of the revenue-separation between \rev\ and \mrev, that combines an
entry-fee menu with a separately priced bundle. 

The key methodological ideas for the lower bound including, e.g., the
consideration of entry-fee mechanisms, were originally ours. However, we
experimented with the help of AI (see also our \hyperref[sec:ai-methodology]{AI
methodology statement}) in order to fine-tune the numerical parameters of the
construction, and to strike a satisfactory balance between the strength of the
bound and the transparency of its presentation. Throughout this exploration, the
AI agents were even able to identify instances achieving (slightly) better
bounds than the one we present below; however, we decided not to include them
here, due to their more obscure nature. Nevertheless, for follow-up work, it is
important to keep in mind that this suggests that the lower bound given
by~\cref{th:M-lower-bound} should be improvable.

\begin{theorem}
    \label{th:M-lower-bound}
    There exists a product distribution $\vec{F}$ such that 
\[ \rev(\vec{F}) > 2.19 \max\{\srev(\vec{F}),\brev(\vec{F})\}. \] Furthermore,
    this separation can be realized by a (truthful) deterministic mechanism. 
\end{theorem}

\begin{proof}
    We start by defining an instance of items and a buyer with independent
    distributions of how these items are valued. Afterwards, we quantify and
    bound $\srev$ and $\brev$ for this setting. Finally, proposing a mechanism,
    whose expected revenue serves as a lower bound for $\rev$ yields the lower
    bound.
    
    Let $A$ be a set of items with three groups of items, each group containing
    $N_a\in \N$ iid items.
    \begin{itemize}
        \item[$A_1$] (Free items) We assign each item a price $p_E=0$. The
        buyer's values are distributed on $V_1 = \sset{0} \cup
        \sset{\frac{5+k}{40} : k = 0,1, \dots, 275}$ with the following marginal
        \[\prob{X = 0} = 1 - \frac{12}{5 N_a}, \qquad \prob{X \geq v} =
        \frac{3}{10 N_a v}, v\in V_1\setminus \sset{0}.\] If sold separately,
        the $N_a$ items of $A_1$ yield expected revenue $\frac{12}{5} \cdot
        \frac{1}{8} = \frac{3}{10}$.
        \item[$A_2$] (Cheap items) We assign each item a price $p_E=1$. The
        buyer's values are distributed on $V_2 = \sset{0} \cup
        \sset{\frac{40+k}{40} : k = 0,1, \dots,  280}$ with the following
        marginal
        \[\prob{X = 0} = 1 - \frac{19}{20 N_a}, \prob{X \geq v} = \frac{19}{20 N_a v}, v\in V_2\setminus \sset{0}.\]
        If sold separately, the $N_a$ items of $A_2$ yield expected revenue $\frac{19}{20}$.
        \item[$A_3$] (Expensive items) We assign each item a price $p_E=15$. The
        buyer's values are distributed on $V_3 = \sset{0} \cup
        \sset{\frac{600+k}{40} : k = 0,1, \dots,  280}$ with the following
        marginal
        \[\prob{X = 0} = 1 - \frac{9}{40 N_a}, \prob{X \geq v} = \frac{27}{8 N_a
        v}, v\in V_3\setminus \sset{0}.\] If sold separately, the $N_a$ items of
        $A_3$ yield expected revenue $\frac{9}{40} \cdot 15 = \frac{27}{8}$.
    \end{itemize}
    Note that all items are distributed such that they achieve equal revenue for
    every posted price, $v \in V_j\setminus \sset{0}$. If all items of set $A$
    are sold separately, the expected revenue therefore is,
    \begin{equation}\label{eq:lower_bound_srev_of_A}
        S_A = \frac{3}{10} + \frac{19}{20} + \frac{27}{8} = 4.625.
    \end{equation}
    For our instance we choose $N_a = 1000$.
    
    Next we define a second set of items $B$. For some $N_b \in \N$ the set
    contains $N_b^2$ iid items, where each item is supported on 
    \[ V_B = \sset{k \frac{S_A}{N_b^2 H_{N_b}} , k \in [N_b]}, \] with $H_{N_b}$
    the $N_b$ harmonic number. The distribution is defined with equal revenue
    for selling at any supported value:
    \[\prob{X \geq k \frac{S_A}{N_b^2 H_{N_b}}} = \frac{1}{k}\] for all $k\in
    [N_b]$. Thus, selling all these items separately achieves expected revenue
    of number of items times the first support point, which is
    $\frac{S_A}{H_{N_b}}$. Altogether, this instance yields
    \[\srev = S_A + \frac{S_A}{H_{N_b}}.\] For an item of set $B$ and value
    given by random variable $Y$ it holds $\expect{Y} = \frac{S_A}{N_b^2}$,
    $\expect{Y^2} = \frac{S_A^2}{N_b^4 H_{N_b}^2} (2 N_b - H_{N_b})$, and
    therefore $\mathrm{Var}(Y) = \frac{S_A^2}{N_b^4 H_{N_b}^2} (2 N_b - H_{N_b})
    - \frac{S_A^2}{N_b^4}$.
    
    Let $Z_A$ be the sum of all random variables of item values in set $A$ as
    well as $Z_B$ the sum of all random variables of item values in set $B$.
    Therefore, $\brev = \sup_{p\geq 0} p \prob{ Z_A + Z_B \geq p}$. The items'
    tails of item set $A$ are chosen such that 
    \[ \sup_{p\geq 0} p \prob{ S_A + Z_A \geq p} = S_A. \]
    Now condition on a realization $Z_B = x$. If $x\leq S_A$, then $\sup_{p\geq 0} p \prob{ x + Z_A \geq p} \leq S_A$ and $\sup_{p\geq 0} p \prob{ x + Z_A \geq p} \leq x$ otherwise. Therefore, $\sup_{p\geq 0} p \prob{ x + Z_A \geq p} \leq \max \sset{S_A, x}$. Now varying over all realizations of $x$ gives
    \[\brev \leq \expect{ \max \sset{S_A, Z_B}} \leq S_A + \expect{ | Z_B - S_A
    |} \leq S_A + \sqrt{ \mathrm{Var}(Z_B) } = S_A + N_b \sqrt{\mathrm{Var}(Y)}.
    \] We use $\mathrm{Var}(Y) < \frac{2S_A^2}{N_b^3 H_{N_b}^2}$ and get
    \[\brev \leq S_A + N_b \sqrt{ \frac{2S_A^2}{N_b^3 H_{N_b}^2} } = S_A + \frac{S_A}{ H_{N_b}} \sqrt{ \frac{2}{N_b} }.\]
    Choosing $N_b>2$ implies
    \begin{equation}\label{eq:lower_bound_brev}
        \brev \leq S_A + \frac{S_A}{H_{N_b}} = \srev.
    \end{equation}
    
    We propose the following deterministic mechanism as a lower bound on the
    optimal truthful mechanism of this instance: All items of $B$ are offered in
    a bundle at price $p_B = S_A - S_A/H_{N_b}$. The probability of not selling
    at the proposed bundle price is 
    \[\prob{Z_B < p_B} \leq \prob{| Z_B - S_A | \geq \frac{S_A}{H_{N_b}}} \leq \frac{\mathrm{Var}(Z_B)}{(S_A/H_{N_b})^2} \leq \frac{2S_A^2}{N_b H_{N_b}^2}\frac{H_{N_b}^2}{S_A^2} = \frac{2}{N_b} \]
    where we use independence and Chebyshev's inequality. Therefore, the $B$ bundle contributes expected revenue of
    \[ \left(S_A - \frac{S_A}{H_{N_b}}\right) \left(1 - \frac{2}{N_b} \right).\]
    If the buyer wants to purchase an item of $A$, an entry fee of $7$ has to be paid. Then each item is offered at its price $p_E$.
    
    The buyer will pay the entry fee, if the total surplus from the item prices
    is at least $7$. Ignoring the entry restriction the expected payments are
    $S_A - 3/10$ since items in $A_1$ are sold at price $0$ and do not
    contribute to the expected revenue. If the buyer enters, the seller receives
    the entry fee $7$.  For the instance and the chosen item numbers the
    probability that the buyer does not enter is upper bounded by $19/28$. The
    expected loss of payments due to non-entering is upper bounded by $107/100$.
    In total the expected revenue $E$ of this entry-fee mechanism
    \[E > S_A + 7 \left( 1 - \frac{19}{28} \right) - \frac{107}{100} - \frac{3}{10} = S_A + \frac{22}{25}.\]
    
    Taking $N_b \to \infty$, the proposed mechanism has expected revenue
    \[E+ \left(S_A-\frac{S_A}{H_{N_b}}\right)\left(1-\frac{2}{N_b}\right) \xrightarrow{N_b \to \infty} E+S_A >  2S_A + \frac{22}{25}.\]
    By~\eqref{eq:lower_bound_brev} we have to compare this value only to $\srev$ as on this instance it achieves at least as much revenue as $\brev$, i.e., $\mrev = \srev$. This means
    \[\frac{\rev}{\mrev} = \frac{\rev}{\srev} \geq \frac{E + \left(S_A - \frac{S_A}{H_{N_b}} \right)\left(1 - \frac{2}{{N_b}} \right)}{S_A + \frac{S_A}{H_{N_b}}}\]
    Furthermore, we use the exact value of $S_A$ in~\eqref{eq:lower_bound_srev_of_A} and for sufficiently large $N_b$ this right-hand side is arbitrarily close to
    \[\frac{E + S_A}{S_A} > \frac{2S_A + 22/25}{S_A} = \frac{9.25 + 0.88}{4.625} \approx 2.190270 > 2.19.\]
    This proves the statement.
\end{proof}

\section{Omitted Proofs from \texorpdfstring{\cref{sec:selling_separately_2_iid}}{Section 5}}
\label{append:sec:iid-results}

\subsection{Proof of
\texorpdfstring{\cref{lemma:first_order_stochastic_dominance_er_anchored}}{Lemma
5.3}}
\label{append:lemma:first_order_stochastic_dominance_er_anchored}

Fix an arbitrary $r\in[K]$, let $f$ be the probability mass function of an
arbitrary distribution anchored in $r$ and $\bar{f}$ the probability mass
function of $\er_K(r)$. We have to show that
\[\sum_{k=s}^K \bar{f}(k) \geq \sum_{k=s}^K f(k) \] for every $s\in [K]$. By
definition, \textit{every} tail of an $\er_K(r)$ distribution satisfies
$\sum_{k=s}^K \bar{f}(k) = \frac{r}{s}$ for $s\geq r$ and is equal to one
otherwise, which reduces the inequality to
\[\frac{r}{s} \geq \sum_{k=s}^K f(k). \] Multiplying both sides with $s$
this is by definition satisfied by $f$ for every $s$.

\subsection{Proof of \texorpdfstring{\cref{lemma:omega_K_properties}}{Lemma 5.4}}
\label{append:lemma:omega_K_properties}

Fix $K\geq 1$ and let $1\leq s\leq k \leq K$. For a fixed $k \leq K$, let
$s^\star$ be the maximizer
\[s^\star = \argmax_{1\leq s \leq k} \frac{s(H_k-H_s)}{k+s+1}.\] We use $H_{k+1}
= H_k + 1/(k+1)$ and compute the difference
\[ \frac{s^\star(H_{k+1}-H_{s^\star})}{(k+1)+s^\star+1} -
\frac{s^\star(H_k-H_{s^\star})}{k+s^\star+1} = \frac{s^\star}{k+1} \left[
\frac{s^\star(H_k-H_{s^\star})}{k+s^\star+1} - \frac{(s^\star+1)
(H_{k}-H_{s^\star +1})}{k+(s^\star+1)+1} \right] \geq 0.\] Thus, monotonicity
holds for all $k$ and the expression is maximized for $k=K$ which is maximized
for $s^\star = r^\star$ and the first inequality follows.

For every $r\in[K]$,
\[H_K-H_r=\sum_{h=r+1}^K\frac{1}{h}\leq \int_r^K\frac{1}{x}\,dx = \ln \left( \frac{K}{r} \right).\]
Therefore,
\[\frac{r(H_K-H_r)}{K+r+1}\leq\frac{r\ln \left( \frac{K}{r}
\right)}{K+r+1}\leq\frac{\frac rK\ln \left( \frac{K}{r} \right)}{1+\frac rK}. \]
We substitute $x=r/K$ and consider the following function and its derivative for
$x\in (0,1]$ :
\[g(x):=\frac{x \ln \left( \frac{1}{x} \right)}{1+x}, \qquad g'(x)=\frac{\ln
\left( \frac{1}{x} \right)-1-x}{(1+x)^2}.\] The numerator of $g(x)$ goes to zero
as $x\to 0^{+}$; furthermore, $g(1)=0$ and $g(x)>0$ for all $x \in (0,1)$. Thus,
its unique interior maximizer satisfies
\[\ln(1/x)=1+x \quad \Longleftrightarrow \quad xe^x=\frac{1}{e},
\]
which directly gives $x=W(1/e)$.
Consequently, for every $K$
\[\omega_K\leq W(1/e).\]

For the convergence we set $x^\star \coloneq W(1/e)$ and choose $r_K=\lfloor x^\star K\rfloor$.
Then the three limits hold:
\[\begin{aligned}
&\lim_{K\to\infty}\frac{r_K}{K}= x^\star, \\
&\lim_{K\to\infty}\frac{K+r_K+1}{K} =1+x^\star, \\
&\lim_{K\to\infty}\left(H_K-H_{r_K}\right)=\lim_{K\to\infty} \left( \sum_{h=r_K+1}^K\frac{1}{h} \right) = \ln \left( \frac{1}{x^\star} \right).
\end{aligned}\]
With these limits and as $x^\star$ satisfies $\ln(1/x^\star)=1+x^\star,$
\[\lim_{K\to\infty}\frac{r_K(H_K-H_{r_K})}{K+r_K+1}=\lim_{K\to\infty}\frac{\frac{r_K}{K}(H_K-H_{r_K})}
{1+\frac{r_K}{K}+\frac1K} = \frac{x^\star\ln(1/x^\star)}{1+x^\star} = x^\star.
\] By definition of $\omega_K$ (\cref{def:discrete_lambert_approximate}) as the
maximum choice of all possible $r$, it is at least as large as for the choice
$r_K$ and the limit follows.

\subsection{Proof of
\texorpdfstring{\cref{lemma:equal_revenue_tightness}}{Proposition 5.6}}
\label{append:lemma:equal_revenue_tightness}

Fix $r\in[K]$ and let $X_1,X_2$ be independently distributed according
to $\er_K(r)$. Recall that
\[ \srev\left( \er_K(r) \times \er_K(r)\right)=\frac{2r}{K}.\] For the
index-notation we define $J_1 := KX_1$ and $J_2 := KX_2$. We next consider
selling the two items as a bundle at price $\frac{K+r}{K}.$ The bundle is
sold at this price if and only if $J_1+J_2\geq K+r$. Conditioning on $J_1=k$
and independence gives
\begin{align*}
    \prob{J_1+J_2\geq K+r} &= \sum_{k=r}^K f(k)\prob{J_2\geq K+r-k}\\
    &= \frac{r}{K} + \sum_{k=r}^{K-1} \frac{r}{k(k+1)} \frac{r}{K+r-k}\\
    &= \frac{r}{K} + r^2\sum_{k=r}^{K-1} \frac{1}{k(k+1)(K+r-k)}.
\end{align*}
For the sum over $r\leq k<K$, we use
\[\frac{1}{k(k+1)(K+r-k)} = \frac{1}{K+r}\frac{1}{k}-\frac{1}{K+r+1}\frac1{k+1}
+\frac{1}{(K+r)(K+r+1)} \frac{1}{K+r-k}. \]
Therefore, applying the sum gives the harmonic number expressions
\begin{align*}
    \sum_{k=r}^{K-1} \frac{1}{k(k+1)(K+r-k)} &= \frac{1}{K+r}\sum_{k=r}^{K-1} \frac{1}{k} - \frac{1}{K+r+1} \sum_{k=r}^{K-1} \frac{1}{k+1} \\ &\qquad+ \frac{1}{(K+r)(K+r+1)} \sum_{k=r}^{K-1} \frac{1}{K+r-k} \\
    &= \frac{1}{K+r} (H_{K-1}-H_{r-1}) - \frac{1}{K+r+1} (H_K-H_r) \\ &\qquad+ \frac{1}{(K+r)(K+r+1)} (H_K-H_r) \\
    &= \frac{1}{K+r} \left(H_{K}-\frac{1}{K}-H_{r}+\frac{1}{r}\right) - \frac{1}{K+r+1} (H_K-H_r) \\ &\qquad+ \frac{1}{(K+r)(K+r+1)} (H_K-H_r) \\
    &= \frac{2(H_K-H_r)}{(K+r)(K+r+1)} + \frac{1}{K+r}\left(\frac{1}{r}-\frac{1}{K}\right).
\end{align*}
Consequently,
\begin{align*}
    \prob{J_1+J_2\geq K+r} &= \frac{r}{K} + r^2 \left[ \frac{2(H_K-H_r)}{(K+r)(K+r+1)} + \frac{1}{K+r}\left(\frac{1}{r}-\frac{1}{K}\right)\right] \\
    &= \frac{r}{K} + \frac{2r^2(H_K-H_r)}{(K+r)(K+r+1)} + \frac{r}{K+r} - \frac{r^2}{K(K+r)}\\
    &= \frac{2r^2(H_K-H_r)}{(K+r)(K+r+1)} +\frac{r}{K} \left(1-\frac{r}{K+r}\right) + \frac{r}{K+r} \\
    &= \frac{2r^2(H_K-H_r)}{(K+r)(K+r+1)} +\frac{2r}{K+r} \\
    &= \frac{2r}{K+r} \left[1+\frac{r(H_K-H_r)}{K+r+1}\right]. \\
\end{align*}  
By selling the bundle at price $(K+r)/K$ the bundling revenue is
\[\frac{K+r}{K} \prob{J_1+J_2\geq K+r} = \frac{2r}{K} \left[ 1+\frac{r(H_K-H_r)}{K+r+1} \right]. \]
Now choose $r=r^\star$. By the definition of $r^\star$ and $\rev \geq \brev$
\[ \rev\bigl(\er_K(r^\star) \times \er_K(r^\star)\bigr) \geq
\frac{2r^\star}{K}(1+\omega_K)  = (1+\omega_K)\, \srev\bigl(\er_K(r^\star)
\times \er_K(r^\star)\bigr). \] Since $\er(r)$ is anchored in $r$, this
completes the proof.

\subsection{Proof of \texorpdfstring{\cref{lemma:two_iid_anchored_upper_bound}}{Lemma 5.5}}
\label{append:lemma:two_iid_anchored_upper_bound}

The proof is structured in the following way. First we verify that the chosen
multipliers are indeed feasible. Then we show that by this choice the Lagrangian
reduces to a constant term, and finally, we quantify this term.

\paragraph{Feasibility of the multipliers} The multipliers of the rectangular,
circular as well as the diagonal flow layer are clearly nonnegative as by
considering an anchored distribution, it holds $h\sum_{l=h}^Kf(l)\le r$. It
therefore remains only to verify nonnegativity of the boundary multipliers,
i.e., nonnegativity of $\frac{2r}{K}\omega_K - B(f;s)$ for all $s=r, \dots, K$. 

We start by bounding $B(f;s)$. For this recall the definition of the discrete
equal-revenue distribution $\bar f$ anchored in $r$
(see~\cref{def:distribution_equal_revenue_r}), given by
\[\bar{f}(k)=
\begin{cases}
    \dfrac{r}{k(k+1)}, & r\leq k< K,\\[1ex]
    \dfrac{r}{K}, & k=K.
\end{cases}
\] Its tail probabilities satisfy $\sum_{k=s}^K \bar{f}(k)=\frac{r}{s}$ for
every $s=r,\ldots,K$. On the other hand, since $f$ is anchored in $r$, by
\cref{lemma:first_order_stochastic_dominance_er_anchored} $\bar{f}$ first-order
stochastically dominates $f$. Define the function
\[g(k_1,k_2;s) \coloneq \max\{k_1+k_2-r-s,0\} - \max\{k_1-s,0\} -
\max\{k_2-s,0\}.\] Recall the definition of $B(f;s)$ of
\eqref{def:boundary_correction_term} and observe that by the independence of the
two items this captures
\begin{equation}\label{eq_B_function_as_expectation}
    K  \cdot B(f;s) = \expect{g(K\cdot X_1,K\cdot X_2;s)}.
\end{equation}
Furthermore, $g$ is nondecreasing in $k_1$ and $k_2$. For this, we need to
verify that
\[g(k_1+1,k_2;s)-g(k_1,k_2;s)=\indicator \{k_1+k_2\ge r+s\}-\indicator \{k_1\ge
s\}\] is nonnegative, where $\indicator$ denotes the usual indicator function.
For $k_1<s$ we have $\indicator \{k_1\ge s\}=0$ and for $k_1\geq s$ the
difference is zero, since every value in the support satisfies $k_2\geq r$. By
symmetry the same holds when increasing $k_2$. Let $X_1,X_2$ be independent with
marginal $f$, and let $\bar{X}_1,\bar{X}_2$ be independent with marginal
$\bar{f}$. Since $g$ is nondecreasing in each coordinate and $\bar f$
first-order stochastically dominates $f$ it follows
\[K \cdot B(f;s) = \expect{g(K\cdot X_1,K\cdot X_2;s)} \leq \expect{g(K\cdot
\bar{X}_1,K\cdot \bar{X}_2;s)} = K \cdot B(\bar{f};s).\]

Now observe that pointwise for $k_1,k_2=r\dots, K$
\[ \max\{k_1+k_2-r-s,0\}=\max\{k_1-s,0\}+\sum_{h=r+1}^K\indicator\{k_2\ge h\}
\indicator\{k_1\ge r+s+1-h\}.\] For $k_1\geq s$ the second indicator is one for
all $h\geq r+1$. Hence, the sum counts exactly the integers between $r+1$ and
$k_2,$ so it equals $k_2-r$. If otherwise $k_1< s$, the maximum expression on
the right-hand side is zero and the indicators are one exactly when
$r+s+1-k_1\le h\le k_2$. Therefore, the sum counts the integers in this
interval, namely $\max\{k_1+k_2-r-s,0\}$. Taking expectations, we obtain 
\begin{align*}
    \expect{\max\{X_1+X_2-r-s,0\}}= &\expect{\max\{X_1-s,0\}}\\
    &+\sum_{h=r+1}^K \prob{X_2\geq h} \prob{X_1\ge r+s+1-h}.
\end{align*}
Next, we write $K \cdot B(\bar{f};s)$ as in \eqref{eq_B_function_as_expectation}
and obtain by the identity we just showed
\begin{align*}
    K \cdot B(\bar{f};s) &= \sum_{h=r+1}^K \prob{\bar{X}_2\geq h} \prob{\bar{X}_1\ge r+s+1-h} - \expect{\max\{\bar{X}_2-s,0\}} \\
    &= \sum_{h=r+1}^K \prob{\bar{X}_2\geq h} \prob{\bar{X}_1\ge r+s+1-h} - \sum_{h=s+1}^K \prob{\bar{X}_2\geq h} \\
    &=\sum_{h=r+1}^s \prob{\bar{X}_2\geq h} \prob{\bar{X}_1\ge r+s+1-h},
\end{align*}
where we use in the last step that for every $h\ge s+1$, it holds $r+s+1-h\le
r$. We insert the exact probability tails of the equal-revenue distribution
anchored in $r$ and obtain
\[K \cdot B(\bar{f};s) = \sum_{h=r+1}^s \frac{r}{h} \frac{r}{r+s+1-h} = r^2 \sum_{h=r+1}^s \frac{1}{h(r+s+1-h)}.\]
With the partial fraction identity $\frac{1}{h(r+s+1-h)}=\frac{1}{r+s+1}\left(\frac1h+\frac{1}{r+s+1-h}\right)$ and the fact that the second fraction in this sum also counts exactly over $\frac{1}{h}$ in reverse order we get
\begin{align*}
    K  \cdot B(\bar{f};s) &= \frac{r^2}{r+s+1} \sum_{h=r+1}^s \left(\frac1h+\frac{1}{r+s+1-h}\right) \\
    &= \frac{2r^2}{r+s+1} \sum_{h=r+1}^s \frac{1}{h} \\
    &= \frac{2r^2}{r+s+1} \left(H_s - H_r\right).
\end{align*}
Applying~\cref{lemma:omega_K_properties}, we obtain
\begin{align*}
    B(f;s) \leq B(\bar{f};s) = \frac{2r}{K} \left(\frac{r(H_s - H_r)}{r+s+1}\right) \leq \frac{2r}{K} \omega_K.
\end{align*}
By this we have proven that the defined multipliers are certainly nonnegative.
We now evaluate the Lagrangian under these multipliers and show that all primal
variables cancel, leaving only the constant $1+\omega_K$. Already at the start
we confirmed that the coefficient of $\rho$ equals zero as its coefficients are
\[1-\theta(r,r) \srev(F\times F) = 1 - \frac{K}{2r}\frac{2r}{K}=0.\]

\paragraph{Payment variables' coefficients} By the choice of $\theta$
multipliers, for any grid point $(k_1,k_2)$ the coefficient of $p(k_1,k_2)$ is
\[\frac{K}{2r} f(k_1) f(k_2) + \sum_{(k_1^\prime,k_2^\prime)}
\lambda_{(k_1^\prime,k_2^\prime)\to (k_1,k_2)} - \sum_{(k_1^\prime,k_2^\prime)}
\lambda_{(k_1,k_2)\to (k_1^\prime,k_2^\prime)} - \mu(k_1,k_2).\] Note, that
every circular edge carries an equal amount of flow in both directions, hence,
has zero net contribution to every payment coefficient. Thus, circular flows can
be omitted in the analysis of the payment variables. 

First, consider any interior point, i.e., $(k_1,k_2)$ with $k_1,k_2>r$, and the
coefficients arising by the rectangular flows. For $r<k_1<K$ this is
\[\lambda_{(k_1+1,k_2)\to (k_1,k_2)}^\textup{rect} - \lambda_{(k_1,k_2)\to
(k_1-1,k_2)}^\textup{rect} = \frac{K}{2r} f(k_2) \sum_{l=k_1+1}^K f(l)
-\frac{K}{2r} f(k_2) \sum_{l=k_1}^K f(l) = - \frac{K}{2r} f(k_1)f(k_2).\] For
$k_1 = K$ it reduces only to $- \lambda_{(K,k_2)\to (K-1,k_2)}^\textup{rect} =
-\frac{K}{2r} f(K)f(k_2)$. The same holds for the rectangular flow from and to
the neighbouring grid points in the dimension of $k_2$, thus, we subtract the
resulting term twice. Now from the perspective of interior points the diagonal
flows are only the sinks. The two incoming flows add up to
\begin{equation}\label{eq:incoming_diagonal_mass}
    \frac{K}{2r} f(k_1) f(k_2) \frac{k_1-r}{k_1+k_2-2r} + \frac{K}{2r} f(k_1) f(k_2) \frac{k_2-r}{k_1+k_2-2r} = \frac{K}{2r} f(k_1) f(k_2).
\end{equation}
In total the coefficient of $p(k_1,k_2)$ is the initial source, subtracting the
value from the rectangular flows twice, and adding the contribution of the
diagonal flows, i.e.,
\[\frac{K}{2r} f(k_1) f(k_2) - 2\frac{K}{2r} f(k_1) f(k_2) + \frac{K}{2r} f(k_1) f(k_2) = 0.\]

Next, we consider the boundary points $(k,r)$ with $r<k<K$. The initial source
mass and the rectangular flow contribution is
\[\begin{aligned}
&\frac{K}{2r} f(k) f(r) + \lambda_{(k+1,r)\to (k,r)}^\textup{rect} - \lambda_{(k,r)\to (k-1,r)}^\textup{rect} + \lambda_{(k,r+1)\to (k,r)}^\textup{rect} \\
&= \frac{K}{2r} \left[ f(k) f(r) + f(r) \sum_{l=k+1}^K f(l) -  f(r) \sum_{l=k}^K f(l) +  f(k) \sum_{l=r+1}^K f(l) \right] \\
&=\frac{K}{2r} f(k) \sum_{l=r+1}^K f(l).
\end{aligned}
\] In the diagonal flow layer, $(k,r)$ serves as the source for all points on
the diagonal to $(r,k)$, i.e., to all $(k_1,k_2)$ with $k_1+k_2 = k+r$, while
the self-loop $(k,r)\to(k,r)$ is omitted. Hence, the total outgoing flow is
\[\sum_{\substack{k_1+k_2 = k+r,\\ k_1\neq k}} \lambda_{(k,r)\to
(k_1,k_2)}^\textup{diag} = \frac{K}{2r}\sum_{\substack{k_1+k_2 = k+r,\\ k_1\neq
k}} f(k_1) f(k_2) \frac{k_1-r}{k-r}. \] Using symmetry of the diagonal we can
write 
\[
\begin{aligned}
    &2\sum_{k_1+k_2 = k+r} f(k_1) f(k_2) (k_1-r) = \sum_{k_1+k_2 = k+r} f(k_1) f(k_2) (k_1-r) + \sum_{k_1+k_2 = k+r} f(k_1) f(k_2) (k_2-r) \\
    &= \sum_{k_1+k_2 = k+r} f(k_1) f(k_2) (k_1-r + k_2 -r) = \sum_{k_1+k_2 = k+r} f(k_1) f(k_2) (k-r).
\end{aligned}\]
Subtracting the loop afterwards, we obtain the sum of diagonal flows as
\begin{equation}\label{eq:outgoing_diagonal_flow}
    \sum_{\substack{k_1+k_2 = k+r,\\ k_1\neq k}} \lambda_{(k,r)\to (k_1,k_2)}^\textup{diag} = \frac{K}{2r} \left[ \frac{1}{2} \sum_{k_1+k_2 = k+r} f(k_1) f(k_2) - f(k)f(r)\right].
\end{equation}
By the definition of $B(f;k)$ in~\eqref{def:boundary_correction_term} we can
compute the difference $K\cdot B(f;k)-K\cdot B(f;k-1)$ as
\begin{equation}\label{eq:B_function_difference}
    \begin{aligned}
        &\sum_{k_1+k_2\geq r+k+1} (k_1+k_2 - r-k) f(k_1)f(k_2) - 2 \sum_{s=k+1}^K (s-k) f(s) \\ &\qquad- \sum_{k_1+k_2\geq r+k} (k_1+k_2 - r-k+1) f(k_1)f(k_2) + 2 \sum_{s=k}^K (s-k+1) f(s) \\
        &= - \sum_{k_1+k_2\geq r+k} f(k_1)f(k_2) + 2  \sum_{s=k}^K f(s).
    \end{aligned}
\end{equation}
Hence, the total net boundary flow at $(k,r)$ is
\begin{equation}\label{eq:boundary_flow_item2_at_r}
    \begin{aligned}
        & \lambda_{(k+1,r)\to (k,r)}^\textup{bdry} +\lambda_{(k-1,r)\to (k,r)}^\textup{bdry} - \lambda_{(k,r)\to (k+1,r)}^\textup{bdry}-\lambda_{(k,r)\to (k-1,r)}^\textup{bdry}\\
        &=\frac{K}{2r}\frac{K}{2}
        \left[B(f;k+1)-B(f;k) - \left(B(f;k)-B(f;k-1)\right)\right] \\
        &=     \frac{K}{2r}\frac{1}{2}\left[- \sum_{k_1+k_2\geq r+s+1} f(k_1)f(k_2) + 2  \sum_{k=s+1}^K f(k) + \sum_{k_1+k_2\geq r+s} f(k_1)f(k_2) - 2  \sum_{k=s}^K f(k)\right] \\
        &=\frac{K}{2r}\left[ \frac{1}{2} \sum_{k_1+k_2 = r+s} f(k_1)f(k_2) - f(s)\right].
    \end{aligned}
\end{equation}
Note that the $\frac{2r}{K} \omega_K$ is constant for all
$\lambda^\textup{bdry}$ and cancels directly whenever the sum of
$\lambda^\textup{bdry}$ coefficients equals zero. In total, we have a
coefficient of
\begin{align*}
    \frac{K}{2r} f(k) \sum_{l=r+1}^K f(l) - \frac{K}{2r} &\left[ \frac{1}{2}
    \sum_{k_1+k_2 = k+r} f(k_1) f(k_2) - f(k)f(r)\right]\\ 
    &+\frac{K}{2r}\left[ \frac{1}{2} \sum_{k_1+k_2 = k+r} f(k_1)f(k_2) - f(k)\right],
\end{align*}
which equals zero. Analogously, this holds for any $(r,k)$ with $r<k<K$.

At grid point $(K,r)$ source mass and rectangular flow net to
\[\frac{K}{2r} f(K) f(r) - \lambda_{(K,r)\to (K-1,r)}^\textup{rect} +
\lambda_{(K,r+1)\to (K,r)}^\textup{rect} = \lambda_{(K,r+1)\to
(K,r)}^\textup{rect} = \frac{K}{2r} f(K) \sum_{l=r+1}^K f(l). \] $(K,r)$ serves
as the source for diagonal flows to every upper triangular grid point, i.e., for
every $(k_1,k_2)$ with $k_1+k_2\geq K+r$. By the same symmetry argument as for
other boundary sources (see~\eqref{eq:outgoing_diagonal_flow}), now only for
more target points, the total outgoing diagonal flow is
\[\sum_{\substack{k_1+k_2 \geq K+r,\\ (k_1,k_2)\neq (K,r)}} \lambda_{(K,r)\to (k_1,k_2)}^\textup{diag} = \frac{K}{2r} \left[ \frac{1}{2} \sum_{k_1+k_2 \geq K+r} f(k_1) f(k_2) - f(K)f(r)\right].\]
The net boundary flow in $(K,r)$ is
\[\lambda_{(K-1,r)\to (K,r)}^\textup{bdry}-\lambda_{(K,r)\to (K-1,r)}^\textup{bdry} = \frac{K}{2r}\frac{K}{2} (B(f;K-1) - B(f;K))\]
and using \eqref{eq:B_function_difference} the total flow is
\begin{multline*}
    \frac{K}{2r} \left[ f(K) \sum_{l=r+1}^K f(l) - \frac{1}{2} \sum_{k_1+k_2 \geq K+r} f(k_1) f(k_2) + f(K)f(r)\right.\\ 
    \left. + \frac{1}{2} \sum_{k_1+k_2\geq r+K} f(k_1)f(k_2) -   \sum_{k=K}^K f(k)  \right],
\end{multline*}
which clearly equals zero. Analogously, this holds for $(r,K)$.

Now at $(r,r)$ we have initial source mass and rectangular flows
\begin{align*}
    \frac{K}{2r} f(r) f(r) + \lambda_{(r+1,r)\to (r,r)}^\textup{rect} + \lambda_{(r,r+1)\to (r,r)}^\textup{rect} &= \frac{K}{2r} \left[ f(r) f(r) + 2 f(r) \sum_{l=r+1}^K f(l) \right]\\ 
    &= \frac{K}{2r} \left[2f(r) - f(r)f(r)\right].
\end{align*}
The grid point is not involved in any diagonal flow. Using
\eqref{eq:B_function_difference} we simplify the boundary flows
\[\begin{aligned}
\lambda_{(r+1,r)\to (r,r)}^\textup{bdry}+\lambda_{(r,r+1)\to (r,r)}^\textup{bdry} -& \lambda_{(r,r)\to (r+1,r)}^\textup{bdry}-\lambda_{(r,r)\to (r,r+1)}^\textup{bdry} \\
&= \frac{K}{2r} K \left[ B(f;r+1) - B(f;r) \right] \\
&= \frac{K}{2r} \left[- \sum_{k_1+k_2\geq 2r+1} f(k_1)f(k_2) + 2  \sum_{k=r+1}^K f(k) \right]\\
&= \frac{K}{2r} \left[f(r)f(r)-1 + 2  (1-f(r)) \right]\\
&= \frac{K}{2r} \left[1 + f(r)f(r) - 2f(r) \right].
\end{aligned} \]
Adding all terms up, we obtain
\[\frac{K}{2r} (2f(r) - f(r)f(r)) + \frac{K}{2r} \left[1 + f(r)f(r) - 2f(r))
\right] = \frac{K}{2r}.\] Finally, since $\mu(r,r)=\frac{K}{2r}$, the last
coefficient of the payment variables cancels.

\paragraph{Allocation variables' coefficients} Next we do the same for the
allocation variables' coefficients. None of these variables has an initial
source mass as this is linked to the $\theta$ constraints which do not affect
allocation variables at all. Therefore, the total contribution to any
$a_1(k_1,k_2)$ is
\[ \sum_{(k_1^\prime,k_2^\prime)} \frac{k_1}{K}\lambda_{(k_1,k_2)\to
(k_1^\prime,k_2^\prime)} - \sum_{(k_1^\prime,k_2^\prime)} \frac{k_1^\prime}{K}
\lambda_{(k_1^\prime,k_2^\prime)\to (k_1,k_2)} + \frac{k_1}{K} \mu(k_1,k_2) -
\psi_1(k_1,k_2) + \alpha_1(k_1,k_2). \] Note, that we will do the analysis from
the perspective of item one's allocation variables. Symmetrically any argument
holds analogously for item two. Furthermore, observe that then the contribution
of circular flows in the direction of item two is zero as it is multiplied with
the constant value of item one. Circular flows in the dimension of item one's
value have the same flow mass, but the values differ by $\frac{1}{K}$. Thus, we
can simplify 
\[\frac{k_1}{K}\lambda_{(k_1,k_2)\to (k_1-1,k_2)}^\textup{circ} -
\frac{k_1-1}{K}\lambda_{(k_1-1,k_2)\to (k_1,k_2)}^\textup{circ} =
\frac{1}{K}\lambda_{(k_1,k_2)\to (k_1-1,k_2)}^\textup{circ}\] where we
conventionally write the downwards multiplier.

We start with an arbitrary variable $a_1(k_1,k_2)$ with $r<k_1<K$ and $k_2>r$.
The horizontal rectangular and circular flows contribute
\begin{equation} \label{eq:horzontal_allocation_coefficient}
    \begin{aligned}
        &\frac{k_1}{K}\lambda_{(k_1,k_2)\to (k_1-1,k_2)}^\textup{rect}  - \frac{k_1+1}{K} \lambda_{(k_1+1,k_2)\to (k_1,k_2)}^\textup{rect}  + \frac{1}{K}\lambda_{(k_1,k_2)\to (k_1-1,k_2)}^\textup{circ} - \frac{1}{K}\lambda_{(k_1+1,k_2)\to (k_1,k_2)}^\textup{circ} \\
        &=\frac{f(k_2)}{2r} \left[ k_1 \sum_{l=k_1}^K f(l)  - (k_1+1) \sum_{l=k_1+1}^K f(l) + \left(r - k_1\sum_{l=k_1}^K f(l)\right) - \left(r - (k_1+1)\sum_{l=k_1+1}^K f(l)\right) \right]
    \end{aligned}
\end{equation}
which equals zero. Vertically, all circular flows cancel, and the rectangular
flows contribute
\begin{align*}
    \frac{k_1}{K}\left(\lambda_{(k_1,k_2)\to (k_1,k_2-1)}^\textup{rect} - \lambda_{(k_1,k_2+1)\to (k_1,k_2)}^\textup{rect}\right) &=\frac{k_1}{K}\frac{K}{2r} \left[ f(k_1) \sum_{l=k_2}^K f(l) - f(k_1) \sum_{l=k_2+1}^K f(l)\right]\\ & =\frac{k_1}{K} \frac{K}{2r} f(k_1) f(k_2).
\end{align*}
Note, that for $k_2=K$, the second rectangular term is absent, and the same
final expression is obtained.

For grid points in the lower triangle there are two incoming flows from
$(k_1+k_2-r,r)$ and from $(r,k_1+k_2-r)$. Together they contribute the same mass
as in \eqref{eq:incoming_diagonal_mass}, now weighted by the corresponding value
\[
\begin{aligned}
    &-\frac{k_1+k_2-r}{K}\frac{K}{2r}f(k_1)f(k_2)\frac{k_1-r}{k_1+k_2-2r}-\frac{r}{K}\frac{K}{2r}f(k_1)f(k_2)\frac{k_2-r}{k_1+k_2-2r} \\
    &=-\frac{K}{2r}\frac{1}{K}f(k_1)f(k_2) \frac{k_1^2+k_1 k_2-rk_1-rk_1-rk_2+r^2+rk_2-r^2}{k_1+k_2-2r} \\
    &=-\frac{K}{2r}\frac{1}{K}f(k_1)f(k_2) \frac{(k_1+k_2-2r)k_1}{k_1+k_2-2r}  \\&= -\frac{K}{2r} \frac{k_1}{K}f(k_1)f(k_2)
\end{aligned}
\]
which directly cancels with the previous flow. For the upper triangular part all
incoming flows are from $(K,r)$ and $(r,K)$. The values change and this time we
also have to subtract the corresponding $\psi_1^\mathrm{diag}$ multiplier. Thus,
together they contribute
\[
\begin{aligned}
    &-\frac{K}{2r}f(k_1)f(k_2)\left[ \frac{k_1-r}{k_1+k_2-2r}+\frac{r}{K}\frac{k_2-r}{k_1+k_2-2r} + \frac{k_1-r}{k_1+k_2-2r} \frac{k_1+k_2-r-K}{K}\right] \\
    &=-\frac{K}{2r}f(k_1)f(k_2)\frac{1}{k_1+k_2-2r} \left[ \frac{K}{K}(k_1-r)+\frac{r}{K}(k_2-r) +\frac{k_1-r}{K} (k_1+k_2-r-K) \right] \\
    &=-\frac{K}{2r}f(k_1)f(k_2)\frac{1}{k_1+k_2-2r} \left[ \frac{(k_1-r)(K+k_1+k_2-r-K)+r(k_2-r)}{K} \right] \\
    &=-\frac{K}{2r}f(k_1)f(k_2)\frac{1}{k_1+k_2-2r} \left[ \frac{k_1^2+k_1k_2-k_1r-k_1r-k_2r+r^2+k_2r-r^2}{K} \right] \\
    &=-\frac{K}{2r}f(k_1)f(k_2)\frac{1}{k_1+k_2-2r} \left[ \frac{k_1^2+k_1k_2-2k_1r}{K} \right] = -\frac{K}{2r}\frac{k_1}{K}f(k_1)f(k_2) 
\end{aligned}
\]
which cancels with the previous flow. Next we consider $a_1(K,k_2)$ with
$k_2>r$. Horizontally, we obtain
\begin{align*}
    &\lambda_{(K,k_2)\to (K-1,k_2)}^\textup{rect} + \frac{1}{K}\lambda_{(K,k_2)\to (K-1,k_2)}^\textup{circ} - \psi_1^\text{circ}(K,k_2)\\
    &=\frac{K}{2r} f(k_2)f(K)  + \frac{1}{K}\frac{K}{2r} f({k_2}) (r - Kf(K)) -\frac{1}{2} f(k_2) = 0
\end{align*}
and vertically, all circular flows cancel and for $k_2<K$ the rectangular flows
contribute
\begin{align*}
    \lambda_{(K,k_2)\to (K,k_2-1)}^\textup{rect} - \lambda_{(K,k_2+1)\to (K,k_2)}^\textup{rect} &=\frac{K}{2r} f(K) \sum_{l=k_2}^K f(l) - \frac{K}{2r} f(K) \sum_{l=k_2+1}^K f(l) \\
    &=\frac{K}{2r} f(K) f(k_2).
\end{align*}
For $k_2=K$ the same holds, namely,
\[ \lambda_{(K,K)\to (K,K-1)}^\textup{rect}=\frac{K}{2r} f(K) \sum_{l=K}^K f(l)
=\frac{K}{2r} f(K) f(K).\] Each such point receives a flow from $(K,r)$ and from
$(r,K)$
\[
\begin{aligned}
    &-\lambda_{(K,r)\to (K,k_2)}^\textup{diag} -\frac{r}{K}\lambda_{(r,K)\to (K,k_2)}^\textup{diag} - \psi_1^\textup{diag}(K,k_2)  \\
    &=-\frac{K}{2r}f(K)f(k_2) \left[ \frac{K-r}{K+k_2-2r}+\frac{r}{K}\frac{k_2-r}{K+k_2-2r} + \frac{K-r}{K+k_2-2r} \frac{K+k_2-r-K}{K} \right] \\
    &=-\frac{K}{2r}f(K)f(k_2) \left[ \frac{(K-r)+\frac{r}{K}(k_2-r)+\frac{K-r}{K}(k_2-r)}{K+k_2-2r}  \right] = -\frac{K}{2r}f(K)f(k_2) \\
\end{aligned}
\]
which directly cancels with the rectangular flows.

Next we consider all $a_1(r,s)$ with $r<s<K$. The corresponding vertical flows
contribute
\begin{align*}
    \frac{r}{K}\left(\lambda_{(r,s)\to (r,s-1)}^\textup{rect} - \lambda_{(r,s+1)\to (r,s)}^\textup{rect}\right) =\frac{r}{K}\frac{K}{2r} f(r)\left[ \sum_{l=s}^K f(l) - \sum_{l=s+1}^K f(l)\right] =\frac{r}{K}\frac{K}{2r} f(r) f(s).
\end{align*}
and the horizontal flows
\begin{align*}
    &- \frac{r+1}{K} \lambda_{(r+1,s)\to (r,s)}^\textup{rect}  - \frac{1}{K}\lambda_{(r,s)\to (r+1,s)}^\textup{circ} \\
    &=-\frac{K}{2r} f(s) \left[ \frac{r+1}{K}  \sum_{l=r+1}^K f(l)  + \frac{1}{K} (r - (r+1)\sum_{l=r+1}^K f(l)) \right] \\
    &=- \frac{1}{2}f(s) .
\end{align*}
Together, this is $\frac{r}{K}\frac{K}{2r} (f(r) f(s)- f(s))$. In the diagonal
flow layer $(r,s)$ serves as a source. We use \eqref{eq:outgoing_diagonal_flow},
i.e., symmetry, to get the diagonal contribution
\[\frac{K}{2r} \frac{r}{K} \left[\frac{1}{2} \sum_{k_1+k_2=r+s} f(k_1) f(k_2) -
f(r)f(s) \right].\] Now the boundary-flow contribution is all vertical, hence,
we just multiply value $\frac{r}{K}$ with \eqref{eq:boundary_flow_item2_at_r}
and switch signs to obtain
\[\begin{aligned}
&\frac{r}{K} \Big(
\lambda^{\mathrm{bdry}}_{(r,s)\to(r,s-1)}
+\lambda^{\mathrm{bdry}}_{(r,s)\to(r,s+1)}
-\lambda^{\mathrm{bdry}}_{(r,s-1)\to(r,s)}
-\lambda^{\mathrm{bdry}}_{(r,s+1)\to(r,s)}
\Big) \\
&=- \frac{r}{K} \frac{K}{2r}\left[ \frac{1}{2} \sum_{k_1+k_2 = r+s} f(k_1)f(k_2) - f(s)\right].
\end{aligned}
\]
We can now see that all the above contributions cancel out.

In the next step we consider the coefficient of $a_1(s,r)$ for $r<s<K$. The
horizontal rectangular and circular flows are the same as in
\eqref{eq:horzontal_allocation_coefficient}, thus, the horizontal flows cancel,
and the vertical contribution is
\begin{align*}
    - \frac{s}{K} \lambda_{(s,r+1)\to (s,r)}^\textup{rect} = -\frac{s}{K} \frac{K}{2r} f(s) \sum_{l=r+1}^K f(l).
\end{align*}
The diagonal flows are sent towards the lower triangle to every
$(k_1,k_2)\neq(s,r)$ with $k_1+k_2=s+r$. Hence, we use
\eqref{eq:outgoing_diagonal_flow} to obtain the total diagonal flow of
\[\frac{s}{K} \frac{K}{2r} \left[ \frac{1}{2}\sum_{k_1+k_2=s+r} f(k_1)f(k_2) - f(s)f(r) \right] .\]
The boundary flows which affect the coefficient of $a_1(s,r)$ now have different value coefficients
\[\begin{aligned}
&\frac{s}{K} \lambda^{\mathrm{bdry}}_{(s,r)\to(s-1,r)}
+\frac{s}{K} \lambda^{\mathrm{bdry}}_{(s,r)\to(s+1,r)} 
-\frac{s-1}{K} \lambda^{\mathrm{bdry}}_{(s-1,r)\to(s,r)}
-\frac{s+1}{K} \lambda^{\mathrm{bdry}}_{(s+1,r)\to(s,r)} \\
&= \frac{K}{2r}\frac{1}{2} \left[ s  \left( - B(f;s-1) \right)
+ s \left( - B(f;s+1) \right)
- (s-1) \left( - B(f;s) \right)
- (s+1) \left(- B(f;s) \right) \right] \\
&= \frac{K}{2r}\frac{s}{2} \left[ B(f;s) - B(f;s-1) - (B(f;s+1) - B(f;s)) \right]
\end{aligned}
\] Up to a constant factor we already computed this difference in
\eqref{eq:boundary_flow_item2_at_r}. Thus, we directly get the total
contribution
\[ -\frac{K}{2r}\frac{s}{K} \left[\frac{1}{2}\sum_{k_1+k_2 = r+s} f(k_1) f(k_2) - f(s) \right]. \]
Therefore, in total also these coefficients cancel.

Now we consider the coefficient of $a_1(r,K)$. There the rectangular and
circular flows contribute
\[ \begin{aligned}
&-\frac{r+1}{K} \lambda_{(r+1,K)\to (r,K)}^\textup{rect} + \frac{r}{K} \lambda_{(r,K)\to (r,K-1)}^\textup{rect} - \frac{r+1}{K} \lambda_{(r+1,K)\to (r,K)}^\textup{circ} + \frac{r}{K} \lambda_{(r,K)\to (r+1,K)}^\textup{circ} \\
&=-\frac{r+1}{K} \frac{K}{2r} f(K) \sum_{l=r+1}^K f(l) + \frac{r}{K} \frac{K}{2r} f(r) \sum_{l=K}^K f(l) - \frac{1}{K} \frac{K}{2r} f({K}) (r - (r+1)\sum_{l=r+1}^K f(l)) \\
&= \frac{r}{K} \frac{K}{2r} \left( f(r) f(K) - f({K}) \right). \end{aligned}
\] The diagonal flows from $(r,K)$ again by symmetry of the diagonal
\eqref{eq:outgoing_diagonal_flow} and multiplied by the value $\frac{r}{K}$,
are,
\[\frac{r}{K}\frac{K}{2r} \left[ \frac{1}{2} \sum_{k_1+k_2 \geq K+r} f(k_1) f(k_2) - f(K)f(r)\right].\]
Finally, the boundary flow is
\[\frac{r}{K}\left( \lambda_{(r,K)\to (r,K-1)}^\textup{bdry} -  \lambda_{(r,K-1)\to (r,K)}^\textup{bdry}\right).\]
We use \eqref{eq:B_function_difference} to obtain
\[-\frac{r}{K}\frac{K}{2r} \left[ \frac{1}{2} \sum_{k_1+k_2\geq r+K} f(k_1)f(k_2) -  f(K) \right] \]
which cancels with all other contributions.

Next we consider the coefficient of $a_1(K,r)$. Here the rectangular and
circular flows contribute
\[ \begin{aligned}
&\lambda_{(K,r)\to (K-1,r)}^\textup{rect} - \lambda_{(K,r+1)\to (K,r)}^\textup{rect} - \frac{K-1}{K} \lambda_{(K-1,r)\to (K,r)}^\textup{circ} + \lambda_{(K,r)\to (K-1,r)}^\textup{circ} - \psi_1^\textup{circ}(K,r) \\
&=\frac{K}{2r} f(K)f(r) - \frac{K}{2r} f(K) \sum_{l=r+1}^K f(l) + \frac{1}{K} \frac{K}{2r} f(r) (r - K f(K)) - \frac{K}{2r} \frac{r}{K} f(r)  \\
&=- \frac{K}{2r} f(K) \sum_{l=r+1}^K f(l). \end{aligned}  \] The diagonal
flows again rely on \eqref{eq:outgoing_diagonal_flow} just with a different
value as the multiplier in front:
\[\frac{K}{2r} \left[ \frac{1}{2} \sum_{k_1+k_2 \geq K+r} f(k_1) f(k_2) - f(K)f(r)\right].\]
Again, using \eqref{eq:B_function_difference}, the boundary flows contribute
\[\begin{aligned}
&\lambda_{(K,r)\to (K-1,r)}^\textup{bdry} - \frac{K-1}{K} \lambda_{(K-1,r)\to (K,r)}^\textup{bdry} - \psi_1^\textup{bdry}(K,r) \\
&=\frac{K}{2r} \frac{K}{2} \left(B(f;K)- B(f;K-1) \right) \\
&= - \frac{K}{2r} \left[ \frac{1}{2} \sum_{k_1+k_2\geq r+K} f(k_1)f(k_2) -  f(K) \right].
\end{aligned}\]
which again cancels exactly the coefficient.

Finally, we consider $a_1(r,r)$ where there are no diagonal flows. The
rectangular and neighbouring flows contribute
\[ \begin{aligned}
&-\frac{r+1}{K} \lambda_{(r+1,r)\to (r,r)}^\textup{rect} -\frac{r}{K} \lambda_{(r,r+1)\to (r,r)}^\textup{rect} - \frac{1}{K} \lambda_{(r+1,r)\to (r,r)}^\textup{circ} \\
&=-\frac{r+1}{K} \frac{K}{2r} f(r) \sum_{l=r+1}^K f(l) -\frac{r}{K} \frac{K}{2r} f(r) \sum_{l=r+1}^K f(l) - \frac{1}{K} \frac{K}{2r} f(r) \left( r - (r+1)\sum_{l=r+1}^K f(l) \right) \\
&= -\frac{r}{K} \frac{K}{2r} f(r) \sum_{l=r+1}^K f(l) - \frac{r}{K} \frac{K}{2r} f(r)  \\
&= -\frac{r}{K} \frac{K}{2r} f(r) \left( 1+ \sum_{l=r+1}^K f(l) \right).
\end{aligned}  \]
The boundary flow contribution is
\[\begin{aligned}
&-\frac{r+1}{K} \lambda_{(r+1,r)\to (r,r)}^\textup{bdry} + \frac{r}{K} \lambda_{(r,r)\to (r+1,r)}^\textup{bdry} -\frac{r}{K} \lambda_{(r,r+1)\to (r,r)}^\textup{bdry} + \frac{r}{K} \lambda_{(r,r)\to (r,r+1)}^\textup{bdry} + \alpha_1^\textup{bdry}(r,r) \\
&=\frac{r}{K} \frac{K}{2r} \frac{K}{2} \left( - B(f;r+1)\right) -\frac{r}{K} \frac{K}{2r} \frac{K}{2} \left( - 2B(f;r)\right) + \frac{r}{K} \frac{K}{2r} \frac{K}{2} \left(- B(f;r+1)\right) \\
&=  \frac{K}{2} \left( B(f;r) - B(f;r+1)\right)  \\
&=  \frac{1}{2} \left( (1-f(r)f(r)) - 2(1-f(r)) \right)  \\
&=  - \frac{1}{2} \left( 1 - f(r) \right)^2.
\end{aligned}\]

The last remaining term is the $\mu(r,r)$ multiplier. Its contribution to
$a_1(r,r)$ is 
\[\frac{r}{K}\mu(r,r)=\frac{r}{K}\frac{K}{2r}=\frac{1}{2}.\]
So altogether the coefficient is
\[-\frac{1}{2}f(r)(2-f(r)) - \frac{1}{2} (1-f(r))^2 + \frac{1}{2} = 0.\]
All these calculations hold symmetrically as well for item two.

Until now, we have proven that all coefficients of the primal variables cancel.
Finally, we have to quantify the remaining Lagrangian. The only remaining
constants are just the sum of all $\psi_1$ and $\psi_2$ variables.

\paragraph{Quantify the objective function} For this we compute the sum of all $\psi^\textup{circ}$, all $\psi^\textup{diag}$, and all $\psi^\textup{bdry}$ separately. Since the considered distribution is anchored in $r$, we obtain
\[
\sum_{k_2=r}^K \psi_1^{\mathrm{circ}}(K,k_2)=\frac{K}{2r}\frac {r}{K}\sum_{k_2=r}^K f(k_2)=\frac{1}{2}.
\]
For both items this sums up to one.

The diagonal terms appear only in the upper triangle. There at any $(k_1,k_2)$ with $k_1+k_2>K+r$ we have
\[\begin{aligned}
&\psi_1^{\mathrm{diag}}(k_1,k_2)+\psi_2^{\mathrm{diag}}(k_1,k_2) \\
&=\frac{K}{2r}f(k_1)f(k_2)\frac{k_1-r}{k_1+k_2-2r}\frac{k_1+k_2-r-K}{K}+\frac{K}{2r}f(k_1)f(k_2)\frac{k_2-r}{k_1+k_2-2r}\frac{k_1+k_2-r-K}{K} \\
&=\frac{K}{2r}f(k_1)f(k_2)\frac{k_1-r+k_2-r}{k_1+k_2-2r}\frac{k_1+k_2-r-K}{K}\\
&=\frac{1}{2r}f(k_1)f(k_2)(k_1+k_2-r-K).
\end{aligned}\]
Hence, the total diagonal constant is
\[ \frac{1}{2r}\sum_{k_1+k_2>K+r}\left(k_1+k_2-r-K\right)f(k_1)f(k_2) \] which
is exactly $\frac{K}{2r} B(f;K)$ as for $s=K$ the second sum in the definition
of $B(f;K)$ \eqref{def:boundary_correction_term} is empty.

The last contribution is from the boundary multipliers:
\[\psi_1^{\mathrm{bdry}}(K,r)=\frac{K}{2r}\frac{1}{2}\left(
\frac{2r}{K}\omega_K-B(f;K)\right),\]
and symmetrically the same for $\psi_2^{\mathrm{bdry}}(r,K)$. Hence, together they contribute
\[\frac{K}{2r}\left(\frac{2r}{K}\omega_K-B(f;K)\right)=
\omega_K-\frac{K}{2r}B(f;K).
\]

Hence, for the fixed anchored distribution $f$, the Lagrangian equals
$1+\omega_K$ for any choice of primal variables. This concludes the proof.

\subsection{Proof of \texorpdfstring{\cref{prop:two_items_iid_SREV}}{Proposition 5.7}}
\label{append:prop:two_items_iid_SREV}

Without loss,\footnote{On a support containing the zero type, \eqref{eq:ir-def}
gives $p(\vec{0})\leq 0$ and \eqref{eq:dsic-def} $p(\vec{v}) \geq p(\vec{0})$. A
unified lifting of all payments $p^\prime(\vec{v}) = p(\vec{v}) - p(\vec{0})$
preserves \eqref{eq:ir-def} and leaves \eqref{eq:dsic-def} unchanged.
Furthermore, it satisfies $p^\prime(\vec{0}) = 0$ and has nonnegative payments
everywhere.} we impose the constraint set that $p(k_1,k_2)\geq0$ associated with
the Lagrangian multiplier $\nu(k_1,k_2)$.

For the proof we fix $K$ at the start and induct over the number of
\textit{positive probability support points}
\[N \coloneq |\{k : f(k)>0\}|.\] We omit the trivial case where the distribution
has a single point mass in zero as there the approximation ratio is just set to
one.

\paragraph{Base case: $N=1$, a point mass distribution} If there is a single
point with probability mass of $1$, i.e., $f(r)=1$ for some index $r$, this
distribution is anchored at $r$. Thus, there exists an upper bound
$(1+\omega_K)$ certificate.

\paragraph{Induction step: $N-1$ to $N$}

Now assume the claim for all marginals with fewer than $N$ positive support
points, i.e., for any (general) distribution with at most $N-1$ positive
probability support points there exists an upper bound $(1+\omega_K)$
certificate. 

We distinguish between two cases. First we consider the case, that the smallest
positive probability support point is positive, or in other words,
$\prob{X_i=0}=0$, and afterwards if $\prob{X_i=0}>0$. \paragraph{Let
$\prob{X_i=0}=0$.} Let $F$ have $N$ support points, and let $r > 0$ be its
smallest support index. Since this is not the base case, $0 < f(r) < 1$. 

We define the family of distributions which keeps the shape of the distribution
strictly above index $r$ scaling this shape by the probability of being above
$r$. Let $H$ be this conditional distribution and $t\in[0,1]$ the probability of
being above $r$. The family is
\[F_t = (1-t) \delta_r + t H,\] and it is a convex combination of two
distributions (a point mass and the fixed shape distribution). The distribution
$F$ is one member of this family, namely for the case $t=1-f(r)$. By the
induction hypothesis, $H$ has a $(1+\omega_K)$ certificate.

Let the optimal revenue of selling a single item under $H$ be defined as
\[R_H \coloneq \max_{r<k\leq K} k \sum_{l=k}^K h(l).\] For any $F_t$, selling at
price $\frac{r}{K}$ gives revenue $\frac{r}{K}$, while the best price strictly
above $\frac{r}{K}$ yields $t\frac{R_H}{K}$. If $F$ is anchored at $r$ we know
that there exists a $(1+\omega_K)$ certificate, hence, we assume the contrary,
i.e., $tR_H > r$. We define the threshold and the distribution
\[t_r \coloneq \frac{r}{R_H}, \qquad F_{t_r} = (1-t_r) \delta_r + t_r H.\] Then
$0 < t_r < t$ and $F_{t_r}$ is anchored at $r$ as selling at $\frac{r}{K}$
yields $\frac{r}{K}$ and the best revenue under $H$ is $\frac{R_H}{K}$, thus,
under $t_r H$ is $\frac{t_r R_H}{K}$ which by definition is
$\frac{R_H}{K}\frac{r}{R_H}$. Hence, $F_{t_r}$ has a $(1+\omega_K)$ certificate.
Moreover, with $\beta = \frac{1-t}{1-t_r}$ we write $F$ as the convex
combination
\[F = \beta F_{t_r} + (1-\beta) H.\]

Now we construct a $(1+\omega_K)$ certificate for $F$ by convexly combining the
dual multipliers of the $(1+\omega_K)$ certificates of $F_{t_r}$ and $H$. We do
this for every multiplier except for the $\theta$ vector:
\[\lambda^F \coloneq \beta \lambda^{F_r} + (1-\beta) \lambda^H, \qquad
\text{analogously for} \; \mu, \psi_1, \psi_2, \alpha_1, \alpha_2, \nu .\] The
new variables remain nonnegative. We define $\theta^F_q = \frac{K}{2 t R_H}$ for
an index $q>r$ which corresponds to a value of an optimal posted price and
yields revenue $R_H/K$. Set other components of $\theta^F$ to zero. By this
choice of dual variables as a convex combination the allocation variables'
coefficients cancel again. 

The coefficients of the payment variables are more delicate: By the defined
$\theta^F_q$ the contribution to the coefficient of $p(k_1,k_2)$ from the
respective constraints is 
\[\frac{K}{2 t R_H} f(k_1) f(k_2).\] The convexly combined certificate provides
coefficients of $p(k_1,k_2)$ of
\[ \beta \frac{K}{2 t_r R_H} f_r(k_1) f_r(k_2) + (1-\beta) \frac{K}{2 R_H} h(k_1) h(k_2) .\]
We make use of the definition of $\beta$ as well as $f(k) = t h(k)$ and $f_r(k) = t_r h(k)$ and obtain
\begin{align*}
    \beta \frac{K}{2 t_r R_H} f_r(k_1) f_r(k_2) + (1-\beta)& \frac{K}{2 R_H} h(k_1) h(k_2)\\  
    &=  \frac{1-t}{1-t_r} \frac{K}{2 t_r R_H} f_r(k_1) f_r(k_2) + \frac{t-t_r}{1-t_r} \frac{K}{2 R_H} h(k_1) h(k_2)  \\
    &=  \frac{K}{2 t R_H} \left[ \frac{1-t}{1-t_r} \frac{t}{t_r} t_r^2 h(k_1) h(k_2) + \frac{t-t_r}{1-t_r} t h(k_1) h(k_2) \right] \\
    &=  \frac{K}{2 t R_H} \left[ \frac{1-t}{1-t_r} \frac{t}{t_r} \frac{t_r^2}{t^2} f(k_1) f(k_2) + \frac{t-t_r}{1-t_r} \frac{1}{t} f(k_1) f(k_2) \right] \\
    &=  \frac{K}{2 t R_H} \left[ \frac{1-t}{1-t_r} \frac{t_r}{t} + \frac{t-t_r}{1-t_r} \frac{1}{t} \right] f(k_1) f(k_2)  \\
    &=  \frac{K}{2 t R_H}  \frac{t_r-t t_r+t-t_r}{t-t t_r}  f(k_1) f(k_2)  \\
    &=  \frac{K}{2 t R_H} f(k_1) f(k_2).
\end{align*}
Hence, the payments vanish for $k_1,k_2>r$. For $k=r$ we cannot use $f(k) = t
h(k)$ and $f_r(k) = t_r h(k)$ as $h(r) = 0$. There we can simply use $f(r) =
1-t$, $f_r(r) = 1 - t_r$ and $h(r) = 0.$ Thus, we obtain for $k_1>r$ and $k_2=r$
\begin{align*}
    \beta \frac{K}{2 t_r R_H} f_r(k_1) f_r(r) + (1-\beta) \frac{K}{2 R_H} h(i) h(r)  &=  \frac{1-t}{1-t_r} \frac{K}{2 t_r R_H} f_r(k_1) f_r(r)\\
    &=  \frac{1-t}{1-t_r} \frac{K}{2 t_r R_H} t_r h(k_1) (1-t_r)\\
    &=   \frac{K}{2 R_H} (1-t) h(k_1) \\
    &=   \frac{K}{2 t R_H} (1-t) f(k_1) \\
    &= \frac{K}{2 t R_H} f(r) f(k_1)
\end{align*}
which analogously holds for $k_1=r, k_2>r.$ Finally, we consider $k_1=k_2=r$ and
obtain
\begin{align*}
    \beta \frac{K}{2 t_r R_H} f_r(r) f_r(r) + (1-\beta) \frac{K}{2 R_H} h(r) h(r)  &=  \frac{1-t}{1-t_r} \frac{K}{2 t_r R_H} f_r(r) f_r(r)\\
    &=  \frac{1-t}{1-t_r} \frac{K}{2 t_r R_H} (1-t_r)^2\\
    &=  \frac{K}{2 t_r R_H} (1-t)(1-t_r)\\
    &=  \frac{K}{2 t R_H} \frac{t}{t_r}(1-t)(1-t_r)
\end{align*}
but the certificate for $F$ requires $\frac{K}{2 t R_H} (1-t)^2.$ The difference
to the targeted certificate of the coefficient is 
\begin{align*}
    \frac{K}{2 t R_H} \frac{t}{t_r}(1-t)(1-t_r) - \frac{K}{2 t R_H} (1-t)^2 &= 
    \frac{K}{2 R_H} (\frac{(1-t)(1-t_r)}{t_r} - \frac{(1-t)^2}{t} ) \\
    &= \frac{(1-t) K}{2 R_H} (\frac{1-t_r}{t_r} - \frac{1-t}{t} ) \\
    &= \frac{(1-t) K}{2 R_H} \frac{t-t_r}{t_rt} >0
\end{align*}
where the last sign conclusion comes from $t_r<t$ by definition of $t_r$. Hence,
for equality we only need the coefficient to be zero which can be obtained by
increasing $\nu(r,r)$ by the exact difference $\frac{(1-t) K}{2 R_H}
\frac{t-t_r}{t_rt}$. The multiplier remains non-negative and we obtain a
$(1+\omega_K)$ certificate for $F$.

Now the payment variables' coefficients vanish again. The constant remaining
part of the Lagrangian solely follows from the convex combination of the two
$(1+\omega_K)$ certificates and evaluates in the Lagrangian as:
\[\beta (1+\omega_K) + (1-\beta) (1+\omega_K) = (1+\omega_K).\]

\paragraph{Let $\prob{X_i=0}>0$.} Again let $F$ have $N$ support points, and let
$r = 0$ be its smallest support index with positive probability. We define again
$t \coloneq 1-f(0)$ and write $F$ as the convex combination of the point mass
distribution and the conditional distribution of being above:
\[F = (1-t)\delta_0 + tH\] where $t h(k) = f(k)$ for $k>0$. Note, that this time
we do not make use of the whole family $F_t$ since there is no $t$ such that
$F_t$ is anchored in $r=0$. The only possibility would be $t=0$, i.e.\ a point
mass in $0$, but this would be the base case.

By the induction hypothesis, $H$ has a $(1+\omega_K)$ certificate. We again
define $\frac{R_H}{K}$ to be the optimal revenue, selling a single item under
$H$, i.e., 
\[R_H \coloneq \max_{r<k\leq K} k \sum_{l=k}^K h(l)\] and in addition let $s$ be
the smallest positive probability support index of $H$. Again set $\theta^F_q =
\frac{K}{2tR_H}$ for an index $q$ maximizing $R_H$ and all other entries of
$\theta^F$ equal to zero. For the certificate of $F$ we take the certificate of
$H$ (scaled by $t$) for the region $k_1,k_2>0$ and define new dual multipliers
for the boundary regions:

For $k_1=s+1, \dots, K$ and $k_2=0$ set
\[\lambda_{(k_1,0)\rightarrow (k_1-1,0)} = \frac{K}{2R_H} \left[R_H - (k_1-1)\sum_{l=k_1}^K h(l)\right],\]
\[\lambda_{(k_1-1,0)\rightarrow (k_1,0)} = \frac{K}{2R_H} \left[R_H -
k_1\sum_{l=k_1}^K h(l)\right].\] They can be seen as the combination of
rectangular and circular flow layers. Also set
\[\mu(s,0) = \frac{K}{2R_H},\]
\[\alpha_1(s,0) = \frac{1}{2R_H} \left[R_H - s\sum_{l=s}^K h(l)\right],\]
\[\psi_1(K,0) = \frac{1}{2}.\] Symmetrically define the same for $k_2=s+1,
\dots, K$ and $k_1=0$. The quantities are nonnegative because $R_H \geq k
\sum_{l=k}^K h(l)$ for every $k$. The net flow contributes $\frac{K}{2R_H} h(k)$
to the payment coefficients at every $(k,0)$ as well as at $(0,k)$. The
allocation coefficients telescope completely and the two $\psi$ constants sum up
to $1$. Finally, we just set
\[\alpha_1(0,0)=\psi_1(0,0) = \omega_K\] which contributes a constant of
$\omega_K$ to the objective. (Note, that this is optional, as otherwise we would
just obtain an upper bound of $1$ for the constructed axes.)

Now we define the certificate of $F$ as the convex combination of $(1-t)$ times
the axis variables we just defined and $t$ times the certificate of $H$. We
again go through all cases: For $k_1,k_2>0$ the axis variables are zero so we
just have to consider the $tH$ certificate
\[   t \frac{K}{2R_H} h(k_1) h(k_2) = t \frac{K}{2R_H} \frac{f(k_1)}{t} \frac{f(k_2)}{t} = \frac{K}{2tR_H} f(k_1) f(k_2).\]
For $k_1>0$ and $k_2=0$ we have $1-t$ times the defined axis net flow which is
\[     (1-t) \frac{K}{2R_H} h(k_1) = (1-t) \frac{K}{2R_H} \frac{f(k_1)}{t} =
\frac{K}{2tR_H} (1-t)f(k_1) = \frac{K}{2tR_H} f(0) f(k_1) . \] The case $k_1=0$
and $k_2>0$ is identical. At the origin we want to have $\frac{K}{2tR_H}
(1-t)^2$. As the dual multiplier of the utility at the origin does not affect
the coefficients of the allocations, we can simply set
\[\mu(0,0) = \frac{K}{2tR_H} (1-t)^2 \] which cancels the payment coefficient in
the origin. Now the certificate is complete, feasible and the remaining constant
part is $(1-t)(1+\omega_K) + t(1+\omega_K) = (1+\omega_K)$. The induction is
complete; the statement holds for any distribution on $I_K$.

\section*{AI Methodology}
\label{sec:ai-methodology}

We used OpenAI's GPT-6 Astra, GPT-6.1 Sol, and GPT-5.6 Sol, to assist with
mathematical exploration, checking proof arguments, computational searches for
candidate bounds, and designing numerical experiments.

This assistance materially influenced our proofs of the three technical claims
in \cref{appendix:proof_of_claims} and the numerical aspects of our lower-bound
proof in \cref{sec:M-lower-bound}. In these cases we used AI-assistance to
fine-tune the numerical parameters of the proofs, in an effort to strike a good
balance between the strength of the bounds and the accessibility of the
presentation. As a matter of fact, in all these proofs, the AI-agents were
indicating alternative analyses resulting to (slightly) improved final bounds;
however, our personal judgment was that these approaches were significantly more
obscure and numerically convoluted, without providing any new insights or
techniques, and thus we chose to omit them. However, the reader should be aware
of that, as a strong indication that our bounds resulting from these specific
proofs should be considered as improvable for future work.

The authors critically evaluated the suggestions provided by AI and verified all
mathematical arguments; we take full responsibility for the manuscript's
correctness, originality, and attribution. The paper's research ideas, main
methodological approaches (including the duality formulation
in~\cref{sec:srev-ratio-formulation} and the capped master inequality framework
of~\cref{eq:rev_master_inequality}) and the technical flow and extremal
distribution constructions in the analyses
of~\cref{sec:lagrangian-dual-many-items-weak},
\cref{sec:selling_separately_2_iid}, and \cref{append:sec:iid-results} are our
own.

\printbibliography
\end{document}